\documentclass[11pt]{article}
\usepackage[utf8]{inputenc}
\usepackage{mystyle}

\begin{document}

\title{\huge Covariate Assisted Entity Ranking with Sparse Intrinsic Scores \thanks{Emails: \texttt{jqfan@princeton.edu}, \texttt{jikaih@princeton.edu}, and \texttt{mengxiny@wharton.upenn.edu}. Research is supported by NSF grants DMS-2210833 and DMS-2053832, and ONR grant N00014-22-1-2340}
}

\author{Jianqing Fan\qquad Jikai Hou \qquad Mengxin Yu}


\maketitle




\begin{abstract}
This paper addresses the item ranking problem with associate covariates, focusing on scenarios where the preference scores can not be fully explained by covariates, and the remaining intrinsic scores, are sparse. Specifically, we extend the pioneering Bradley-Terry-Luce (BTL) model by incorporating covariate information and considering sparse individual intrinsic scores. Our work introduces novel model identification conditions and examines the statistical rates of the regularized penalized Maximum Likelihood Estimator (MLE). We then construct a debiased estimator for the penalized MLE and analyze its distributional properties. Additionally, we apply our method to the goodness-of-fit test for models with no latent intrinsic scores, namely, the covariates fully explaining the preference scores of individual items.  We also offer confidence intervals for ranks.
Our numerical studies lend further support of  our theoretical findings, demonstrating validation for our proposed method.

\end{abstract}


\section{Introduction}

Ranking plays an essential role across a wide scope of domains. Specifically, it holds particular significance in many real-world applications, including individual choice \citep{luce2005individual}, ranking web pages \citep{dwork2001rank}, recommendation systems \citep{baltrunas2010group,li2019estimating}, education \citep{caron2014bayesian}, sports ranking \citep{massey1997statistical,turner2012bradley}, scientific journals ranking \citep{stigler1994citation}, elections \citep{plackett1975analysis}, assortment optimization \citep{talluri2004revenue,rusmevichientong2010dynamic} and even instruction tuning used in recently popular artificial intelligence product ChatGPT \citep{ouyang2022training}. 

Luce \citep{luce2012individual} introduced the renowned \emph{Axiom of Choice}, which plays a pivotal role in the field of decision theory. According to this axiom, when comparing two items, denoted as $i$ and $j$, within a set of alternatives $A$ that contains $\{i,j\}$, the probability of selecting $i$ over $j$ remains constant regardless of the presence of other alternatives in the set, i.e.,
$$
{\frac{\PP(i \text{ is preferred in } A)}{\PP(j \text{ is preferred in } A)} = \frac{\PP(i \text{ is preferred in } \{i,j\})}{\PP(j \text{ is preferred in } \{i,j\})}\,}.
$$
Two well-known parametric choice models stem from this axiom of choice: the Bradley-Terry-Luce (BTL) model \citep{bradley1952rank,luce2012individual}, designed for pairwise comparisons, and the Plackett-Luce (PL) model \citep{plackett1975analysis}, tailored for $M$-way rankings where $M\ge 2$. Specifically, the BTL model assumes a collection of $n$ items whose true ranking is determined by some unobserved preference scores $\theta_{i}^*$ for $i=1,\cdots,n$.  In this scenario, an individual ranks item $i$ over item $j$ has probability { $\PP(\text{item $i$ is preferred over item $j$}) = e^{\theta_i^*} / (e^{\theta_i^*} + e^{\theta_j^*})$}.  In addition, the Plackett-Luce model is an expanded version of pairwise comparison, which allows for a more comprehensive $M$-way full ranking \citep{plackett1975analysis}. These models provide valuable insights and tools for analyzing and modeling decision-making processes in various domains of study.

It is worth noting that in both the Bradley-Terry-Luce (BTL) model and the Plackett-Luce (PL) model, it is assumed that the latent scores attributed to the items of interest are fixed and do not use the characteristics of these items. Nonetheless, in numerous practical scenarios, such as university rankings and sports competitions, the outcomes always 
depend on the covariate information of items being ranked, and it becomes crucial to incorporate this heterogeneity into the modeling framework.

Some pioneering works study the ranking problem with covariates.
For example, \cite{turner2012bradley,li2022bayesian} and \cite{fan2022uncertainty} study ranking with covariates by incorporating feature information of items into the BTL model. Specifically, they assume the underlying score (ability) of the $i$-th item is given by $\alpha_i^*+\bx_i^\top \bbeta^*$ where $\bx_i^\top \bbeta^*$ captures the covariate effect and $\alpha_i^*$ is the intrinsic score that cannot be explained by the covariate. This basically assumes all involved latent scores are intrinsic scores plus attributes explained by covariates.
In this case, the outcome of pairwise comparison is modeled as  
 $$
 	\mathbb{P}(\textrm{item i is preferred over j}) = \frac{e^{\alpha_i^*+\bx_i^\top\bbeta^*} }{ e^{\alpha_i^*+\bx_i^\top\bbeta^*} + e^{\alpha_j^*+\bx_j^\top\bbeta^*}}.
$$ 
On the other hand, there are another line of research \citep{guo2018experimental,schafer2018dyad,zhao2022learning,chau2023spectral,finch2022introduction} that considers the ranking estimation by assuming the underlying score is expressed as $\theta^*=\xb_i^\top \beta^*.$ In other words, they assume the underlying scores of all compared items are fully explained by covariates (i.e., $\alpha_i^* = 0$ for all $i$). 


The aforementioned formulations exhibit both advantages and drawbacks. Adopting a model where all $n$ items are assumed to possess intrinsic scores $\balpha\in \RR^n$ results in a more comprehensive framework but will also introduce additional noise when the inherent contributions of these intrinsic scores are sparse. Empirical investigations based on real-world data, as explored in the study by \cite{fan2022uncertainty} on portfolio selection (version 1 on Arxiv) and pokemon competitions, suggest that employing a model with sparse intrinsic scores can lead to improved predictive performance. However, directly assuming that all latent scores of items are entirely explained by the observed covariates ($\alpha_i^* = 0$ for all $i$) results in a strong assumption and will easily lead to model mis-specification.

In response to the aforementioned challenges, and inspired by empirical observations in the study by \cite{fan2022uncertainty}, this paper examines the entity ranking with covariates that exhibit sparse intrinsic scores. Specifically, we consider a scenario where a total of $n$ items are subject to comparison, and we assume that the latent score associated with the $i$-th item is represented as $\alpha_i^*+\xb_i^\top \beta^*.$ However, unlike the setting explored in \cite{fan2022uncertainty}, we assume that the vector $\balpha^*=[\alpha_1^*,\cdots,\alpha_n^*]$ is sparse. In other words, the majority of the item scores are explained by their respective covariates, while some items with size $k=o(n)$ have non-vanishing intrinsic scores. It's worth noting that this model can accommodate scenarios in which no intrinsic scores are assumed as a special case.

For other specifications of the model, we adhere to the conditions established in previous works \citep{chen2015spectral,chen2019spectral,fan2022uncertainty}, where they study the statistical properties of this model within the context of pairwise comparisons. Likewise, we also do not make the assumption that all pairs undergo direct comparisons. Specifically, we adopt the Erdős-Rényi random graph as the underlying comparison graph, and each pair is selected independently for comparison with a probability $p$. Once a pair is selected for comparison, they undergo the comparison process a total of $L$ times. In this study, we employ a fixed design matrix $\bX$, where randomness only comes from the randomness of the graph generation and outcomes of the comparisons.


In light of the novel sparsity assumption on the intrinsic scores of the items, we establish a new model identification condition. We employ a carefully designed $\ell_1$-penalized likelihood and introduce a proximal gradient descent method to facilitate the estimation of the regularized Maximum Likelihood Estimate (MLE). Additionally, we analyze both $\ell_{\infty}$- and $\ell_2$-statistical errors of the MLE, and achieve the optimal sample complexity in terms of the parameters $n$, $p$, and $L$. We further provide a comprehensive examination of the distributional properties of the MLE.

Challenges arise due to the bias introduced by the $\ell_1$-penalty while studying the distributional results of the MLE. To address this concern, we propose a debiased estimator, particularly tailored for the loss derived from pairwise comparisons. We derive a non-asymptotic expansion of the debiased estimator, effectively controlling the approximation error while maintaining optimal sample complexity. Notably, this research marks the first systematic exploration of ranking with covariates characterized by sparse latent attributes.

To further illustrate the applicability of our method, we begin by examining the goodness-of-fit test for a null model that the covariates explain fully the preference (i.e., $\balpha^*=0$) is considered.
We study the hypothesis testing on $H_0:\|\balpha^*\|_{\infty}=0$ and use Gaussian multiplier bootstrapping to approximate the limiting distribution of the proposed statisitcs. In addition, we also study the out-of-sample ranking inferences. Suppose we obtain covariates $\zb_i,i\in [n]$ that represent the features in the future stage, but the comparisons have not been made. We utilize the statistics $ \widehat{\theta}_i = \widehat{\alpha}^{\debias}_{R, i}+\bz^\top_i\widehat{\bb}_{R}$ to build out-of-sample confidence intervals for the future latent scores   $\theta_i^*=\alpha_i^*+\zb_i^\top\beta^*,$ where $\widehat{\alpha}^{\debias}_{R, i}$ is the debiased estimator of $\alpha_i$ and $\widehat{\bb}_{R}$ is the estimator for $\bbeta^*.$ To reduce the length of the prediction intervals and enhance prediction power, we also consider using a two-stage approach. In this case, we re-estimate the parameters based on the selected subsets in the first stage and subsequently derive their asymptotic distributions. Our comprehensive numerical experiments provide empirical evidence that aligns with our theoretical results and demonstrates the efficacy of our proposed method. 

To summarize, the contributions of this work are several folds. We introduce a novel approach to the study of ranking with covariates, specifically focusing on scenarios where the intrinsic scores of compared items exhibit sparsity. Additionally, we develop an $\ell_1$-regularized loss function, investigate the statistical properties of the penalized Maximum Likelihood Estimate (MLE), and establish optimal statistical rates for our model. Furthermore, we present a debiased estimator and analyze its asymptotic distribution. Expanding upon this, we conduct goodness-of-fit testing for well-established ranking models that do not consider intrinsic scores. We also provide a method for constructing out-of-sample confidence intervals for predicted future unknown scores. Our empirical studies validate the robustness of our proposed theories and methods.

\subsection{Related Works}
Pairwise ranking problems have gained significant attention in many fields. In the case of the Bradley-Terry-Luce (BTL) model, extensive research efforts have been made to study its various aspects. For instance, \cite{chen2015spectral} employed a two-step approach to analyze the BTL model, demonstrating its optimality in terms of sample complexity. Meanwhile, \cite{negahban2012iterative} introduced an iterative rank aggregation algorithm called Rank Centrality, achieving optimal $\ell_2$-statistical rates for recovering the underlying scores of the BTL model. Building upon this, \cite{chen2019spectral} further extended their analysis to derive both $\ell_2$- and $\ell_{\infty}$-optimal statistical rates for these underlying scores. They established that the regularized Maximum Likelihood Estimate (MLE) and spectral methods are optimal for recovering top-K items when the condition number remains constant. \cite{chen2022partial} further demonstrated that for partial recovery, the MLE remains optimal, but the spectral method is less optimal in terms of the general conditional number.

The models and methods discussed in this section so far primarily focus on studying the statistical estimation problems in ranking models, neglecting the incorporation of individual feature information. However, in many real-world applications, covariate data is readily available and plays an important role, introducing additional complexities in both technical derivations and computations. There are a series of works that study ranking with covariates without considering the intrinsic scores of compared items; see \cite{guo2018experimental,schafer2018dyad,zhao2022learning,chau2023spectral,finch2022introduction} for more details. Recently, there have been some other works \citep{turner2012bradley,li2022bayesian,fan2022uncertainty} that study the statistical property of ranking with covariates and unconstrained personal intrinsic scores. A related work is \cite{fan2022uncertainty}, which systematically studied the distributional result of the MLE of the covariate-assisted ranking model with non-sparse intrinsic scores. Our work bridges these two lines by considering covariate-assisted ranking with sparse intrinsic scores.

The aforementioned existing body of literature has predominantly focused on achieving non-asymptotic statistical consistency when estimating item scores within ranking models. It is also essential to investigate the limiting distributions of ranking models. Recently, a few studies have explored the asymptotic distributions of estimated ranking scores, particularly within the Bradley-Terry-Luce (BTL) model framework, where comparison graphs are sampled from Erdős–Rényi graphs with a connection probability denoted as $p$, and each observed pair undergoes the same number of comparisons denoted as $L$.
\cite{simons1999asymptotics,han2020asymptotic} established the asymptotic normality of the maximum likelihood estimator (MLE) for the BTL model when all comparison pairs are fully observed (p = 1) or under dense comparison graph \(p \gtrsim n^{-1/10}\), respectively. More recently, \cite{liu2022lagrangian} introduced a Lagrangian debiasing approach to derive asymptotic distributions for ranking scores under sparse graph regime \(p \asymp \log n/n\) and studied many ranking related applications. Additionally, \cite{gao2023uncertainty} used a "leave-two-out" technique to study the asymptotic distributions for ranking scores, improving the theoretical results of \cite{liu2022lagrangian} by achieving optimal sample complexity (allowing \(L=O(1)\)) in sparse comparison graph settings (\(p \asymp 1/n\) up to logarithmic terms). In the sequel, \cite{fan2022uncertainty} further extends this line by incorporating covariate information into the BTL model. Through an innovative proof technique, they presented the asymptotic distribution of the MLE with optimal sample complexity under sparse comparison graphs. There is also some other literature that broadens the aforementioned analysis to multiway comparisons, we refer interested readers to \cite{fan2022ranking,han2023unified,fan2023spectral} for more details.

\subsection{Roadmap}

In Section \ref{sec:prob_setup}, we provide a comprehensive problem formulation for our Bradley-Terry-Luce (BTL) model, considering both covariate information and sparse attributes. Within the same section, we also establish the statistical rates of the Maximum Likelihood Estimator (MLE) for the associated loss function. Section \ref{sec:distribution} delves into uncertainty quantification, particularly for the debiased variant of the MLE, further enhancing our understanding of the statistical properties. In Section \ref{sec:extensions}, we extend our proposed methodology to encompass the assessment of goodness-of-fit for models that do not consider individual intrisic scores. Additionally, we explore the construction of confidence intervals for predicting future latent scores.

\subsection{Notation}
We introduce some useful notations used in this paper before proceeding. We denote by $[M] = \left\{1,2,\dots, M\right\}$ for any positive integer $M$. For any vector $\mathbf{u}$ and $q\ge 0$, we use $\|\mathbf{u}\|_{\ell_q}$ to represent the vector $\ell_q$ norm of $\mathbf{u}$. In addition, the inner product $\langle \mathbf{u}, \mathbf{v} \rangle $ between any pair of vectors $\mathbf{u}$ and $\mathbf{v}$ is defined as the Euclidean inner product $\mathbf{u}^\top\mathbf{v}$. For vector $\mathbf{u}\in \mathbb{R}^m$ and index $i\in [m]$, we denote by $\mathbf{u}_{-i}$ the vector we get by deleting the $i$-th element in $\mathbf{u}$. For any given matrix $\mathbf{X}\in\mathbb{R}^{d_1\times d_2}$, we use $\|\mathbf{X}\|$, $\|\mathbf{X}\|_{F}$, $\|\mathbf{X}\|_{*}$ and $\|\mathbf{X}\|_{2,\infty}$ to represent the operator norm, Frobenius norm, nuclear norm and two-to-infinity norm of matrix $\mathbf{X}$ respectively. Moreover, we use  $\mathbf{X}\succcurlyeq 0$ or $\mathbf{X}\preccurlyeq 0$ to denote positive semidefinite or negative semidefinite
of matrix $\mathbf{X}$.
Moreover, we use the notation $a_n\lesssim b_n$ or $a_n = O(b_n)$ for non-negative sequences $\left\{a_n\right\}$ and $\left\{b_n\right\}$ if there exists a constant $\nu_1$ such that $a_n\leq \nu_1 b_n$. We use the notation $a_n\gtrsim b_n$ for non-negative sequences $\left\{a_n\right\}$ and $\left\{b_n\right\}$ if there is a constant $\nu_2$ such that $a_n\geq \nu_2 b_n$. 
We write $a_n\asymp b_n$ if $a_n\lesssim b_n$ and $b_n\lesssim a_n$.

\section{Problem Setup and Estimation Results}\label{sec:prob_setup}
In this section, we outline our problem setup and establish estimation results.
Given $n$ items with individual features $\bx_i\in\mathbb{R}^d$ for $i\in [n]$, the probability of item $j$ is preferred over item $i$ is modeled as
\begin{align}
    \mathbb{P}\{\text {item } j \text { is preferred over item } i\}=\frac{e^{\alpha_j^*+\bx_j^\top \bb^*}}{e^{\alpha_i^*+\bx_i^\top \bb^*}+e^{\alpha_j^*+\bx_j^\top \bb^*}} ,\quad \forall 1\leq i\neq j \leq n. \label{comparison}
\end{align}
Here $\alpha_i^*$ is an intrinsic score for item $i$, while the linear term $\bx_i^\top\bb^*$ captures the part of the scores explained by the variables $\bx_i$. Let $\tx_i = \left(\boldsymbol{e}_i^\top, \bx_i^\top\right)^\top\in \mathbb{R}^{n+d}$ and $\tb = \left(\ba^\top,\bb^\top\right)^\top\in \mathbb{R}^{n+d}$, where $\{\boldsymbol{e}_i\}_{i=1}^n$ stand for the canonical basis vectors in $\mathbb{R}^n$ and $\ba = (\alpha_1,\alpha_2,\dots,\alpha_n)^\top\in\mathbb{R}^n$. We make the following assumption on $\bx_i$.
\begin{assumption}\label{nondegenerate}
Let $\bar{\bx}_i=[1,\bx_i], \forall i\in[n]$ and $\bar{\bX}=[\bar{\bx}_1,\cdots,\bar{\bx}_n]^\top\in \RR^{n\times (d+1)}$. We assume that the dimension $d < n$  and $\bar{\bX}$ is non-degenerate.
\end{assumption}

We next impose a sparsity constraint on the intrinsic scores $\balpha=\tb_{1:n}$. Given a positive integer $k$, we consider the following parameter space:
\begin{align}\label{parameter_space}
    \Theta(k) = \left\{\tb\in \mathbb{R}^{n+d}:\left\|\tb_{1:n}\right\|_0\leq k\right\},
\end{align}
and we assume the true parameter vector $\tb^*\in \Theta(k)$. As the first step, the following proposition ensures $\Theta(k)$ is identifiable. 
\begin{proposition}\label{proposition1}
As long as $2k+d+1\leq n$, model \eqref{comparison} with parameter space $\Theta(k)$ is identifiable under Assumption \ref{nondegenerate}.
\end{proposition}
\begin{proof}
	See \S\ref{proposition1proof} for a detailed proof.
\end{proof}
Throughout the paper, we assume that $\tb^* = (\alpha_1^*,\dots,\alpha_n^*,\bb^*)\in \Theta(k)$ for some $k$ such that $2k+d+1\leq n$.

As the second part of the our model, we do assume that all pairs in the  the comparison graph $\mathcal{G}=(\mathcal{V}, \mathcal{E})$ are compared. Here  $\mathcal{V} := \{1,2,\dots, n\}$ and $\mathcal{E}$ represent the collections of vertexes ($n$ items) and edges, respectively.  More specifically, $(i,j)\in \mathcal{E}$ if and only if item $i$ and item $j$ are compared. Throughout our paper,  the comparison graph is assumed to follow the Erdős-Rényi random graph $\mathcal{G}_{n,p}$ where each edge appears independently with probability $p$ (i.e., items $i$ and $j$ with $(i,j)\in[n]\times[n]$ are compared at random with probability $p$). 

In addition, for any $(i,j) \in \mathcal{E}$,  we observe $L$  independent and identically distributed realizations from the Bernoulli random variables
\begin{align*}
   P( y_{i, j}^{(l)} = 1)  = \frac{e^{\alpha_j^*+\bx_j^\top \bb^*}}{e^{\alpha_i^*+\bx_i^\top \bb^*}+e^{\alpha_j^*+\bx_j^\top \bb^*}}.
\end{align*}
Denote by $ y_{i,j} = \frac{1}{L}\sum_{l=1}^L y_{i,j}^{(l)}$, a sufficient statistic.

With these settings in hand, we consider the following loss function, which is the negative log-likelihood conditioned on comparison graph $\mathcal{G}$ and scaled by $1/L$
\begin{align}
	\mathcal{L}(\tb): 
	&=\sum_{(i, j) \in \mathcal{E}, i>j}\left\{-y_{j, i}\left(\tx_i^\top\tb-\tx_j^\top\tb\right)+\log \left(1+e^{\tx_i^\top\tb-\tx_j^\top\tb}\right)\right\}.\label{nll}
\end{align}
In the following contents, we rescale $\bx_i$ to $\bx_i / K$, where $K>0$ is a positive number such that $\| \bx_i\|_2 \leq \sqrt{(d+1)/n}$ for all $\bx_i$ after the transformation.
The likelihood function, prediction and the column space spanned by $\bar{\bX}$ are not affected by the scaling. However, this normalization facilitates scaling issues in the technical derivations.  

Furthermore, we consider the following regularized estimator
\begin{align}
    \tb_R = \argmin_{\tb\in \mathbb{R}^{n+d}}\cL_R(\tb),\label{regularizedMLE}
\end{align}
where $\cL_R(\cdot)$ is defined as
\begin{align*}
    \cL_R(\tb):=\cL(\tb)+\lambda \left\|\ba\right\|_1+ \frac{\tau}{2}\left\|\tb\right\|_2^2.
\end{align*}
In this context, we introduce the parameters $\lambda>0$ and $\tau>0$ as regularization coefficients aimed at ensuring both a sparse solution and the strong convexity of the loss function, respectively. To derive the results for $\tb_R$, we make the following two key assumption on the covariates.
\begin{assumption}\label{assump0}[Incoherence Condition]
We assume that there exists a positive constant $c_0$ such that 
\begin{align*}
   \|\bar{\bX}(\bar{\bX}^\top\bar{\bX})^{-1}\bar{\bX}^\top\|_{2,\infty}\le c_0\sqrt{(d+1)/n}.
\end{align*}
\end{assumption}
To explain the rationale behind Assumption \ref{assump0}, we begin by observing that $\|\bar{\bX}(\bar{\bX}^\top\bar{\bX})^{-1}\bar{\bX}^\top\|^2_{F}\le d+1$. Consequently, a condition sufficient for the validity of this assumption is that the rows of the projection $\cP_{\bar{\bX}}:=\bar{\bX}(\bar{\bX}^\top\bar{\bX})^{-1}\bar{\bX}^\top$ exhibit nearly balanced characteristics, with the sum of squares of row elements all on the order of $(d+1)/n$ or smaller. Note that when there is an absence of covariates (i.e., $\bar{\bX}=\mathbf{1}$), we have $\cP_{\bar{\bX}}=\mathbf{1}\mathbf{1}^\top/n.$ Under this scenario, the assumption automatically holds with $c_0=1.$

Next, we impose the same assumption on $\boldsymbol{\Sigma} = \sum_{i>j}(\tx_i-\tx_j)(\tx_i-\tx_j)^\top$ as \cite{fan2022uncertainty}, which guarantees that the loss function will behave well and that the Maximum Likelihood Estimator (MLE) will have good statistical properties. 
\begin{assumption}\label{assump1}
Consider $\boldsymbol{\Sigma} := \sum_{i>j}(\tx_i-\tx_j)(\tx_i-\tx_j)^\top$. Assume that there exists positive constants $c_1$ and $c_2$ such that 
\begin{align*}
    c_2n\leq \lambda_{\text{min}, \perp}(\boldsymbol{\Sigma})\leq \Vert \boldsymbol{\Sigma}\Vert\leq c_1n,
\end{align*}
where $\Vert \boldsymbol{\Sigma}\Vert$ is the operator norm of $\boldsymbol{\Sigma}$ and  
\begin{align*}
    \lambda_{\text{min}, \perp}(\boldsymbol{\Sigma}):=\min \left\{\mu:  \tb^\top \boldsymbol{\Sigma}\tb\geq \mu\Vert \tb \Vert_2^2\text{ for all }\tb\in \RR^{n+d} \text{ s.t. }\bar{\bX}^\top \tb_{1:n} = \boldsymbol{0}_{d+1}\right\}.
\end{align*}
\end{assumption}
In Assumption \ref{assump1}, we assume that $\displaystyle\boldsymbol{\Sigma}$ exhibits favorable characteristics of being positive definite in directions orthogonal to the columns of $\bar{\bX}$. This assumption aligns with the corresponding assumption presented in \cite{fan2022uncertainty}. We note that the upper bound presented in Assumption \ref{assump1} is implicitly satisfied based on the rescaled $\bx_i$ (such that $\| \bx_i
\|_2 \leq \sqrt{(d+1)/n}$) when concatenated with the vector $\be_i$. 
When no covariates $\bX$ are included in the model, Assumption \ref{assump1} simplifies to the condition that $\mathbf{\Sigma}=\sum_{i>j}(\mathbf{e}_i-\mathbf{e}_j)(\mathbf{e}_i-\mathbf{e}_j)^\top$ is positive definite within the subspace defined by $\mathbf{1}^\top \bx=0$. This condition is inherently satisfied by its original definition \citep{chen2019spectral,chen2022partial}.

Consequently, for any $(k+d)$-sparse parameter $\tb\in \Theta(k)$, after projecting it onto the space $\bar{\bX}^\top\tb_{1:n}=\mathbf{0}_{d+1}$, we are able to ensure the non-degeneracy of $\mathbf{\Sigma}$ on our parameter space of interest, defined in \eqref{parameter_space}.

We now present the theoretical guarantees of the estimator presented in \eqref{regularizedMLE} on its statistical rate of convergence. Prior to unveiling the results, we introduce three quantities of conditional numbers, depicting the difficulty associated with the recovery of $\tb^*$.
\begin{align*}
	\kappa_1 := \exp\left( \max_{i,j\in [n]}\left(\alpha_i^*+\bx_i^\top \bb^*-\alpha_j^*-\bx_j^\top \bb^* \right)\right), \quad \kappa_2 := \max_{i\in [n]}|\alpha_i^*|, \quad \kappa_3 := \frac{\left\Vert \tb^*\right\Vert_2}{\sqrt{n}}
\end{align*}
For a vector $\bx$, we use $\cS(\bx)$ to represent its support. 
\begin{theorem}\label{estimationthm}
\em
Suppose $k(d+1)< n, d\log n\lesssim np$. We consider $L \leq  c_4\cdot n^{c_5}$ for any absolute constants $c_4,c_5>0$ and
\begin{align}
    \lambda= c_{\lambda}\kappa_1\sqrt{\frac{(d+1)np\log n}{L}},\quad \tau = c_{\tau}\min\left\{\frac{\kappa_1}{\kappa_2},\frac{1}{\kappa_3\sqrt{d+1}}\right\}\sqrt{\frac{\log n}{nL}}\label{lambdaandtau}
\end{align}
for some constants $c_\tau, c_\lambda>0$. Let $\tb_R = (\widehat\ba_R^\top,\widehat\bb_R^\top)^\top$ be the solution of the regularized MLE Eq.~\eqref{regularizedMLE}. Then with probability at least $1-O(n^{-6})$, we have
\begin{align*}
\Vert \widehat\ba_R-\ba^*\Vert_\infty\lesssim \kappa_1^2\sqrt{\frac{(d+1)\log n}{npL}}, \quad
\left\Vert \tb_R-\tb^*\right\Vert_2 \lesssim\kappa_1\sqrt{\frac{\log n}{pL}}, \quad \text{and}\quad  \cS(\widehat\ba_R)\subset\cS(\ba^*).
\end{align*}
 If the signal strength satisfies
\begin{align*}
    \min_{i\in \cS(\ba^*)} |\ba_i^*|\gg \kappa_1^2\sqrt{\frac{(d+1)\log n}{npL}},
\end{align*}
the support of $\ba^*$ is exactly recovered, i.e., $\cS(\widehat\ba_R)=\cS(\ba^*)$.

 If we further have $npL\geq C \kappa_1^6(k+d)^3(d+1)\log n$, $np\geq C\kappa_1^2(k+d)$ and $n\geq C\kappa_1^2(k+d)(d+1)$ for some constant $C>0$, then with probability exceeding $1-O(n^{-6})$, it holds that
    \begin{align*}
        \left\|\tb_R-\tb^*\right\|_2\lesssim \kappa_1^2\sqrt{\frac{(k+d)(d+1)\log n}{npL}}.
    \end{align*}
\end{theorem}
In Theorem \ref{estimationthm}, we present the $\ell_2$-statistical error of $\tb_R$ to $\tb^*$ as well as the $\ell_{\infty}$-error of intrinsic scores $\ba_R$ to $\ba^*.$ When the sample complexity $npL$ is sufficiently large, these statistical rates are optimal in terms of $n,p,L$ from the perspective of the information-theoretic principle \citep{chen2019spectral,chen2022partial,chen2022optimal,fan2022uncertainty}. Compared to \cite{fan2022uncertainty}, we improve an order of $\sqrt{n}$ in the statistical error of $\|\tb_R-\tb^*\|_2$ due to the sparsity assumption on $\balpha^*$ and its exploration in the estimation.

We next address the selection of tuning parameters, denoted as $\lambda$ and $\tau$, within the loss function $\cL_{R}(\tb)$. To ensure that the Karush-Kuhn-Tucker (KKT) condition of the optimal solution $\tb^*$ is satisfied, we choose a value of $\lambda$ on the same order of magnitude as $\|\nabla \cL(\tb^*)\|_{\infty}$. In addition, the primary purpose in introducing the $\ell_2$-regularizer in loss $\cL_{R}(\tb)$ with tuning parameter $\tau$ is to guarantee the strong convexity of $\cL_{R}(\tb)$. However, we do not intend to introduce any additional bias through this $\ell_2$-regularization term. In fact, $\tau$ can be chosen as any non-negative real number such that $\tau \leq c_{\tau}\min\left\{\kappa_1 / \kappa_2, 1 / \kappa_3\sqrt{d+1}\right\}\sqrt{\log n /nL}$ for some fixed constant $c_{\tau}>0$.


\section{Debiased Estimator and Distributional Results} \label{sec:distribution}
This section presents the distributional results pertaining to the regularized maximum likelihood estimator  (MLE) $\tb_{R}$ in \eqref{regularizedMLE}. Note that the inclusion of the $\ell_1$- and $\ell_2$-regularization term within the loss function $\mathcal{L}_{R}(\cdot)$ introduces additional bias into the estimator. Therefore, as an initial step, we present detailed procedures for mitigating this bias in the MLE.

For a given index $i\in [n]$, we introduce the following univariate function: 
\begin{align*}
    \cL_{R, \tb_{R, -i}}(x) = \cL_R(\tb)\bigg|_{\tb_i=x, \tb_{-i}=\tb_{R, -i}}.
\end{align*}
Since $\tb_R$ is the minimizer of $\cL_R(\cdot)$, it holds that $\widehat{\alpha}_{R,i}$ (the $i$-th entry of $\hat\alpha_{R}\in \RR^n$) is the minimizer of $\cL_{R, \tb_{R, -i}}(x)$. As a result, we have
\begin{align}
    0=\cL_{R, \tb_{R, -i}}'(\widehat{\alpha}_{R,i}) = \cL_{\tb_{R, -i}}'(\widehat{\alpha}_{R,i}) +\tau \widehat{\alpha}_{R,i}+\lambda\partial |\widehat{\alpha}_{R,i}|,\label{uqeq1}
\end{align}
where $\cL_{\tb_{R, -i}}(x)$ is defined similarly to $\cL_{R, \tb_{R, -i}}(x)$ and
$\partial |\cdot|$ is a subgradient of the absolute value function. By the mean value theorem, there exists a real number $b$ between $\alpha_i^*$
 and $\widehat{\alpha}_{R,i},$ such that 
 \begin{align}
      \cL_{\tb_{R, -i}}'(\widehat{\alpha}_{R,i}) = \cL_{\tb_{R, -i}}'(\alpha_i^*) + \cL_{\tb_{R, -i}}''(b)(\widehat{\alpha}_{R,i} - \alpha_i^*).\label{uqeq2}
 \end{align}

After combining \eqref{uqeq1} and \eqref{uqeq2} together, it holds that
\begin{align*}
    0
    &\approx \cL_{\tb_{R, -i}}'(\alpha_i^*) + \left(\nabla^2\cL(\tb_R)\right)_{i,i}(\widehat{\alpha}_{R,i} - \alpha_i^*) +\tau \widehat{\alpha}_{R,i}+\lambda\partial |\widehat{\alpha}_{R,i}|.
\end{align*}
After re-organizing the terms, we have
\begin{align}\label{expansion_alpha}
    \widehat{\alpha}_{R,i} +\frac{\tau \widehat{\alpha}_{R,i}+\lambda\partial |\widehat{\alpha}_{R,i}|}{\left(\nabla^2\cL(\tb_R)\right)_{i,i}} \approx \alpha_i^*-\frac{\cL_{\tb_{R, -i}}'(\alpha_i^*)}{\left(\nabla^2\cL(\tb_R)\right)_{i,i}}.
\end{align}
Note that the right-hand side is asymptotically unbiased.  This leads us to 
define the debiased estimator of $\widehat{\alpha}_{R, i}$ as the left-hand side of \eqref{expansion_alpha}, while the subgradient can be found by the optimality condition in \eqref{uqeq1}:
\begin{align}
    \widehat{\alpha}_{R,i}^{\debias} := \widehat{\alpha}_{R,i}+\frac{\tau \widehat{\alpha}_{R,i}+\lambda\partial |\widehat{\alpha}_{R,i}|}{\left(\nabla^2\cL(\tb_R)\right)_{i,i}} = \widehat{\alpha}_{R,i}-\frac{\left (\nabla\cL(\tb_R)\right)_i}{\left(\nabla^2\cL(\tb_R)\right)_{i,i}},\label{debiasedestimator}
\end{align}
where we used $\cL_{\tb_{R, -i}}'(\widehat{\alpha}_{R,i}) = \left (\nabla\cL(\tb_R)\right)_i$
To assess the uncertainty associated with $\hat\bb_R$, given that we do not impose sparsity regularization on it and the $\ell_2$-regularization parameter $\tau$ is relatively small, there is no need to perform debiasing on $\widehat{\bb}_R$ to obtain the distributional result. 

The following Theorem \ref{uqmainthm} establishes the distributional results for $\widehat{\alpha}_{R,i}^{\debias}$ and $\widehat{\beta}_{R,k}$.
\begin{theorem}\label{uqmainthm}
Suppose $npL\geq C \kappa_1^6(k+d)^3(d+1)\log n$, $np\geq C\kappa_1^2(k+d)$ and $n\geq C\kappa_1^2(k+d)(d+1)$ for some constant $C>0$. Given any $i\in [n]$ and $k\in [d]$, with probability at least $1-O(n^{-6})$ we have
\begin{align*}
    & \sup_{x\in \mathbb{R}}\left|\mathbb{P}\left(\sqrt{\left(\nabla^2\cL(\tb^*)\right)_{i,i}L}(\widehat{\alpha}_{R,i}^{\debias}-\alpha_i^*)\leq x\right) - \mathbb{P}(\mathcal{N}(0,1)\leq x)\right| \\
    \lesssim &\frac{\kappa_1^3(d+1)}{\sqrt{np}}\left(\frac{\kappa_1^3\log n}{\sqrt{L}}+\sqrt{(k+d)\log n}\right), \\
    &\sup_{x\in \mathbb{R}}\left|\mathbb{P}\left(\frac{\sqrt{L}\left(\hat{\beta}_{R, k}-\beta_k^*\right)}{\sqrt{(\bA^{-1})_{k,k}}}\leq x\right) - \mathbb{P}(\mathcal{N}(0,1)\leq x)\right|\\
    \lesssim & \frac{\kappa_1^3(d+1)\sqrt{kd(k+d)\log n}}{\sqrt{np}} + \frac{\kappa_1^{4.5}((k+d)(d+1)\log n)^{3/4}}{(npL)^{1/4}},
\end{align*}    
where $\bA := (\nabla^2\cL(\tb^*))_{n+1:n+d,n+1:n+d}$.
\end{theorem}

The proof of Theorem \ref{uqmainthm} is deferred to  \S\ref{uqmainthmproof}. Next, we comment on a two-stage method when the signal of $\ba^*$ is sufficiently strong, leading to the recovery of true support $\cS(\widehat{\ba}_R) = \cS(\ba^*)$.
Specifically, under this assumption, the problem becomes a low-dimensional problem.  We refit the model to get the two-stage estimator via the negative log-likelihood function $\tilde \cL:\mathbb{R}^{|\cS(\ba^*)|+d}\to \mathbb{R}$ defined as 
\begin{align*}
    \tilde \cL(\bgamma) = \cL(\tb)\bigg|_{\tb_{[n]\backslash\cS(\ba^*) }= \boldsymbol{0}, \tb_{([n]\backslash\cS(\ba^*))^c} = \bgamma},\quad  \forall \bgamma\in \mathbb{R}^{|\cS(\ba^*)|+d}.
\end{align*}
We let 
\begin{align*}
    \widehat{\bgamma} = \argmin_{ \bgamma\in \mathbb{R}^{|\cS(\ba^*)|+d}} \tilde \cL(\bgamma),
\end{align*}
be the two-stage (or re-fitted) estimator, and we establish the distributional results for $\widehat{\bgamma}$, presented in the following Theorem \ref{thm:dist_two}.

\begin{theorem}\label{thm:dist_two}
Given $\cS(\widehat{\ba}_R) = \cS(\ba^*)$ and the aforementioned two-stage estimator $\widehat{\bgamma}$, as long as $npL\gtrsim \kappa_1^2(k+d)\log n$, for any convex set $\mathcal{D}\subset \mathbb{R}^{|\cS(\ba^*)|+d}$, we have
\begin{align*}
    \left|\mathbb{P}(\widehat{\bgamma}-\bgamma^* \in \mathcal{D}) - \mathbb{P}(\mathcal{N}(\boldsymbol{0}, (\nabla^2 \tilde \cL (\bgamma^*))^{-1}) \in \mathcal{D})\right|\lesssim \kappa_1^3 \frac{(k+d)^{5/4}\log n}{\sqrt{npL}}+\frac{1}{n^{10}}.
\end{align*}    
\end{theorem}
The proof of Theorem \ref{thm:dist_two} is deferred to \S\ref{proof_twostage}.

\section{Applications}\label{sec:extensions}
In this section, we study two practical applications of our distributional results. First, we perform a goodness-of-fit test of a special case of our model \eqref{comparison}, in order to substantiate the necessity of introducing the sparsity-inducing intrinsic scores $\alpha_i, i \in [n]$. Second, we establish out-of-sample rank confidence intervals utilizing future covariates $\{\bz_1, \cdots, \bz_n\}$ as a demonstration of the predictive capability inherent in our model.

\subsection{Goodness-of-Fit Test}\label{gofsection}
In this section, we test whether the covariates can fully capture the preference scores of all items.  Specifically, we are interested in the following hypothesis testing problem:
\begin{align*}
     H_0:\ba^* = \boldsymbol{0} \quad \text{ v.s. }\quad H_a:  \ba^* \neq \boldsymbol{0}.
\end{align*}
Therefore, it is natural to consider the following test statistic:
\begin{align*}
    \cT_1 = \max_{i\in [n]}\left|\sqrt{\left(\nabla^2\cL(\tb_R)\right)_{i,i}L}\widehat{\alpha}_{R,i}^{\debias}\right|.
\end{align*}

Leveraging Theorem \ref{estimationthm} and the linear expansions of $\widehat{\alpha}_{R,i}, i \in [n]$, we deduce that
\begin{align*}
    \left|\cT_1 - \max_{i\in [n]}\left|\sqrt{L}\left(\nabla\cL(\tb^*)\right)_i/\sqrt{\left(\nabla^2\cL(\tb_R)\right)_{i,i}}\right|\right| = o_p(1).
\end{align*}
According to the definition of $\cL(\cdot),$ the linear approximation of $\cT_1$ can be derived as independent sums, as presented below
\begin{align*}  \frac{\sqrt{L}\left(\nabla\cL(\tb^*)\right)_i}{\sqrt{\left(\nabla^2\cL(\tb_R)\right)_{i,i}}} = \sum_{l=1}^L\sum_{(i,j)\in\mathcal{E}}\sqrt{\frac{1}{\left(\nabla^2\cL(\tb_R)\right)_{i,i}L}}\left(\phi(\tx_i^\top\tb^*-\tx_j^\top\tb^*)-y_{j, i}^{(l)}\right),
\end{align*}
where $\phi(t) = e^t /(1+e^t)$.
We employ the Gaussian multiplier bootstrap technique with the Gaussian approximation theory outlined in \cite{chernozhuokov2022improved} to derive the asymptotic distribution of $\cT_1$.

Specifically, let $\omega_{j, i}^{(l)}$, $1\leq j<i\leq n, 1\leq l\leq L$ be i.i.d $\mathcal{N}(0,1)$ random variables, we define the Gaussian multiplier bootstrap counterpart of $\cT_1$ as 
\begin{align}
    \cG_1 = \max_{i\in [n]}\left|\sum_{l=1}^L\sum_{(i,j)\in\mathcal{E}}\sqrt{\frac{1}{\left(\nabla^2\cL(\tb_R)\right)_{i,i}L}}\left(\phi(\tx_i^\top\tb_R-\tx_j^\top\tb_R)-y_{j, i}^{(l)}\right)\omega_{j, i}^{(l)}\right|. \label{Bootstrap1}
\end{align}
Given any $\alpha\in (0,1)$, let $c_{1, 1-\alpha}$ be the $(1-\alpha)$-th quantile of $\cG_1$ conditioned on $\mathcal{E}$ and $\{y_{j, i}:1\leq j<i\leq n\}$, which yields
\begin{align*}
    c_{1,1-\alpha} = \inf\{z\in \mathbb{R}:P(\cG_1\leq z|\mathcal{E}, \{y_{j,i}\})\geq 1-\alpha\}.
\end{align*}
Then we have the following theorem for the test statistics $\cT_1$.
\begin{theorem}\label{T1thm}
    Under the conditions of Theorem \ref{uqmainthm}, we have
    \begin{align*}
        \left|P(\cT_1>c_{1, 1-\alpha})-\alpha\right| \lesssim \left(\frac{\log^5 n}{np}\right)^{1/4}  + \frac{\kappa_1^3(d+1)\log n}{\sqrt{np}}\left(\kappa_1^3\sqrt{\frac{\log n}{L}}+\sqrt{k+d}\right).
    \end{align*}
\end{theorem}

We defer the proof of Theorem \ref{T1thm} to \S\ref{T1thmproof}.

\subsection{Out-of-Sample Ranking Inferences}\label{oosRCIsection}
In this section, we turn to constructing both two-sided confidence intervals for out-of-sample ranks based on the information from observed covariates. 

Recall that, in our model, we divide the ranking score into two parts: the part of the intrinsic score $\alpha_i^*$ and the part $\tx_i^\top \tb^*$ explained by the covariates. Therefore, when a new set of covariates $\{\bz_1,\bz_2,\dots,\bz_n\}$ are observed, the out-of-sample unknown ranking scores are given by $\tilde{\theta}_i^* := \alpha_i^*+\bz_i^\top \bb^*.$ Note that these $\bz_i$'s can be the same as $\bx_i$'s, as both are non-random. Let $\tilde r_m$ be the rank of $\tilde \theta_m$ among $\{\tilde \theta_i\}_{i=1}^n$.  

Let $\mathcal{M}$ be the set of items among $n$ items of interest. We aim to construct the $(1 - \alpha) \times 100\%$ confidence interval for the out-of-sample population rank $\tilde{r}_{m},m\in \cM$ simultaneously, where $\alpha \in (0, 1)$ denotes a pre-specified significance level. We deduce this problem to a simultaneous pairwise comparison problem as follows.

Let ${[\mathcal{C}_{L}(k, m), \mathcal{C}_{U}(k, m)]},{k \neq m}, m\in \cM, k \in[n]$ represent the simultaneous confidence intervals of the pairwise differences ${\tilde{\theta}_{k}^{*} - \tilde{\theta}_{m}^{*}},{k \neq m}, m\in \cM$, with the following property:
\begin{align}
\label{eq_SCIs_theta}
\PP\Big({\mathcal{C}_{L}(k, m) \leq
\tilde{\theta}_{k}^{*} - \tilde{\theta}_{m}^{*} \leq \mathcal{C}_{U}(k, m), \forall k\neq m}, \forall m\in \cM\Big) \geq 1 - \alpha.
\end{align}
One observes that if $\mathcal{C}_{U}(k, m) < 0$ (respectively, $\mathcal{C}_{L}(k, m) > 0$), it implies that $\tilde{\theta}_{k}^{*} < \tilde{\theta}_{m}^{*}$ (respectively, $\tilde{\theta}_{k}^{*} > \tilde{\theta}_{m}^{*}$). Enumerating the number of items whose scores are higher than item $m,$ for each $m\in \cM,$ we obtain the lower bounds of the confidence intervals for rank $r_m$, $m\in \cM,$ and vice versa. 
In other words, we deduce from \eqref{eq_SCIs_theta} that
\begin{align}
\label{eq_confidence_interval_m_probability}
    \PP\left(1 + \sum_{k \neq m} 1\{\cC_{L}(k, m) > 0\} \leq \tilde{r}_{m} \leq n - \sum_{k \neq m} 1\{\cC_{U}(k, m) < 0\},\forall m\in \cM\right) \geq 1 - \alpha. 
\end{align}
This yields a $(1 - \alpha) \times 100\%$ two-sided confidence interval for $\tilde{r}_{m}$.

To this end, we next construct simultaneous confidence intervals for the pairwise differences ${\tilde{\theta}_{k}^{*} - \tilde{\theta}_{m}^{*}},{k \neq m}, m\in \cM$. This motivate us to consider
\begin{align*}
    \cT_2 := \max_{m\in \cM}\max_{k\neq m}\left|\frac{\widehat{\theta}_k - \widehat{\theta}_m - (\tilde{\theta}_k^* - \tilde{\theta}_m^*)}{\widehat{\sigma}_{m,k}}\right|,
\end{align*}
where  with $\tz_i:=\left(\boldsymbol{e}_i^\top, \bz_i^\top\right)^\top$, 
\begin{align}
    \widehat{\sigma}_{m,k}^2 = (\tz_m-\tz_k)^\top (\nabla^2\cL(\tb_R))^\diamond \nabla^2\cL(\tb_R) (\nabla^2\cL(\tb_R))^\diamond    (\tz_m-\tz_k) / L. \label{defi:widehatsigmamk}
\end{align}
Similar to \eqref{Bootstrap1}, we consider the bootstrap counterparts of $\cT_2$ as 
\begin{align*}
    \cG_2 :=& \max_{m\in \cM}\max_{k\neq m}\Bigg|\sum_{l=1}^L \sum_{(i,j)\in \cE,i>j} \frac{(\tz_m-\tz_k)^\top(\nabla^2\cL(\tb_R))^\diamond (\tx_i-\tx_j)}{\widehat{\sigma}_{m,k}L}(\phi(\tx_i^\top\tb_R-\tx_j^\top\tb_R)-y_{j, i}^{(l)})\omega_{j, i}^{(l)} \Bigg| 
\end{align*}
Let $c_{2,1-\alpha}$ be the $(1-\alpha)$-th quantile of $\cG_2$, we have the following theorem for the test statistics $\cT_2$.
\begin{theorem}\label{T23thm}
Assume $k\geq 2$. Under the conditions of Theorem \ref{uqmainthm}, as long as $n\gtrsim (d+1)^2 k$, we have
    \begin{align*}
        &\left|P(\cT_2>c_{2,1-\alpha})-\alpha\right|\lesssim \left(\frac{\kappa_1^3\log^5 n}{np}\right)^{1/4} + \frac{\kappa_1^{3.5}(d+1)\log n}{\sqrt{np}}\left(\kappa_1^3\sqrt{\frac{\log n}{L}}+\kappa_1^4\sqrt{\frac{\log n}{L(d+1)}}+\sqrt{k+d}\right).
    \end{align*}
\end{theorem}
\begin{proof}
    See \S \ref{T23thmproof} for a detailed proof.
\end{proof}
\begin{remark}
    The above results can be adapted to two-stage estimator. To apply the two-stage estimator, we replace $(\widehat{\alpha}^{\debias}_{R, 1}, \dots,\widehat{\alpha}^{\debias}_{R, n}, \widehat{\bb}_R)$ with $\tb$ such that $\tb_{[n]\backslash\cS(\ab^*)} = \boldsymbol{0}$ and $\tb_{([n]\backslash\cS(\ab^*))^c} = \widehat{\bgamma}$. Also, the definition of $\tz_i$ is changed to be
    \begin{align*}
         \tz_i =
    \begin{cases}
      \left(\boldsymbol{e}_i^\top, \bz_i^\top\right)^\top\quad \quad \quad & i\in \cS(\ba^*), \\
      \left(\boldsymbol{0}_n^\top, \bz_i^\top\right)^\top \quad \quad \quad& i\notin \cS(\ba^*).
    \end{cases}    
    \end{align*}
    Then Theorem \ref{T23thm} also holds for the two-stage estimator under the same assumption of Theorem \ref{thm:dist_two}.  
\end{remark}
We next present an application of Theorem \ref{T23thm} regarding two-sided out-of-sample ranking confidence intervals in Example \ref{RCI}, and we also discuss the applications to top-K candidate selection and screening in \S\ref{sec:add_app}.
\begin{example}\label{RCI}
Let $\cM$ be the set of items of interest. Given $\alpha\in (0,1)$, let $\widehat{c}_{2,1-\alpha}$ be the estimated $(1-\alpha)$-th quantile of $\cT_2$ from the bootstrap samples. According to Theorem \ref{T23thm}, we can construct simultaneous confidence intervals for $\tilde{\theta}_k^* - \tilde{\theta}_m^*, k\neq m,m\in \cM,$ as
    \begin{align}
        [\cC_L(k,m),\cC_U(k,m)] = [\widehat{\theta}_k-\widehat{\theta}_m\pm \widehat{c}_{2,1-\alpha}\widehat{\sigma}_{m,k}],\label{RCIeq1}
    \end{align}
    where $\widehat{\sigma}_{m,k}$ is defined in \eqref{defi:widehatsigmamk}. Since 
    \begin{align*}
        \PP\left(\tilde{\theta}_k^* - \tilde{\theta}_m^*\in [\cC_L(k,m),\cC_U(k,m)], \forall k\neq m,\forall m\in \cM\right)\geq 1-\alpha,
    \end{align*}
    we know that
    \begin{align*}
        \PP\left(\tilde r_m \in \left[1+\sum_{k\neq m}\textbf{1}(\cC_L(k,m)> 0), n - \sum_{k\neq m}\textbf{1}(\cC_U(k,m)< 0)\right], \forall m\in \cM\right)\geq 1-\alpha.
    \end{align*}
    In this way, we construct a $(1-\alpha)\times 100\%$ confidence interval $[\cR_L(m),\cR_U(m)]$ for all $r(m),m\in \cM$.
\end{example}
We next discuss the advantages and drawbacks of applying the two-stage method to construct confidence intervals for ranks. When the signal strength of $\balpha^*$ is strong, employing this method for ranking inference leads to narrow confidence intervals for items whose scores are fully explained by the covariates (i.e., $\alpha^*_i=0$). On the other hand, however, in cases where some signals of $\balpha^*$ are relatively weak, choosing a larger value of the tuning parameter $\lambda$ may result in false negatives in the estimation stage.  For practical reasons, we recommend initially employing a relatively small value for the tuning parameter $c_{\lambda}$ in equation \eqref{lambdaandtau} with the aim of reducing the dimensionality. Additionally, we advocate the integration of element-wise distributional results of $\balpha$ to facilitate the selection of an appropriate support $\mathcal{S}(\hat{\alpha}_R)$ for the second stage.

In the forthcoming numerical experiments, we will demonstrate that the confidence intervals we construct for true ranks, using both one-stage and two-stage methods, exhibit stability and match the predefined level of $1-\alpha.$

\section{Numerical Experiments}
This section is dedicated to illustrating the efficacy of the proposed methodology and validating its theoretical underpinnings through numerical studies. Specifically, we will validate the asymptotic normality of the de-biased estimator and study the applications discussed in \S\ref{sec:extensions}. 
\subsection{Asymptotic Normality}
In this subsection, we validate the asymptotic normality of the debiased estimator, present in Theorem \ref{uqmainthm}. 
We let the number of compared items be $n = 200$, covariate dimension be $d = 3$ and the sparsity level of $\balpha^*\in \RR^n$ (size of $|\cS(\ba^*)|$) be $k = 5$.


We generate $\alpha_i^*$ by sampling $|\alpha_i^*|\sim \text{Uniform}[0.3, 0.3\times \log(5)]$ with a random sign. 
For $\bb^*$,  it is generated uniformly from the hypersphere $\displaystyle\{\bb:\Vert \bb\Vert_2 = 0.5\sqrt{{n}/{(d+1)}}\}$. In addition, entries of the covariate matrix $\bX = [\bx_1,\bx_2,\dots,\bx_n]^\top\in \RR^{n\times d}$ are sampled independently from $\text{Uniform}[-0.5, 0.5]$, and are normalized to have mean $0$ and scaled with $ \max_{i\in [n]}\Vert\bx_i\Vert_2 = \sqrt{(d+1)/{n}}$. We choose $(p, L)$ from $\{(0.5, 25), (0.1, 10)\}$ and adjust $\lambda$ correspondingly. When $(p, L) = (0.5, 25)$, we choose $\lambda = 3$ and $1$. When $(p, L) = (0.1, 10)$, we choose $\lambda = 1.2$ and $0.4$. This results in $4$ combinations. For each setting, we generate the comparison graph $\mathcal{E}$ and data $\{y_{i,j}^{(\ell)},l\in[L],(i,j)\in\mathcal{E}\}$ $500$ times and record $\widehat{\alpha}_{R,1}^{\debias}$ and $\widehat{\beta}_{R,1}$. In Figures \ref{normalityhistogram} and {\red 2}, we report the histograms and Q-Q plots of the following two normalized random variables:
\begin{align*}
    RV_1 = \sqrt{\left(\nabla^2\cL(\tb_R)\right)_{1,1}L}(\widehat{\alpha}_{R,1}^{\debias}-\alpha_1^*), \quad RV_2 = \frac{\sqrt{L}\left(\hat{\beta}_{R, 1}-\beta_1^*\right)}{\sqrt{(\bA^{-1})_{1,1}}},
\end{align*}
respectively, where $\boldsymbol{A}$ is defined in Theorem \ref{uqmainthm}. 

We conclude from Figures \ref{normalityhistogram} and {\red 2} that the empirical distributions of $RV_1$ and $RV_2$ follow closely the standard Gaussian distribution. These results validate our theoretical guarantee of normal approximation in Theorem \ref{uqmainthm} with the right asymptotic variance.

\begin{figure}[h]
\begin{center}
\includegraphics[width=16cm]{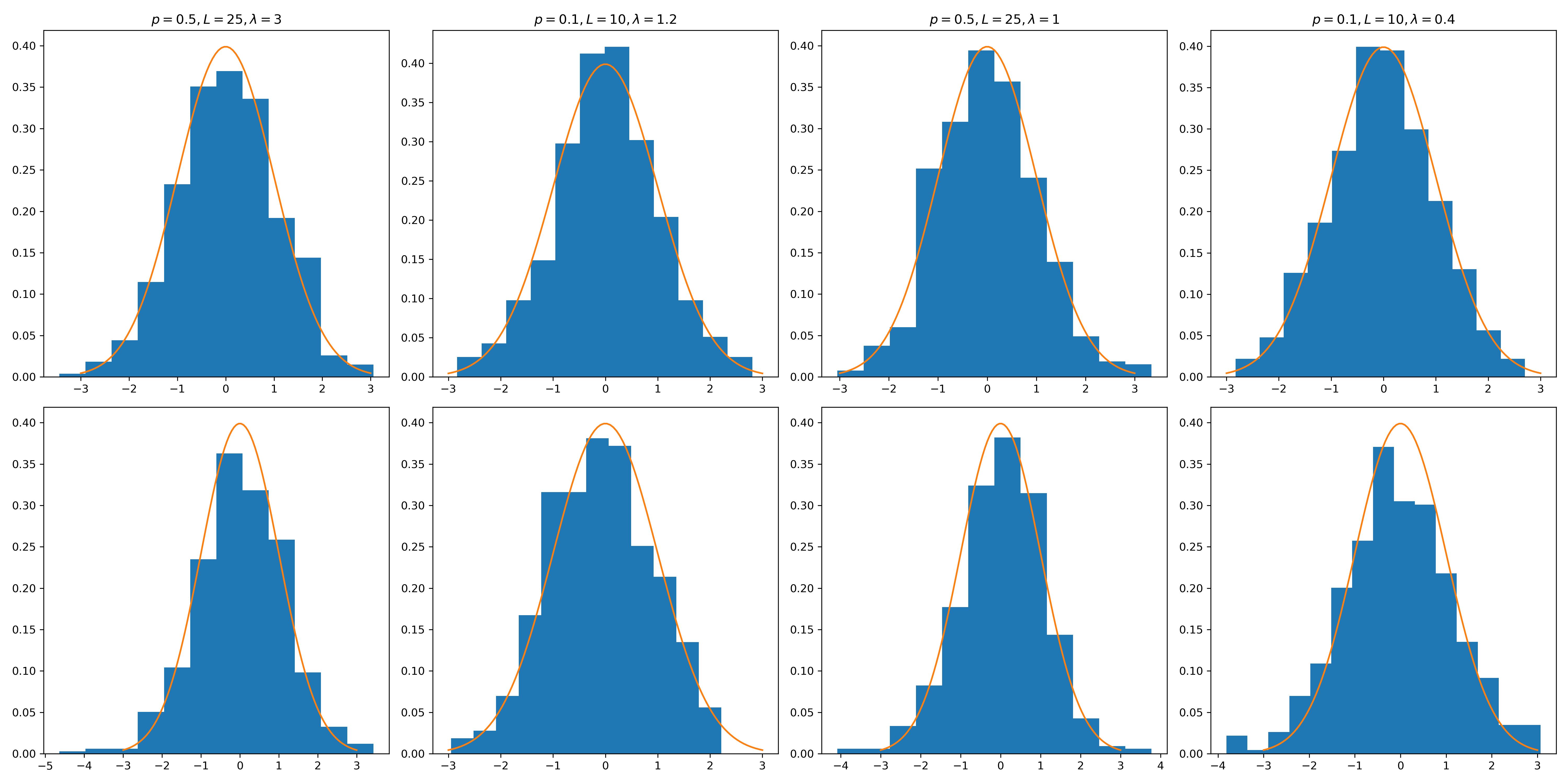}
\end{center}
\caption{Histograms of $RV_1$ and $RV_2$ against the standard Gaussian distribution under different parameter combinations. Each column represents one of the aforementioned combinations. The top and bottom panels are the histograms for $RV_1$ and $RV_2$, respectively, based on 500 simulations.}
\label{normalityhistogram}
\end{figure}

\begin{figure}[h]
\label{normalityQQplot2}
\begin{center}
\includegraphics[width=16cm]{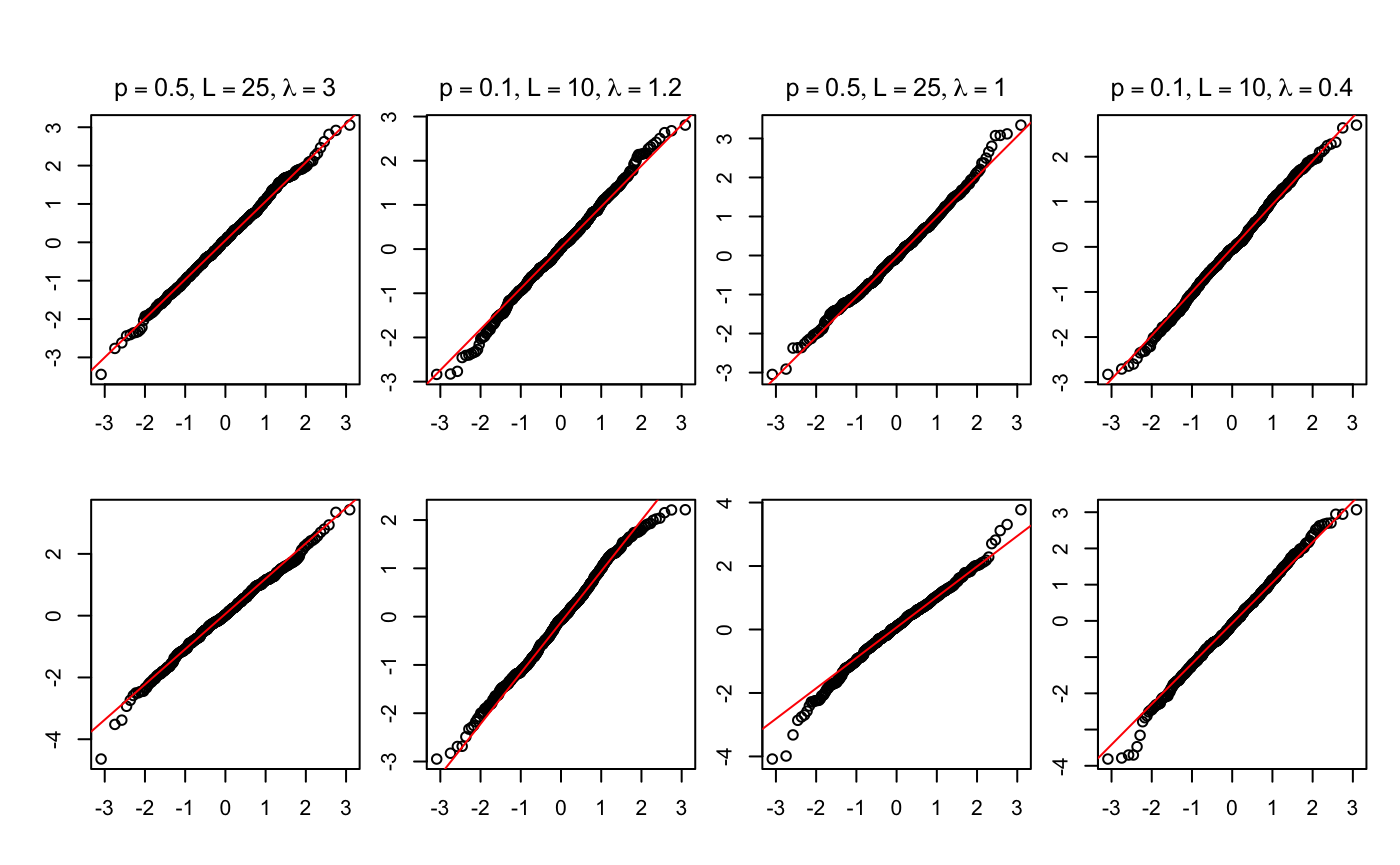}
\end{center}
\caption{Q-Q plots for checking the normality of $RV_1$ and $RV_2$. The results are reported in the same order as Figure \ref{normalityhistogram}. Each column represents one of the combination of the parameter. The top and bottom panels are Q-Q plots for $RV_1$ and $RV_2$, respectively, based on 500 simulations.} 
\end{figure}

\subsection{Goodness-of-Fit Test}
In this section, we validate the theoretical guarantee of the goodness-of-fit test presented in Section \ref{gofsection} through a synthetic dataset.

We keep $n,d,k$ in the same way as those in the previous section. In order to set the signal strength at different levels, we generate $\balpha_{\cS}$ and $\balpha_{\cS^c}$ separately. Specifically, we first generate two vectors, $\bomega_1\in \mathbb{R}^5$ and $\bomega_2\in \mathbb{R}^{195}$. The entries of $\bomega_1$ are sampled independently from $\text{Uniform}\left([-\log(5), -1]\cup[1,  \log(5)]\right)$, while all entries of $\bomega_2$ are fixed to be $0$. We set $\balpha^*$ at different signal levels as 
\begin{align*}
    \balpha^*(\rho) = \frac{3\rho}{100} \left[\bomega_1^\top, \bm{0}^\top\right]^\top,\quad \rho = 0,1,\dots, 5,
\end{align*}
where $\rho\in [0,5]$ controls the signal strength.
Additionally, $\bbeta^*$ and the covariate matrix are generated in a similar way as in the previous section. 

The comparison graph and results are generated with $p = 0.5$ and $L = 160$, and we fix $\lambda = 0.5$. Applying the approach presented in Section \ref{gofsection} \footnote{we let $\alpha = 0.05$ and estimate the critical value $c_{2, 1-\alpha}$ using $200$ bootstraps. Given $\balpha^*(\rho)$ at each signal level, the comparison graph $\mathcal{E}$ and data $\{y_{i,j}^{(\ell)},l\in[L],(i,j)\in\mathcal{E}\}$ are generated for $100$ times to calculate the power function $\widehat{P}(\mathcal{T}_1> c_{1, 1-\alpha})$}, we present the power functions in Figure {\red 3}. We conclude  from Figure {\red 3} that the Type I error is well controlled when the null hypothesis holds ($\rho=0$). When the alternative holds, as the signal level increases, the power of the test increases rapidly, and the empirical probability reaches $1$ when $\rho = 3$. This demonstrates the efficacy of our proposed method for conducting the good-of-fit test.

\begin{figure}[h]
\begin{center}
\includegraphics[width=12cm]{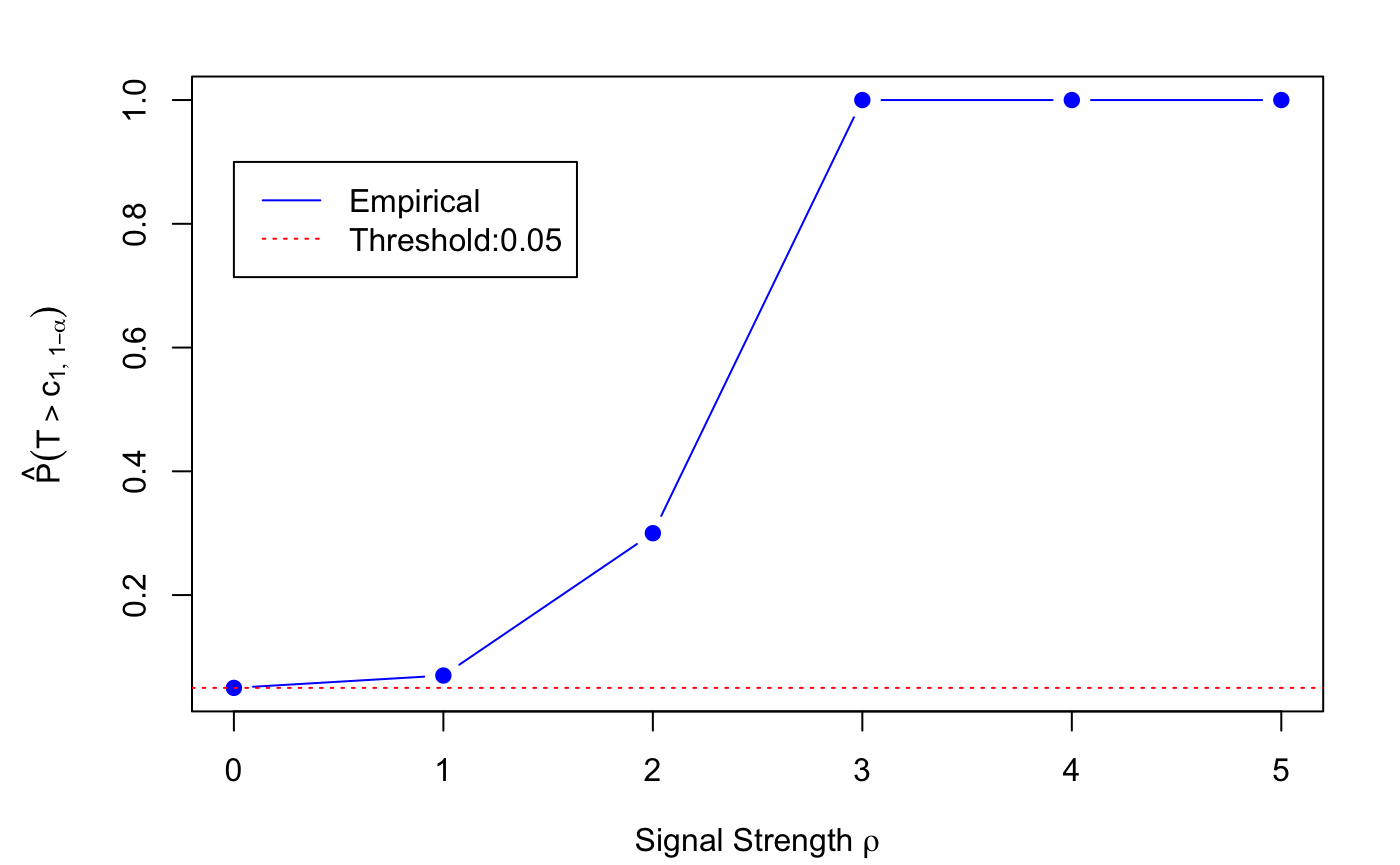}
\end{center}
\label{fig: GoFTest}
\caption{Empirical probability $\widehat{P}(\mathcal{T}_1> c_{1, 1-\alpha})$ at different signal level $\rho$ when $\alpha$ is fixed to be $0.05$. The empirical probability is calculated over $100$ repetitions for each signal level.} 
\end{figure}

\subsection{Rank Confidence Interval}
In this section, we study our out-of-sample ranking inferences application in Example \ref{RCI} and the corresponding bootstrap theory from Theorem \ref{T23thm} in detail using synthetic data. We let $n = 100$, $d = 3$ and let $k=|\cS(\balpha^*)|= 5$. For $i\in \cS(\ba^*)$, we generate $\alpha_i^*$ by sampling $|\alpha_i^*|\sim \text{Uniform}[0.3, 0.3\times \log(5)]$ with a random sign. For $\bbeta^*$ and the covariate matrix, they are generated in a similar way as in the previous sections.


The comparison graph is generated with $p = 0.5$ and $L = 160$. We are interested in six items with indices $T = \{1,2,3, 6,7,8\}$ as representatives to validate our method, containing $3$ items from $\cS(\ba^*)$ and $3$ items from $(\cS(\ba^*))^c$. For each $m\in T$, we apply our method in Example \ref{RCI} with $\cM = \{m\}$. The regularized estimator $\tb_\lambda$ is fitted with $\lambda=1$, and refitting with fixed support given by $\tb_\lambda$ yields the two-stage estimator $\widehat{\bgamma}$. We let $\alpha = 0.05$, and the critical value $c_{2,1-\alpha}$ is estimated by $200$ bootstrap samples.

The following Table {\red 1} summarizes the experiment results, where for both one-step and two-step estimators, we report: (i). $\text{EC}(r):$ the empirical coverage of the rank confidence interval $\widehat{P}(r(m)\in [\cR_L(m),\cR_U(m)])$, (ii). $\text{EC}(\theta):$ the empirical coverage of the simultaneous confidence intervals $\widehat{P}\left(\tilde{\theta}_k^* - \tilde{\theta}_m^*\in [\cC_L(k,m),\cC_U(k,m)], \forall k\neq m\right)$ (iii). the average length $\cR_U(m) - \cR_L(m)$ of the rank confidence interval, and the associated standard deviations. The empirical coverage  proportion and the mean and standard deviation of the length are calculated from $100$ repeated experiments.

\begin{table}[h!]
\begin{center}
\begin{tabular}{c||c|c|c||c|c|c}
\hline \multicolumn{1}{c||}{ } & \multicolumn{3}{c||}{\text{One-stage}} & \multicolumn{3}{c}{\text{Two-stage}}  \\
\hline item $m$ & $\mathrm{EC}(r)$ & $\mathrm{EC}(\theta)$ & Length & $\mathrm{EC}(r)$ & $\mathrm{EC}(\theta)$ & Length \\
\hline
$m = 1\; (r = 75)$ & 1 & 0.96 & 27.61 $\pm$ 2.97 & 1 & 0.94 & 13.0 $\pm$ 1.51 \\
$m = 2\; (r = 3)$ & 1 & 0.95 & 9.41 $\pm$ 2.67 & 1 & 0.93 & 4.23 $\pm$ 2.23 \\
$m = 3\; (r = 5)$ & 1 & 0.97 & 14.06 $\pm$ 2.44 & 1 & 0.94 & 8.35 $\pm$ 2.05 \\
$m = 6\; (r = 80)$ & 1 & 0.95 & 25.22 $\pm$ 2.03 & 1 & 0.88 & 3.68 $\pm$ 0.78 \\
$m = 7\; (r = 58)$ & 1 & 0.98 & 38.95 $\pm$ 4.01 & 1 & 0.89 & 6.03 $\pm$ 0.96 \\
$m = 8\; (r = 82)$ & 1 & 0.97 & 24.4 $\pm$ 2.03 & 1 & 0.89 & 3.56 $\pm$ 1.29 \\
\hline
\end{tabular}

\end{center}
\caption{Rank confidence interval for selected items  using one-stage and two-stage (post-selection MLE) approaches. $r$ represents the true rank of the selected items in this simulation experiment. $\text{EC}(r)$ is the the empirical coverage of the rank confidence interval and $\text{EC}(\theta)$ is the the empirical coverage of the score differences. The table also includes the average length of the rank confidence intervals and the associated standard deviation. Reported results are calculated over $100$ repetitions.}
\end{table}

\subsection{Application to Pokemon Challenge Dataset}

In this section, we apply our approaches to the Pokemon challenge dataset \url{https://www.kaggle.com/c/intelygenz-pokemon-challenge/data}. The dataset records $50000$ pairwise competitions among $800$ pokemons. Each pokemon is accompanied by a set of covariate information, and  each competition takes place between two pokemons and has one winner. We begin by utilizing our goodness-of-fit test in \S\ref{gofsection} to examine if covariates along can explain individual ability. Subsequently, we employ our predictive rank confidence interval approach to rank specific mega-evolved pokemon.

We assume that pokemons share the same intrinsic scores  $\balpha^*$ before and after mega evolution. Mega evolution only alters the covariates, thereby affecting the overall abilities of the pokemon. Consequently, it is natural to predict the abilities of mega-evolved pokemon using the combat results of their pre-evolutionary forms and their current covariates.

We randomly select $28$ mega evolved pokemons (out of a total of $48$ mega evolved pokemons) as our target. The remaining $800-28 = 772$ pokemons are left for training purpose. We constructed the comparison graph via the existing pairwise competitions. Since the graph was not connected, we selected the largest connected component as our training set to obtain a valid ranking result instead. Therefore, after this pre-processing step, we have $757$ pokemons left for training. For each pokemon, we consider a 3-dimensional covariates $\bx_i$ consisting of $\log(\emph{Attack})$, $\log(\emph{HP})$ and $\emph{Mega or not}$. \emph{Attack} and \emph{HP} represent the ability to attack and durability, respectively, while \emph{Mega or not} is a binary variable that denotes whether this pokemon is mega evolved or not. 

As the first step, we conduct the goodness-of-fit test in \S\ref{gofsection} to verify the existence of non-zero intrinsic scores. We consider five values of the regularization parameter $\lambda$: $\{10, 20, 30, 40, 50\}$ and conduct the test for each $\lambda$. In Table {\red 2}, we report the test statistic $\mathcal{T}_1$, the $\alpha = 0.05$ critical value $c_{1, 0.95}$, as well as the estimated support size $|\{i\in [757]:\widehat{\alpha}_i\neq 0\}|$ under each choice of $\lambda$. The critical value $c_{1, 0.95}$ is estimated by $200$ bootstrapping samples. From Table {\red 2} we can see that the null hypothesis is consistently rejected, and our approach is robust to the choice of regularization parameter $\lambda$. Note that in practice, one can choose $\lambda$ via cross-validation. 

\begin{table}[h]
	\begin{center}
		\begin{tabular}[!htbp]{ p{2cm}|p{2.5cm}p{2.5cm}p{5cm}}
			$\lambda$ & $\mathcal{T}_1$ &  $c_{1, 0.95}$ &  $|\{i\in [757]:\widehat{\alpha}_i\neq 0\}|$  \\
			\hline
			$10$ &$11.011$ &  $3.864$ & $505$\\ 
			$20$ &$9.571$ &  $4.374$ &  $277 $\\ 
			$30$ &$9.289$ &  $4.498$ & $86$ \\ 
            $40$ &$9.423$ &  $4.653$ & $4$\\ 
            $50$ &$9.464$ &  $4.552$ & $0$\\ 
		\end{tabular}
	\end{center}
 \label{goftable}
 \caption{Goodness-of-fit test statistics $\mathcal{T}_1$ and critical value $c_{1, 0.95}$, as well as support size under difference choice of regularization parameter $\lambda$.}
\end{table}

Next, we apply our approach to construct out-of-sample rank confidence intervals for the mega evolved pokemons. We fix the regularization parameter $\lambda = 30$ and fit your model to get $\widehat{\ba}_\lambda$. Next, we refit our model with support constrained on the support of $\widehat{\ba}_\lambda$ and obtain the two-stage estimator $\widehat{\bgamma}:=(\hat\balpha,\hat\bbeta)\in \mathbb{R}^{|\cS(\widehat{\ba}_\lambda)|+d}$. For each of the $28$ mega evolved pokemons with index $i\in  [28]$, we let pokemon with index $pe(i)$ be its pre-evolutionary version. If $pe(i)\in \cS(\widehat{\ba}_\lambda)$, the score of evolved pokemon $i \in [28]$ is predicted as $\widehat{\gamma}_{pe(i)}+ \boldsymbol{z}_i^\top\widehat{\bgamma}_{|\cS(\widehat{\ba}_\lambda)|+1:|\cS(\widehat{\ba}_\lambda)|+d}$, otherwise, the score of pokemon $i$ is predicted as $ \boldsymbol{z}_i^\top\widehat{\bgamma}_{|\cS(\widehat{\ba}_\lambda)|+1:|\cS(\widehat{\ba}_\lambda)|+d}$. On the other hand, we construct joint rank confidence interval with $\alpha = 0.05$ using the approach in \S\ref{oosRCIsection}. In Table {\red 3} we show the rank of the predicted scores and the associated confidence intervals. We only present the results for pokemons whose IDs are multiples of $5$ as representatives. The results regarding the lengths of the confidence intervals are consistent with our conclusions from the simulation results. 

\begin{table}[h]
	\begin{center}
		\begin{tabular}[!htbp]{p{3cm}|p{3cm}p{3cm}}
			Pokemon ID & One-Stage method  & Two-Stage method  \\
			\hline
			$20$ & $[3, 19]$ &$[7, 10]$\\ 
            $230$ & $[4, 23]$ &$[19, 19]$ \\
            $280$ & $[6, 23]$ &$[20, 21]$\\  
            $330$ & $[2, 17]$ &$[13, 13]$\\
            $340$ & $[10, 28]$ &$[12, 12]$\\
            $350$ & $[18, 28]$ &$[26, 26]$\\ 
            $410$ & $[10, 28]$ &$[23, 25]$\\
		\end{tabular}
	\end{center}
 \label{OfSRCItable}
 \caption{Rank confidence intervals for pokemons whose IDs are multiples of $5$ using One-Stage  and Two-Stage methods.}
\end{table}

\section{Conclusion}
In this paper, we study entity ranking with covariates as well as sparse intrinsic scores. We introduce a novel model identification condition and derive the optimal statistical rates of the regularized maximum likelihood estimator (MLE). We further design a debiased estimator of the MLE and derive its asymptotic distribution. Our proposed method is further applied to studying the goodness-of-fit test of the model with no intrinsic scores, and we construct prediction confidence intervals for future latent scores as well their associated confidence intervals for ranks.

There are several directions for future research that are worth exploring. First, while we focus on a linear model in this paper, it would be valuable to investigate more complex models for explaining the latent scores in the future. Second, we consider only one evaluation criterion (i.e., only a fixed $\beta^*$) for all items. It would be interesting to incorporate heterogeneous evaluation criteria into the models. Third, it would be interesting to see how to combine our ranking framework with online or offline decision making tasks such as connecting with reinforcement learning with human feedback.

\bibliographystyle{ims}
\bibliography{dynamic}

\newpage
\appendix

\appendix
\newpage 
\section{Additional applications}\label{sec:add_app}
We study the distribution of statistics
\begin{align*}
\cT_3 := \max_{m\in \cM}\max_{k\neq m}\frac{\widehat{\theta}_k - \widehat{\theta}_m - (\theta_k^* - \theta_m^*)}{\widehat{\sigma}_{m,k}}
\end{align*}
in order to construct one-sided rank confidence intervals. It's distribution can be approximated by the bootstrap counterpart
\begin{align*}
\cG_3 := \max_{m\in \cM}\max_{k\neq m}\sum_{l=1}^L \sum_{(i,j)\in \cE,i>j} \frac{(\tz_m-\tz_k)^\top(\nabla^2\cL(\tb_R))^\diamond (\tx_i-\tx_j)}{\widehat{\sigma}_{m,k}L}(\phi(\tx_i^\top\tb_R-\tx_j^\top\tb_R)-y_{j, i}^{(l)})\omega_{j, i}^{(l)}.
\end{align*}
We are able to achieve similar results as Theorem \ref{T23thm} for $\cT_3$ and $\cG_3$.
Next, we introduce some applications on constructing (simultaneous) one-sided confidence intervals for out-of-sample ranks via the distribution of $\cT_3$ in the following two examples.
\begin{example}
    For an item $m$ of interest, and let $K$ be the targeted rank threshold, we are interested in the following testing problem
    \begin{align}
        H_0: r(m)\leq K \quad \text{ versus }\quad  H_1:r(m)>K. \label{ranktest}
    \end{align}
    Let $\widehat{c}_{3,1-\alpha}$ be the estimated $(1-\alpha)$-th quantile of $\cT_3$ from the bootstrap samples. As a result, by a similar analysis of Theorem \ref{T23thm}, we have
    \begin{align*}
        P\left(\theta_k^*-\theta_m^*\geq \widehat{\theta}_k - \widehat{\theta}_m-\widehat{c}_{3,1-\alpha} \widehat{\sigma}_{m,k}\right) \geq 1-\alpha.
    \end{align*}
    Similarly, this implies 
    \begin{align*}
        P\left(r(m)\geq 1+\sum_{k\neq m}\mathbf{1}(\widehat{\theta}_k - \widehat{\theta}_m> \widehat{c}_{3,1-\alpha} \widehat{\sigma}_{m,k})\right) \geq 1-\alpha.
    \end{align*}
    This yields a critical region at a significance level of $alpha$ for the test \eqref{ranktest}
    \begin{align*}
        \left\{1+\sum_{k\neq m}\mathbf{1}(\widehat{\theta}_k - \widehat{\theta}_m> \widehat{c}_{3,1-\alpha} \widehat{\sigma}_{m,k})>K\right\}.
    \end{align*}
\end{example}

\begin{example}
    Given a number $K\in [n]$, we are interested in screening the top-$K$ ranked items, i.e., $\cK =\{r^{-1}(1), r^{-1}(2),\dots, r^{-1}(K)\}$. Let $\cM = [n]$ and $\widehat{c}_{3,1-\alpha}$ be the estimated $(1-\alpha)$-th quantile of $\cT_3$ from the bootstrap samples. Again by Theorem \ref{T23thm} we know that 
    \begin{align*}
        P\left(r(m)\geq 1+\sum_{k\neq m}\mathbf{1}(\widehat{\theta}_k - \widehat{\theta}_m> \widehat{c}_{3,1-\alpha} \widehat{\sigma}_{m,k}),\;\forall m\in [n])\right) \geq 1-\alpha.
    \end{align*}
    Therefore, we select the items as 
    \begin{align*}
        \widehat{\cI}_K = \left\{m\in [n]:1+\sum_{k\neq m}\mathbf{1}(\widehat{\theta}_k - \widehat{\theta}_m> \widehat{c}_{3,1-\alpha} \widehat{\sigma}_{m,k})\leq K\right\},
    \end{align*}
    and Theorem \ref{T23thm} ensures that 
    \begin{align*}
        P\left(\cK\subset \widehat{\cI}_K\right)\geq 1-\alpha.
    \end{align*}
\end{example}

\section{Proof Outline of Estimation Results}
\subsection{Preliminaries and Basic Results}
As the first step, let us look into the gradient and Hessian of the functions we are interested in. Except for $\cL(\cdot)$ and $\cL_R(\cdot)$, we also define
\begin{align*}
    \cL_\tau(\tb) := \cL(\tb)+\frac{\tau}{2}\left\|\tb\right\|_2^2.
\end{align*}
The gradient of $\cL_\tau(\cdot)$ is controlled by the following lemma.

\begin{lemma}\label{gradient}
With $\tau$ given by \eqref{lambdaandtau}, the following event 
\begin{align*}
    \mathcal{A}_1 = \left\{\left\|\nabla \mathcal{L}_{\tau}\left(\tb^*\right)\right\|_{2} \leq C_0\sqrt{\frac{n^{2} p \log n}{L}}\right\}
\end{align*}
happens with probability exceeding $1-O(n^{-10})$ for some $C_0>0$ which only depend on $c_\tau$.
\end{lemma}
\begin{proof}
The proof of Lemma \ref{gradient} follows a similar proof of \cite[Lemma 14]{fan2022uncertainty}. Therefore, we omit the details here.
\end{proof}
Let $\boldsymbol{L}_{\mathcal{G}} = \sum_{(i,j)\in\mathcal{E}, i>j}(\tx_i-\tx_j)(\tx_i-\tx_j)^\top$, its eigenvalues is studied in \cite[Lemma 15]{fan2022uncertainty}. We state it in the following lemma.

\begin{lemma}\label{eigenlem}
Suppose $pn > c_p\log n$ for some $c_p>0$. The following event
\begin{align*}
    \mathcal{A}_2 = \left\{\frac{1}{2}c_2pn\leq \lambda_{\text{min}, \perp}(\boldsymbol{L}_{\mathcal{G}})\leq \Vert \boldsymbol{L}_{\mathcal{G}}\Vert\leq 2c_1pn \right\}
\end{align*}
happens with probability exceeding $1-O(n^{-10})$ when $n$ is large enough.
\end{lemma}

In the rest of the content, without loss of generality, we assume the conditions stated in Lemma \ref{eigenlem} hold. 
Moreover, with the help of Lemma \ref{eigenlem}, we next analyze the Hessian $\nabla^2 \cL_\tau(\tb)$ and summarize its theoretical properties in Lemma \ref{ub} and Lemma \ref{lb}, respectively.

\begin{lemma}\label{ub}  Suppose event $\mathcal{A}_2$ holds, we obtain
\begin{align*}
    \lambda_{\text{max}}(\nabla^2 \mathcal{L}_\tau(\tb))\leq \tau +\frac{1}{2}c_1pn,\quad \forall \tb\in\mathbb{R}^{n+d}.
\end{align*}
\end{lemma}
\begin{proof}
Since $\displaystyle\frac{e^{\tx_i^\top \tb}e^{\tx_j^\top \tb}}{\left(e^{\tx_i^\top \tb}+e^{\tx_j^\top \tb}\right)^2}\leq \frac{1}{4}$, we have 
\begin{align*}
    \lambda_{\text{max}}(\nabla^2 \mathcal{L}_\tau(\tb))\leq \tau +\frac{1}{4}\Vert \boldsymbol{L}_{\mathcal{G}} \Vert\leq \tau +\frac{1}{2}c_1pn,\quad \forall \tb\in\mathbb{R}^{n+d}.
\end{align*}
\end{proof}
Lemma \ref{lb} can be viewed as the strong convexity property restricted on $\Theta(k)$.

\begin{lemma}\label{lb}
Suppose event $\mathcal{A}_2$ happens and $4c_0^2(d+1)k\leq n$. Then for all $\tb,\tb'\in\Theta(k)$ and $\tb''$ such that $\Vert\ba''-\ba^*\Vert_\infty\leq C_1$, $\Vert \bb''-\bb^*\Vert_2\leq C_2$, we have
\begin{align*}
    (\tb'-\tb)^\top\nabla^2 \mathcal{L}_\tau(\tb'')(\tb'-\tb)\geq \frac{1}{2}\left(\tau+\frac{c_2pn}{8\kappa_1 e^C}\right)\left\|\tb'-\tb\right\|_2^2,
\end{align*}
where $ C = 2C_1+2\sqrt{\frac{c_3(d+1)}{n}}C_2$.
\end{lemma}
\begin{proof}
	See \S\ref{lbproof} for a detailed proof.
\end{proof}

We then consider the following proximal gradient descent procedure. We set the step size $\eta = \frac{2}{2\tau+c_1 pn}$ and number of iterations $T=n^5$.
\begin{algorithm}[H]
\caption{Proximal Gradient descent for regularized MLE.}
\begin{algorithmic}
\STATE \textbf{Initialize} $\tb^0 = \tb^*$, step size $\eta$, number of iterations $T$

\FOR{$t=0,1,\dots, T-1$ } 
\STATE {$\tb^{t+1} = \textsf{SOFT}_{\eta\lambda}\left(\tb^t-\eta \nabla \mathcal{L}_\tau (\tb^t)\right)$} 
\ENDFOR
\end{algorithmic}
\end{algorithm}
\noindent Here we let $s(x,\gamma) := \text{sign}(x)\cdot\max\left\{|x|-\gamma,0\right\}$ and define 
\begin{align*}
    \textsf{SOFT}_{\gamma}(\tb) = \left[s(\tilde{\beta}_1,\gamma), s(\tilde{\beta}_2,\gamma),\dots, s(\tilde{\beta}_n,\gamma), \tb_{n+1:n+d}\right]^\top
\end{align*}
for any vector $\tb\in\mathbb{R}^{n+d}$. Since $\cL_\tau(\tb)$ is $\tau$-strongly convex, the above proximal gradient descent enjoys exponential convergence. It is formalized in the following results.
\begin{lemma}\label{lem2}
Under event $\mathcal{A}_2$, we have
\begin{align*}
    \left\Vert \tb^t-\tb_R\right\Vert_2\leq \rho^t\left\Vert \tb^0-\tb_R\right\Vert_2,
\end{align*}
where $\displaystyle\rho = 1-\eta\tau$.
\end{lemma}
\begin{proof}
	See \S\ref{lem2proof} for a detailed proof.
\end{proof}

\begin{lemma}\label{lem3}
On the event $\mathcal{A}_1$ happens, it follows that
\begin{align*}
    \left\Vert \tb^0-\tb_R\right\Vert_2 = \left\Vert \tb_R -\tb^*\right\Vert_2 \leq \frac{2C_0\sqrt{n}}{c_\tau}\max\left\{\frac{\kappa_2}{\kappa_1},\kappa_3\right\}+ \sqrt{\frac{2c_\lambda\sqrt{d+1}}{c_\tau}\max\left\{\kappa_2^2,\kappa_1\kappa_2\kappa_3\right\}}.
\end{align*}
\end{lemma}
\begin{proof}
	See \S\ref{lem3proof} for a detailed proof.
\end{proof}

\begin{lemma}\label{lem5}
	On event $\mathcal{A}_1 \cap \mathcal{A}_2$, there exists some constant $C_7$ such that 
\begin{align*}
    \max\left\{\left\Vert \tb^{T-1}-\tb_R\right\Vert_2, \left\Vert \tb^T-\tb_R\right\Vert_2\right\} \leq C_7\kappa_1\sqrt{\frac{(d+1)\log n}{npL}}.
\end{align*}
\end{lemma}
\begin{proof}
	See \S\ref{lem5proof} for a detailed proof.
\end{proof}

Next, we will leverage the leave-one-out technique and use induction to prove that the iterate $\tb^t$ stays close to the initial point $\tb^0=\tb^*$ during all the iterations $t=0, 1, 2,\dots, T-1$. 

\subsection{Analysis of Leave-one-out Sequences}
In this section, we construct the leave-one-out sequences \citep{ma2018implicit,chen2019spectral,chen2020noisy} and bound the statistical error by induction. We consider the following loss function for any $m\in[n]$ to construct the leave-one-out sequence.
\begin{align*}
    \mathcal{L}^{(m)}(\tb) =& \sum_{(i, j) \in \mathcal{E}, i>j, i\neq m, j\neq m}\left\{-y_{j, i}\left(\tx_i^\top\tb-\tx_j^\top\tb\right)+\log \left(1+e^{\tx_i^\top\tb-\tx_j^\top\tb}\right)\right\}\\
    &+p\sum_{i\neq m}\left\{-\frac{e^{\tx_i^\top\tb^*}}{e^{\tx_i^\top\tb^*}+e^{\tx_m^\top\tb^*}}(\tx_i^\top\tb-\tx_m^\top\tb) + \log \left(1+e^{\tx_i^\top\tb-\tx_m^\top\tb}\right)\right\} ;\\
    \mathcal{L}_\tau^{(m)}(\tb)=&\mathcal{L}^{(m)}(\tb)+\frac{\tau}{2}\Vert\tb\Vert_2^2.
\end{align*}
Then for any $m\in[n]$, we construct the leave-one-out sequence $\left\{\tb^{t,(m)}\right\}_{t = 0,1,\dots}$ in the way of Algorithm \ref{alg2}.

\begin{algorithm}[H]
\caption{Construction of leave-one-out sequences.}
	\begin{algorithmic}[1]		
		\STATE \textbf{Initialize} $\tb^{0,(m)} = \tb^*$
        \FOR{$t=0,1,\dots, T-1$ } 
        \STATE {$\tb^{t+1,(m)} = \textsf{SOFT}_{\eta\lambda}\left(\tb^{t,(m)}-\eta \nabla \mathcal{L}^{(m)}_\tau \left(\tb^{t,(m)}\right)\right)$} 
        \ENDFOR
	\end{algorithmic}\label{alg2}
\end{algorithm}
With the help of the leave-one-out sequences, we do induction to demonstrate that the iterate $\tb^T$ will not be far away from $\tb^*$ when $T=n^5.$ With the leave-one-out sequences in hand, we prove the following bounds by induction for $t \leq T$. 
\begin{align}
    \left\Vert \tb^t-\tb^*\right\Vert_2&\leq C_3\kappa_1\sqrt{\frac{\log n}{pL}};\tag{A} \label{inductionA}\\
    \max_{1\leq m\leq n}\left\Vert\tb^t- \tb^{t,(m)}\right\Vert_2&\leq C_4\kappa_1\sqrt{\frac{(d+1)\log n}{npL}}\leq C_4\kappa_1\sqrt{\frac{\log n}{pL}}; \tag{B}\\
    \forall m\in [n],\quad |\alpha_m^{t,(m)}-\alpha_m^*|&\leq C_5\kappa_1^2\sqrt{\frac{(d+1)\log n}{npL}},\;\; \cS(\ba^{t,(m)})\subset \cS(\ba^*);  \tag{C} \\
    \Vert \ba^t-\ba^*\Vert_\infty&\leq C_6 \kappa_1^2\sqrt{\frac{(d+1)\log n}{npL}} ,\;\; \cS(\ba^t)\subset \cS(\ba^*).\tag{D}\label{inductionD}
\end{align}
For $t = 0$, since $\tb^0 = \tb^{0,(1)} = \tb^{0,(2)} = \dots = \tb^{0,(n)} = \tb^*$, the \eqref{inductionA}$\sim$ \eqref{inductionD} hold automatically. In the following lemmas, we prove the conclusions of \eqref{inductionA}-\eqref{inductionD} for the $(t+1)$-th iteration are true when the results hold for the $t$-th iteration.

\begin{lemma}\label{induction1}
Suppose bounds \eqref{inductionA}$\sim$ \eqref{inductionD} hold for the $t$-th iteration. With probability exceeding $1-O(n^{-11})$ we have
\begin{align*}
    \left\Vert \tb^{t+1}-\tb^*\right\Vert_2&\leq C_3\kappa_1\sqrt{\frac{\log n}{pL}},
\end{align*}
as long as $\displaystyle 0<\eta\leq \frac{2}{2\tau+c_1np}$, $\displaystyle C_3\geq \frac{40C_0}{c_2}$, $\displaystyle k(d+1)\leq \frac{c_2^2 C_3^2 n}{1600 c_\lambda^2\kappa_1^2}$ and $n$ is large enough.
\end{lemma}
\begin{proof}
	See \S\ref{induction1proof} for a detailed proof.
\end{proof}

\begin{lemma}\label{induction3}
Suppose bounds \eqref{inductionA}$\sim$ \eqref{inductionD} hold for the $t$-th iteration. With probability exceeding $1-O(n^{-11})$ we have
\begin{align*}
    \max_{1\leq m\leq n}\left\Vert\tb^{t+1}- \tb^{t+1,(m)}\right\Vert_2&\leq C_4\kappa_1\sqrt{\frac{(d+1)\log n}{npL}},
\end{align*}
as long as $\displaystyle 0<\eta\leq \frac{2}{2\lambda+c_1np}$, $\displaystyle C_4\gtrsim \frac{1}{c_2}$ and $np\gtrsim (d+1)\log n$.
\end{lemma}
\begin{proof}
	See \S\ref{induction3proof} for a detailed proof.
\end{proof}

\begin{lemma}\label{induction2}
Suppose bounds \eqref{inductionA}$\sim$ \eqref{inductionD} hold for the $t$-th iteration. With probability exceeding $1-O(n^{-11})$ we have
\begin{align*}
    \forall m\in [n],\quad |\alpha_m^{t+1,(m)}-\alpha_m^*|&\leq C_5\kappa_1^2\sqrt{\frac{(d+1)\log n}{npL}},\;\; \cS(\ba^{t+1,(m)})\subset\cS(\ba^*),
\end{align*}
as long as $C_5\geq 30c_0(C_0+c_1C_3+c_1C_4)$, $C_5\geq 7.5(1+2\sqrt{c_3})(C_3+C_4)$, $C_5\geq 30c_\tau/\sqrt{d+1}$, $c_\lambda\geq \frac{3\sqrt{c_3}}{4}(C_3+C_4)+\frac{2C_4}{\eta np}$ and $n$ is large enough.
\end{lemma}
\begin{proof}
	See \S\ref{induction2proof} for a detailed proof.
\end{proof}

\begin{lemma}\label{induction4}
Suppose bounds \eqref{inductionA}$\sim$ \eqref{inductionD} hold for the $t$-th iteration. With probability exceeding $1-O(n^{-11})$ we have
\begin{align*}
    \Vert \ba^{t+1}-\ba^*\Vert_\infty&\leq C_6\kappa_1^2\sqrt{\frac{(d+1)\log n}{npL}},\;\; \cS(\ba^{t+1})\subset\cS(\ba^*),
\end{align*}
as long as $C_6\geq C_4+C_5$, $c_\lambda\geq \frac{3\sqrt{c_3}}{4}(C_3+C_4)+\frac{2C_4}{\eta np}$ and $n$ is large enough.
\end{lemma}
\begin{proof}
	See \S\ref{induction4proof} for a detailed proof.
\end{proof}

\begin{lemma}\label{supportlemma}
With probability exceeding $1-O(n^{-6})$ we have
\begin{align*}
    \cS(\widehat\ba_R)\subset \cS(\ba^*).
\end{align*}
as long as $c_\lambda\geq \frac{3\sqrt{c_3}}{4}(C_3+C_4)+\frac{C_4+C_7}{\eta np}$.
\end{lemma}
Combine Lemma \ref{lem5}, Lemma \ref{induction4} and Lemma \ref{supportlemma} gives us Theorem \ref{estimationthm}.
\begin{proof}
	See \S\ref{supportlemmaproof} for a detailed proof.
\end{proof}

\section{Proof Outline of Uncertainty Quantification Results}
Let $\ocL(\tb)$ be the quadratic expansion of the loss function $\cL(\tb)$ around $\tb^*$ given by
\begin{align}\label{quadratic_expan}
    \ocL(\tb) = \cL(\tb^*)+\left(\tb-\tb^*\right)^\top\nabla \cL(\tb^*)+\frac{1}{2}\left(\tb-\tb^*\right)^\top\nabla^2 \cL(\tb^*)\left(\tb-\tb^*\right).
\end{align}
Correspondingly, we define
\begin{align*}
    \ocL_R(\tb) = \ocL(\tb)+\lambda \left\|\ba\right\|_1+ \frac{\tau}{2}\left\|\tb\right\|_2^2 \text{ and } \ob_R = \argmin_{\tb\in\mathbb{R}^{n+d}} \ocL_R(\tb).
\end{align*}
First we state the following lemma for $\ob_R$.
\begin{lemma}\label{oaestimation}
For $\lambda$ and $\tau$ defined in Eq.~\eqref{lambdaandtau}, as long as $np\geq C\kappa_1^2(k+d)$ and $n\geq C\kappa_1^2(k+d)(d+1)$ for some constant $C>0$, with probability at least $1-O(n^{-6})$, we have $\cS(\oa_R)\subset \cS(\ba^*)$ and
\begin{align*}
    \left\|\ob_R-\tb^*\right\|_\infty\lesssim \kappa_1^2\sqrt{\frac{(k+d)(d+1)\log n}{npL}}, \quad  \left\|\ob_R-\tb^*\right\|_2\lesssim \kappa_1^2\sqrt{\frac{(k+d)(d+1)\log n}{npL}}.
\end{align*}
\end{lemma}
\begin{proof}
    See \S \ref{oaestimationproof} for a detailed proof.
\end{proof}

With Lemma \ref{oaestimation} in hand, we control the difference $\|\tb_R-\ob_R\|_2$ then.
\begin{theorem}\label{residue2norm}
    For $\lambda$ and $\tau$ defined in Eq.~\eqref{lambdaandtau}, as long as $np\geq C\kappa_1^2(k+d)$ and $n\geq C\kappa_1^2(k+d)(d+1)$ for some constant $C>0$, with probability at least $1-O(n^{-6})$, we have 
\begin{align*}
    \left\|\tb_R-\ob_R\right\|_2&\lesssim \kappa_1^{3.5}\left(\frac{(k+d)(d+1)\log n}{npL}\right)^{3/4}.
\end{align*}
\end{theorem}
\begin{proof}
    See \S \ref{residue2normproof} for a detailed proof.
\end{proof}

Next, we introduce the debiased version of $\tb_R$ and $\ob_R$ and prove the corresponding approximation results. Given a vector $\mathbf{x}\in \mathbb{R}^{n+d}$ and a function $\mathcal{M}:\mathbb{R}^{n+d}\rightarrow \mathbb{R}$, we define 
\begin{align*}
    \mathcal{M}_{\mathbf{x}_{-i}}(x) = \mathcal{M}(\tb)\bigg|_{\tb_i=x, \tb_{-i}=\mathbf{x}_{-i}}.
\end{align*}
With this definition, one can see that given any $i\in [n]$, we have
\begin{align*}
    \widehat{\alpha}_{R,i} = \argmin \cL_{R,\tb_{R, -i}}(x),\quad \overline{\alpha}_{R,i} = \argmin \ocL_{R, \ob_{R, -i}}(x).
\end{align*}
Take $\overline{\alpha}_{R,i}$ as an example first. From the derivative we know that
\begin{align*}
    0 &= \ocL_{R,\ob_{R, -i}}'(\overline{\alpha}_{R,i}) = \ocL_{\ob_{R, -i}}'(\overline{\alpha}_{R,i}) +\tau \overline{\alpha}_{R,i}+\lambda\partial |\overline{\alpha}_{R,i}| \\
    & = \ocL_{\ob_{R, -i}}'(\alpha_{i}^*) + \ocL_{\ob_{R, -i}}''(\alpha_{i}^*)(\overline{\alpha}_{R,i}-\alpha_i^*) +\tau \overline{\alpha}_{R,i}+\lambda\partial |\overline{\alpha}_{R,i}| \\
    & = \ocL_{\ob_{R, -i}}'(\alpha_{i}^*) + \left(\nabla^2\cL(\tb^*)\right)_{i,i}(\overline{\alpha}_{R,i}-\alpha_i^*) +\tau \overline{\alpha}_{R,i}+\lambda\partial |\overline{\alpha}_{R,i}|.
\end{align*}
In other words, we can write
\begin{align}
    \overline{\alpha}_{R,i}+\frac{\tau \overline{\alpha}_{R,i}+\lambda\partial |\overline{\alpha}_{R,i}|}{\left(\nabla^2\cL(\tb^*)\right)_{i,i}} = \alpha^*_i - \ocL_{\ob_{R, -i}}'(\alpha_{i}^*)/\left(\nabla^2\cL(\tb^*)\right)_{i,i}. \label{oadexpansion}
\end{align}
Therefore, we define the debiased estimator as
\begin{align*}
     \overline{\alpha}_{R,i}^{\debias} = \overline{\alpha}_{R,i}+\frac{\tau \overline{\alpha}_{R,i}+\lambda\partial |\overline{\alpha}_{R,i}|}{\left(\nabla^2\cL(\tb^*)\right)_{i,i}},
\end{align*}
similar to $\widehat{\alpha}_{R,i}^{\debias}$ we defined in \eqref{debiasedestimator}. Then in the following content, we focus on controlling the difference $|\widehat{\alpha}_{R,i}^{\debias} - \overline{\alpha}_{R,i}^{\debias}|$. In order to do so, we construct an auxiliary function and consider its minimizer
\begin{align}
    \dot{\alpha}_{R,i} = \argmin \ocL_{R,\tb_{R, -i}}(x). \label{dotalphadefinition}
\end{align}
Again, we define the debiased estimator as
\begin{align*}
    \dot{\alpha}_{R,i}^{\debias} = \dot{\alpha}_{R,i}+\frac{\tau \dot{\alpha}_{R,i}+\lambda\partial |\dot{\alpha}_{R,i}|}{\left(\nabla^2\cL(\tb^*)\right)_{i,i}}.
\end{align*}
Next, we control $|\dot{\alpha}_{R,i}^{\debias} - \overline{\alpha}_{R,i}^{\debias}|$ and $|\widehat{\alpha}_{R,i}^{\debias} - \dot{\alpha}_{R,i}^{\debias}|$ separately. 

\begin{lemma}\label{alphadebiasapprox1}
Under the conditions of Theorem \ref{residue2norm}, with probability at least $1-O(n^{-6})$ we have
\begin{align*}
    | \dot{\alpha}_{R,i}^{\debias} - \overline{\alpha}_{R,i}^{\debias}|\lesssim \frac{\kappa_1^{4.5}(d+1)}{np}\sqrt{\frac{k\log n}{L}}\left(\frac{(k+d)(d+1)\log n}{npL}\right)^{1/4}.
\end{align*}
\end{lemma}
\begin{proof}
See \S \ref{alphadebiasapprox1proof} for a detailed proof.
\end{proof}

\begin{lemma}\label{alphadebiasapprox2}
Under the conditions of Theorem \ref{residue2norm}, with probability at least $1-O(n^{-6})$ we have
\begin{align*}
     |\widehat{\alpha}_{R,i}^{\debias} - \dot{\alpha}_{R,i}^{\debias}|\lesssim \kappa_1^6 \frac{(d+1)\log n}{npL}
\end{align*}
\end{lemma}
\begin{proof}
    See \S \ref{alphadebiasapprox2proof} for a detailed proof.
\end{proof}
From \eqref{oadexpansion} we already know that 
\begin{align*}
    \overline{\alpha}_{R,i}^\debias = \alpha^*_i - \ocL_{\ob_{R, -i}}'(\alpha_{i}^*)/\left(\nabla^2\cL(\tb^*)\right)_{i,i}.
\end{align*}
To get the asymptotic distribution of $\overline{\alpha}_{R,i}^\debias$, we approximate $\overline{\alpha}_{R,i}^\debias$ by $\alpha^*_i - \left(\nabla\cL(\tb^*)\right)_i/\left(\nabla^2\cL(\tb^*)\right)_{i,i}$. The following lemma ensures the approximation error is small.
\begin{lemma}\label{alphadebiasapprox3}
Under the conditions of Theorem \ref{residue2norm}, with probability at least $1-O(n^{-6})$ we have
\begin{align*}
    \left|\overline{\alpha}_{i,R}^\debias - \left(\alpha^*_i - \left(\nabla\cL(\tb^*)\right)_i/\left(\nabla^2\cL(\tb^*)\right)_{i,i}\right)\right|\lesssim \frac{\kappa_1^3(d+1)}{np}\sqrt{\frac{(k+d)\log n}{L}}.
\end{align*}
\end{lemma}
\begin{proof}
    See \S \ref{alphadebiasapprox3proof} for a detailed proof.
\end{proof}

Next, we consider the expansion of $\widehat{\bb}_R$. Given a vector $\ba\in \mathbb{R}^{n}$, a function $\mathcal{M}:\mathbb{R}^{n+d}\rightarrow \mathbb{R}$ and a vector $\bb\in \mathbb{R}^d$, we define 
\begin{align*}
    \mathcal{M}_{\ba}(\bb) = \mathcal{M}(\tb)\bigg|_{\tb_{1:n}=\ba, \tb_{n+1:n+d}=\bb}.
\end{align*}
Therefore, it is easy to see that
\begin{align*}
    \widehat{\bb}_R = \argmin \cL_{R,\widehat{\ba}_R}(\bb),\quad \overline{\bb}_{R,n+1:n+d} = \argmin \ocL_{R,\oa_R}(\bb).
\end{align*}
By the optimality condition we know that
\begin{align*}
    \boldsymbol{0}_{d} &= \nabla\ocL_{R,\oa_R}(\overline{\bb}_{R,n+1:n+d}) = \tau\overline{\bb}_{R,n+1:n+d}+  \nabla \ocL_{\oa_R}(\overline{\bb}_{R,n+1:n+d}) \\
    &=\tau\overline{\bb}_{R,n+1:n+d}+ \left(\nabla\cL(\tb^*)\right)_{n+1:n+d} + \left(\nabla^2\cL(\tb^*)\right)_{n+1:n+d, :}(\overline{\bb}_{R} - \bb^*).
\end{align*}
Reorganizing the terms we get
\begin{align*}
    \overline{\bb}_{R, n+1:n+d} = \left(\boldsymbol{A}+\tau\bI_d\right)^{-1}\left(\boldsymbol{A} \bb^*- \left(\nabla\cL(\tb^*)\right)_{n+1:n+d} - \boldsymbol{B}\left(\oa_R-\ba^*\right)\right),
\end{align*}
 where $\bA:= (\nabla^2\cL(\tb^*))_{n+1:n+d, n+1:n+d}$ and $\bB:= (\nabla^2\cL(\tb^*))_{n+1:n+d, 1:n}$. Inspired by this, we debias $\widehat{\bb}_R$ and $\overline{\bb}_{R, n+1:n+d}$ as
\begin{align}
    \widehat{\bb}_R^\debias = \bA^{-1}\left(\boldsymbol{A}+\tau\bI_d\right)\widehat{\bb}_R,\quad \overline{\bb}_{R, n+1:n+d}^\debias = \bA^{-1}\left(\boldsymbol{A}+\tau\bI_d\right)\overline{\bb}_{R, n+1:n+d}. \label{betadebiasdefinition}
\end{align}
In order to analyze the asymptotic distribution of $$\overline{\bb}_{R, n+1:n+d}^\debias = \bb^*- \bA^{-1}\left(\nabla\cL(\tb^*)\right)_{n+1:n+d} - \bA^{-1}\boldsymbol{B}\left(\oa_R-\ba^*\right),$$ we approximate it by $\bb^*- \bA^{-1}\left(\nabla\cL(\tb^*)\right)_{n+1:n+d}$. The following lemma controls the approximation error.
\begin{lemma}\label{betadebiasapprox}
Under the conditions of Theorem \ref{residue2norm}, with probability at least $1-O(n^{-6})$ we have
\begin{align*}
    \left\|\overline{\bb}_{R, n+1:n+d}^\debias -\left( \bb^*- \bA^{-1}\left(\nabla\cL(\tb^*)\right)_{n+1:n+d}\right)\right\|_2\lesssim \frac{\kappa_1^3(d+1)}{np}\sqrt{\frac{kd(k+d)\log n}{L}}.
\end{align*}
\end{lemma}
\begin{proof}
    See \S\ref{betadebiasapproxproof} for a detailed proof.
\end{proof}

On the other hand, as long as $\tau$ is sufficiently small, the difference between $\widehat{\bb}_R$ and $\widehat{\bb}_R^\debias$ is also very small, which means that there is no need to debias $\widehat{\bb}_R$. We have the following result.
\begin{lemma}\label{betaapprox}
Under the conditions of Theorem \ref{residue2norm}, with probability at least $1-O(n^{-6})$ we have
\begin{align*}
    \left\|\widehat{\bb}_R - \widehat{\bb}_R^\debias\right\|_2 \lesssim \frac{\kappa_1}{np}\sqrt{\frac{\log n}{L}}.
\end{align*}    
\end{lemma}

\begin{proof}
By the definition of the debiased estimator \eqref{betadebiasdefinition} we have
\begin{align*}
    \left\|\widehat{\bb}_R - \widehat{\bb}_R^\debias\right\|_2 = \left\|\widehat{\bb}_R - \bA^{-1}\left(\boldsymbol{A}+\tau\bI_d\right)\widehat{\bb}_R\right\|_2 = \tau\left\|\bA^{-1}\widehat{\bb}_R\right\|_2\lesssim \frac{\tau\kappa_1}{np}\left\|\widehat{\bb}_R\right\|_2\lesssim \frac{\kappa_1}{np}\sqrt{\frac{\log n}{L}}.
\end{align*}
\end{proof}

We combine the aforementioned results in this section and get the following expansion for $\widehat{\alpha}_{R,i}^{\debias}$ and $\hat{\bb}_{R}$.

\begin{theorem}\label{approxthm}
Under the conditions of Theorem \ref{residue2norm}, with probability at least $1-O(n^{-6})$ we have
\begin{align*}
    &\left|\widehat{\alpha}_{R,i}^{\debias} - \left(\alpha^*_i - \left(\nabla\cL(\tb^*)\right)_i/\left(\nabla^2\cL(\tb^*)\right)_{i,i}\right)\right| \\
    \lesssim & \frac{\kappa_1^3(d+1)}{np}\left(\frac{\kappa_1^3\log n}{L}+\sqrt{\frac{(k+d)\log n}{L}}\right), \\
    &\left\|\hat{\bb}_{R} -\left( \bb^*- \bA^{-1}\left(\nabla\cL(\tb^*)\right)_{n+1:n+d}\right)\right\|_2\\
    \lesssim& \frac{\kappa_1^3(d+1)}{np}\sqrt{\frac{kd(k+d)\log n}{L}} + \kappa_1^{4.5}\left(\frac{(k+d)(d+1)\log n}{npL}\right)^{3/4}.
\end{align*}    
\end{theorem}
\begin{proof}
    See \S\ref{approxthmproof} for a detailed proof.
\end{proof}

\subsection{Proof of the two-stage method in Section \ref{sec:distribution}}\label{proof_twostage}

We state the estimation error bounds for $\widehat{\bgamma}$ here. We define $\bgamma^* = \tb^*_{[n+d]\backslash \cS(\ba^*)}$.

\begin{lemma}\label{estimation_eror}
    Given $\cS(\widehat{\ba}_R) = \cS(\ba^*)$ and the aforementioned two stage estimator $\widehat{\bgamma}$, as long as $npL\gtrsim \kappa_1^2(k+d)\log n$, with probability exceeding $1-O(n^{-10})$, we have
    \begin{align*}
        \left\|\widehat{\bgamma}-\bgamma^*\right\|_2\lesssim \kappa_1\sqrt{\frac{(k+d)\log n}{npL}}.
    \end{align*}
\end{lemma}

\begin{proof}
    We apply the \cite[Corollary A.1]{fan2020factor} directly. Let $A = 1$, by Lemma \ref{lb} we know that the conditions in \cite[Corollary A.1]{fan2020factor} hold with $\kappa\gtrsim np/\kappa_1$. On the other hand, by Lemma \ref{cLinfty} we know that 
    \begin{align*}
        \left\|\nabla \tilde{\cL}(\bgamma^*)\right\|_2\lesssim \sqrt{\frac{(k+d)np\log n}{L}}.
    \end{align*}
    As a result, as long as $npL\gtrsim \kappa_1^2(k+d)\log n$, we have 
    \begin{align*}
        \left\|\widehat{\bgamma}-\bgamma^*\right\|_2\lesssim \kappa_1\sqrt{\frac{(k+d)\log n}{npL}}.
    \end{align*}
    
\end{proof}

We next approximate the estimator $\widehat{\bgamma}$ by the minimizer of the quadratic approximation of $\tilde{\cL}$. Specifically, we define
\begin{align*}
    \tilde{\ocL}(\bgamma) = \tilde{\cL}(\bgamma^*) + \nabla\tilde{\cL}(\bgamma^*)^\top (\bgamma - \bgamma^*)+\frac{1}{2}(\bgamma - \bgamma^*)^\top \nabla^2\tilde{\cL}(\bgamma^*)(\bgamma - \bgamma^*),
\end{align*}
with $\overline{\bgamma} = \argmin \tilde{\ocL}(\bgamma)$. As a result, it holds that
\begin{align*}
    \overline{\bgamma} = \bgamma^* - \left(\nabla^2\tilde{\cL}(\bgamma^*)\right)^{-1}\nabla\tilde{\cL}(\bgamma^*).
\end{align*}
Then the following result controls the difference between $\widehat{\bgamma}$ and $\overline{\bgamma}$.
\begin{proposition}
\label{gammaexpansionthm}
Given $\cS(\widehat{\ba}_R) = \cS(\ba^*)$ and the aforementioned two-stage estimator $\widehat{\bgamma}$, as long as $npL\gtrsim \kappa_1^2(k+d)\log n$, with probability exceeding $1-O(n^{-10})$, we have
\begin{align*}
    \left\|\overline{\bgamma}-\widehat{\bgamma}\right\|_2\lesssim \kappa_1^3 \frac{(k+d)\log n}{npL}.
\end{align*}    
\end{proposition}

\begin{proof}
    The optimality conditions tell us
    \begin{align*}
         \boldsymbol{0} = \nabla \tilde{\cL}(\widehat{\bgamma}) &= \nabla \tilde{\cL} (\bgamma^*) +\int_{0}^1\nabla^2\tilde{\cL}(\bgamma^*+t(\widehat{\bgamma}-\bgamma^*))\left(\widehat{\bgamma}-\bgamma^*\right)dt \\
    &= \nabla \tilde{\cL} (\bgamma^*) +\int_{0}^1\nabla^2\tilde{\cL}(\bgamma^*+t(\widehat{\bgamma}-\bgamma^*))dt\left(\widehat{\bgamma}-\bgamma^*\right).  \\
    \boldsymbol{0} = \nabla \tilde{\ocL}(\overline{\bgamma}) &= \nabla \tilde{\cL} (\bgamma^*) +\nabla^2\tilde{\cL}(\bgamma^*)\left(\overline{\bgamma}-\bgamma^*\right).
    \end{align*}
    Combine the above two equations together, we have
    \begin{align}
        \int_{0}^1\nabla^2\tilde{\cL}(\bgamma^*+t(\widehat{\bgamma}-\bgamma^*)) - \nabla^2\tilde{\cL}(\bgamma^*)dt \left(\widehat{\bgamma}-\bgamma^*\right) = \nabla^2\tilde{\cL}(\bgamma^*)\left(\overline{\bgamma}-\widehat{\bgamma}\right).\label{gammaapproxeq1}
    \end{align}
    View $\nabla^2\tilde{\cL}$ as a sub-matrix of corresponding $\nabla^2\cL$, similar to  \eqref{residue2normeq7} we know that
    \begin{align}
        \left\|\nabla^2\tilde{\cL}(\bgamma^*+t(\widehat{\bgamma}-\bgamma^*)) - \nabla^2\tilde{\cL}(\bgamma^*)\right\|\lesssim np\left\|\widehat{\bgamma}-\bgamma^*\right\|_2.\label{gammaapproxeq2}
    \end{align}
    On the other hand, again by the property of sub-matrix as well as Lemma \ref{lb}, we know that 
    \begin{align}
        \left\|\nabla^2\tilde{\cL}(\bgamma^*)\left(\overline{\bgamma}-\widehat{\bgamma}\right)\right\|\gtrsim \frac{np}{\kappa_1}\left\|\overline{\bgamma}-\widehat{\bgamma}\right\|_2.\label{gammaapproxeq3}
    \end{align}
    Plugging \eqref{gammaapproxeq2} and \eqref{gammaapproxeq3} in \eqref{gammaapproxeq1} we get 
    \begin{align*}
        np \left\|\widehat{\bgamma}-\bgamma^*\right\|_2^2&\gtrsim \left\|\int_{0}^1\nabla^2\tilde{\cL}(\bgamma^*+t(\widehat{\bgamma}-\bgamma^*)) - \nabla^2\tilde{\cL}(\bgamma^*)dt \left(\widehat{\bgamma}-\bgamma^*\right)\right\|_2 = \left\|\nabla^2\tilde{\cL}(\bgamma^*)\left(\overline{\bgamma}-\widehat{\bgamma}\right)\right\|_2 \\
        & \gtrsim \frac{np}{\kappa_1}\left\|\overline{\bgamma}-\widehat{\bgamma}\right\|_2.
    \end{align*}
    As a result, we have
    \begin{align*}
        \left\|\overline{\bgamma}-\widehat{\bgamma}\right\|_2\lesssim \kappa_1^3 \frac{(k+d)\log n}{npL}.
    \end{align*}
\end{proof}

\subsubsection{Proof of Theorem \ref{thm:dist_two}}

\begin{proof}
We denote by $\tx_i' = \left((\boldsymbol{e}_i)_{\cS(\ba^*)}^\top, \bx_i^\top\right)^\top\in \mathbb{R}^{|\cS(\ba^*)|+d}$. Then we can write
\begin{align*}
    \overline{\bgamma} -\bgamma^* = - \left(\nabla^2\tilde{\cL}(\bgamma^*)\right)^{-1}\sum_{(i, j)\in \cE, i> j}(\phi(\tx_i^\top\tb^* - \tx_j^\top\tb^*)-y_{j, i})(\tx_i'^\top - \tx_j'^\top).
\end{align*}
Consider the random vector 
\begin{align*}
    \boldsymbol{X}_{i, j}^{(l)} = \frac{1}{L}(y_{j, i}^{(l)} - \phi(\tx_i^\top\tb^* - \tx_j^\top\tb^*)) \left(\nabla^2\tilde{\cL}(\bgamma^*)\right)^{-1} (\tx_i'^\top - \tx_j'^\top).
\end{align*}
Then we have 
\begin{align*}
    &\sum_{(i, j)\in \cE, i> j}\sum_{l=1}^L\mathbb{E}\left[\left\|\frac{1}{\sqrt{L}}(y_{j, i}^{(l)} - \phi(\tx_i^\top\tb^* - \tx_j^\top\tb^*))\left(\nabla^2\tilde{\cL}(\bgamma^*)\right)^{-1/2} (\tx_i'^\top - \tx_j'^\top)\right\|_2^3\right] \\
    = & \sum_{(i, j)\in \cE, i> j}\frac{1}{\sqrt{L}}\mathbb{E}\left[\left\|(y_{j, i}^{(l)} - \phi(\tx_i^\top\tb^* - \tx_j^\top\tb^*))\left(\nabla^2\tilde{\cL}(\bgamma^*)\right)^{-1/2} (\tx_i'^\top - \tx_j'^\top)\right\|_2^3\right] \\
    \lesssim & \sum_{(i, j)\in \cE, i> j}\mathbb{E}\left[\left\|(y_{j, i}^{(l)} - \phi(\tx_i^\top\tb^* - \tx_j^\top\tb^*))\left(\nabla^2\tilde{\cL}(\bgamma^*)\right)^{-1/2} (\tx_i'^\top - \tx_j'^\top)\right\|_2^2\right] \\
    & \quad \cdot \frac{\max_{i, j, l}\left\|\left(\nabla^2\tilde{\cL}(\bgamma^*)\right)^{-1/2} (\tx_i'^\top - \tx_j'^\top)\right\|_2}{\sqrt{L}} \\
    \lesssim & \frac{\max_{i, j, l}\left\|\left(\nabla^2\tilde{\cL}(\bgamma^*)\right)^{-1/2} (\tx_i'^\top - \tx_j'^\top)\right\|_2}{\sqrt{L}} \lesssim \sqrt{\frac{\kappa_1}{npL}}.
\end{align*}
By Berry-Esseen theorem we know that 
\begin{align}
    \left|\mathbb{P}(\overline{\bgamma}-\bgamma^* \in \mathcal{D}) - \mathbb{P}(\mathcal{N}(\boldsymbol{0}, (\nabla^2\tilde{\cL}(\bgamma^*))^{-1}) \in \mathcal{D})\right|\lesssim (k+d)^{1/4}\sqrt{\frac{\kappa_1}{npL}}.\label{jointdistributioneq3}
\end{align}

Next, for any convex set $\mathcal{D}\subset \mathbb{R}^r$ with $r=|\cS(\balpha^*)+d|,$ and point $x\in \mathbb{R}^r$, we define
\begin{align*}
    \delta_{\mathcal{D}}(x):= \begin{cases}-\min_{y\in \mathbb{R}^r \backslash \mathcal{D}}\left\|x- y\right\|_2, & \text { if } x \in \mathcal{D} \\ \min_{y\in\mathcal{D} }\left\|x- y\right\|_2, & \text { if } x \notin \mathcal{D}\end{cases} \text{ and } \mathcal{D}^{\varepsilon}:=\left\{x \in \mathbb{R}^r: \delta_{\mathcal{D}}(x) \leq \varepsilon\right\}.
\end{align*}
Therefore, we know that
\begin{align}
    &\mathbb{P}\left(\sqrt{L}(\nabla^2\tilde{\cL}(\bgamma^*))^{1/2}(\overline{\bgamma} - \bgamma^*)\in \mathcal{D}^{-\varepsilon}\right) \nonumber \\
    = &\mathbb{P}\left(\sqrt{L}(\nabla^2\tilde{\cL}(\bgamma^*))^{1/2}(\overline{\bgamma} - \bgamma^*)\in \mathcal{D}^{-\varepsilon}, \left\|\sqrt{L}(\nabla^2\tilde{\cL}(\bgamma^*))^{1/2}(\overline{\bgamma} - \widehat{\bgamma})\right\|_2\leq \varepsilon\right) \nonumber \\
    & +\mathbb{P}\left(\sqrt{L}(\nabla^2\tilde{\cL}(\bgamma^*))^{1/2}(\overline{\bgamma} - \bgamma^*)\in \mathcal{D}^{-\varepsilon}, \left\|\sqrt{L}(\nabla^2\tilde{\cL}(\bgamma^*))^{1/2}(\overline{\bgamma} - \widehat{\bgamma})\right\|_2> \varepsilon\right) \nonumber \\
    \leq & \mathbb{P}\left(\sqrt{L}(\nabla^2\tilde{\cL}(\bgamma^*))^{1/2}(\widehat{\bgamma} - \bgamma^*)\in \mathcal{D}\right) + \mathbb{P}\left( \left\|\sqrt{L}(\nabla^2\tilde{\cL}(\bgamma^*))^{1/2}(\overline{\bgamma} - \widehat{\bgamma})\right\|_2> \varepsilon\right). \label{jointdistributioneq1}
\end{align}
Taking $\varepsilon = \kappa_1^3 (k+d)\log n/\sqrt{npL}$, by Theorem \ref{gammaexpansionthm} we know that 
\begin{align}
    \mathbb{P}\left( \left\|\sqrt{L}(\nabla^2\tilde{\cL}(\bgamma^*))^{1/2}(\overline{\bgamma} - \widehat{\bgamma})\right\|_2> \varepsilon\right)\lesssim n^{-10}.\label{jointdistributioneq2}
\end{align}

On the other hand, we can write
\begin{align*}
    &\left|\mathbb{P}\left(\sqrt{L}(\nabla^2\tilde{\cL}(\bgamma^*))^{1/2}(\overline{\bgamma} - \bgamma^*)\in \mathcal{D}^{-\varepsilon}\right) - \mathbb{P}\left(\sqrt{L}(\nabla^2\tilde{\cL}(\bgamma^*))^{1/2}(\overline{\bgamma} - \bgamma^*)\in \mathcal{D}\right)\right| \\
    \leq & \left|\mathbb{P}\left(\sqrt{L}(\nabla^2\tilde{\cL}(\bgamma^*))^{1/2}(\overline{\bgamma} - \bgamma^*)\in \mathcal{D}^{-\varepsilon}\right) - \mathbb{P}\left(\mathcal{N}(\boldsymbol{0}, \boldsymbol{I})\in \mathcal{D}^{-\varepsilon}\right)\right|   \\
    & + \left|\mathbb{P}\left(\mathcal{N}(\boldsymbol{0}, \boldsymbol{I})\in \mathcal{D}^{-\varepsilon}\right) - \mathbb{P}\left(\mathcal{N}(\boldsymbol{0}, \boldsymbol{I})\in \mathcal{D}\right)\right| \\
    &+ \left|\mathbb{P}\left(\mathcal{N}(\boldsymbol{0}, \boldsymbol{I})\in \mathcal{D}\right) - \mathbb{P}\left(\sqrt{L}(\nabla^2\tilde{\cL}(\bgamma^*))^{1/2}(\overline{\bgamma} - \bgamma^*)\in \mathcal{D}\right) \right| \\
    \lesssim & (k+d)^{1/4}\sqrt{\frac{\kappa_1}{npL}} + \left|\mathbb{P}\left(\mathcal{N}(\boldsymbol{0}, \boldsymbol{I})\in \mathcal{D}^{-\varepsilon}\right) - \mathbb{P}\left(\mathcal{N}(\boldsymbol{0}, \boldsymbol{I})\in \mathcal{D}\right)\right|.
\end{align*}
By \cite[Theorem 1.2]{raivc2019multivariate}, it holds that
\begin{align*}
    |\mathbb{P}(\mathcal{N}(\boldsymbol{0}_r,\bI_r)\in \mathcal{D}^{-\varepsilon}) - \mathbb{P}(\mathcal{N}(\boldsymbol{0}_r,\bI_r)\in \mathcal{D})|\lesssim (k+d)^{1/4} \varepsilon  \lesssim \kappa_1^3 \frac{(k+d)^{5/4}\log n}{\sqrt{npL}}.
\end{align*}

As a result, we know that 
\begin{align*}
    \left|\mathbb{P}\left(\sqrt{L}(\nabla^2\tilde{\cL}(\bgamma^*))^{1/2}(\overline{\bgamma} - \bgamma^*)\in \mathcal{D}^{-\varepsilon}\right) - \mathbb{P}\left(\sqrt{L}(\nabla^2\tilde{\cL}(\bgamma^*))^{1/2}(\overline{\bgamma} - \bgamma^*)\in \mathcal{D}\right)\right| \lesssim \kappa_1^3 \frac{(k+d)^{5/4}\log n}{\sqrt{npL}}.
\end{align*}
Plugging this as well as \eqref{jointdistributioneq2} in \eqref{jointdistributioneq1}, we obtain 
\begin{align*}
    \mathbb{P}\left(\sqrt{L}(\nabla^2\tilde{\cL}(\bgamma^*))^{1/2}(\overline{\bgamma} - \bgamma^*)\in \mathcal{D}\right)\leq \mathbb{P}\left(\sqrt{L}(\nabla^2\tilde{\cL}(\bgamma^*))^{1/2}(\widehat{\bgamma} - \bgamma^*)\in \mathcal{D}\right) + O\left(\kappa_1^3 \frac{(k+d)^{5/4}\log n}{\sqrt{npL}}+\frac{1}{n^{10}}\right).
\end{align*}
Since this holds for all convex $\mathcal{D}\subset \mathbb{R}^{|\cS(\ba^*)+d|}$, we know that
\begin{align*}
    \mathbb{P}\left(\overline{\bgamma} - \bgamma^*\in \mathcal{D}\right)\leq \mathbb{P}\left(\widehat{\bgamma} - \bgamma^*\in \mathcal{D}\right) + O\left(\kappa_1^3 \frac{(k+d)^{5/4}\log n}{\sqrt{npL}}+\frac{1}{n^{10}}\right).
\end{align*}
Similarly, we can also show that
\begin{align*}
   \mathbb{P}\left(\widehat{\bgamma} - \bgamma^*\in \mathcal{D}\right)\leq  \mathbb{P}\left(\overline{\bgamma} - \bgamma^*\in \mathcal{D}\right) + O\left(\kappa_1^3 \frac{(k+d)^{5/4}\log n}{\sqrt{npL}}+\frac{1}{n^{10}}\right).
\end{align*}
Combine these two aforementioned inequalities with \eqref{jointdistributioneq3}, it holds that
\begin{align*}
    \left|\mathbb{P}(\widehat{\bgamma}-\bgamma^* \in \mathcal{D}) - \mathbb{P}(\mathcal{N}(\boldsymbol{0}_r, (\nabla^2\tilde{\cL}(\bgamma^*))^{-1}) \in \mathcal{D})\right|\lesssim \kappa_1^3 \frac{(k+d)^{5/4}\log n}{\sqrt{npL}}+\frac{1}{n^{10}}.
\end{align*}

\end{proof}

\section{Proof of the Results}
\subsection{Proof of Proposition \ref{proposition1}}\label{proposition1proof}
\begin{proof}
Assume we have two vectors $\tb_1 = (\ba_1^\top,\bb_1^\top)^\top,\tb_2= (\ba_2^\top,\bb_2^\top)^\top\in \Theta(k)$ such that 
\begin{align*}
    \mathbb{P}_{\tb_1}\{\text {item } j \text { is preferred over item } i\}=\mathbb{P}_{\tb_2}\{\text {item } j \text { is preferred over item } i\}  ,\quad \forall 1\leq i\neq j \leq n. 
\end{align*}
By \eqref{comparison} we know that
\begin{align*}
    \frac{e^{\tx_j^\top\tb_1}}{e^{\tx_i^\top\tb_1}+e^{\tx_j^\top\tb_1}} = \frac{1}{e^{\tx_i^\top\tb_1 - \tx_j^\top\tb_1}+1} = \frac{1}{e^{\tx_i^\top\tb_2 - \tx_j^\top\tb_2}+1} = \frac{e^{\tx_j^\top\tb_2}}{e^{\tx_i^\top\tb_2}+e^{\tx_j^\top\tb_2}},\quad \forall 1\leq i\neq j \leq n.
\end{align*}
This tells us that we have $(\tx_i-\tx_j)^\top(\tb_1-\tb_2)=0$ for all $1\leq i\neq j \leq n$. Consider the following index set
\begin{align*}
    A = \left\{i\in [n]: (\tb_1-\tb_2)_i = 0\right\}.
\end{align*}
Since $\tb_1,\tb_2\in \Theta(k)$, we know that $|A|\geq n-2k$. Since $2k+d+1\leq n$, we pick $d+1$ different indices $i_1,i_2,\dots,i_{d+1}$ from $A$. By the construction of $A$ we know that
\begin{align*}
    0=(\tx_{i_j}-\tx_{i_1})^\top(\tb_1-\tb_2) &= (\tb_1-\tb_2)_{i_j}-(\tb_1-\tb_2)_{i_1} +(\bx_{i_j}-\bx_{i_1})^\top(\bb_1-\bb_2) \\
    &= (\bx_{i_j}-\bx_{i_1})^\top(\bb_1-\bb_2)
\end{align*}
for all $j=2,3\dots, d+1$. On the other hand, according to Assumption \ref{nondegenerate}, we know that 
\begin{align*}
    \textbf{rank}[\bx_{i_2}-\bx_{i_1},\bx_{i_3}-\bx_{i_1},\dots,\bx_{i_{d+1}}-\bx_{i_1}] = d.
\end{align*}
As a result, we must that $\bb_1-\bb_2 = \boldsymbol{0}$. This further implies
\begin{align*}
    0 = (\tx_i-\tx_j)^\top(\tb_1-\tb_2) &= (\ba_1-\ba_2)_i-(\ba_1-\ba_2)_j + (\bx_i-\bx_j)^\top(\bb_1-\bb_2) \\
    &= (\ba_1-\ba_2)_i-(\ba_1-\ba_2)_j
\end{align*}
for all $1\leq i\neq j\leq n$. This tells us that all the entries of $\ba_1-\ba_2$ are the same. And, since $|A|\geq n-2k\geq d+1$, we know that at least $d+1$ entries of $\ba_1-\ba_2$ are $0$. As a result, we get $\ba_1-\ba_2 = \boldsymbol{0}$. To sum up, we must have $\tb_1 =\tb_2$.
\end{proof}

\subsection{Proof of Lemma \ref{lb}}\label{lbproof}
\begin{proof}
By \cite[Lemma A.4]{fan2022uncertainty} we know that 
\begin{align*}
    \lambda_{\text{min},\perp}(\nabla^2 \mathcal{L}_\tau(\tb''))\geq\lambda+\frac{c_2pn}{8\kappa_1e^C}.
\end{align*}
As a result, we know that
\begin{align}
    (\tb'-\tb)^\top\nabla^2 \mathcal{L}_\tau(\tb'')(\tb'-\tb)&\geq \left(\tau+\frac{c_2pn}{8\kappa_1 e^C}\right)\left\|\mathcal{P}\left(\tb'-\tb\right)\right\|_2^2 \nonumber \\
    &=\left(\tau+\frac{c_2pn}{8\kappa_1 e^C}\right)\left(\left\|(\boldsymbol{I}-\mathcal{P}_{\bar{\bX}})\left(\ba'-\ba\right)\right\|_2^2 + \left\|\bb'-\bb\right\|_2^2\right).\label{RSCeq1}
\end{align}
Let $S$ be the support of $\ba'-\ba$. Since $\tb,\tb'\in \Theta(k)$, we know that $|S|\leq 2k$. As a result, we have
\begin{align}
    \left\|(\boldsymbol{I}-\mathcal{P}_{\bar{\bX}})\left(\ba'-\ba\right)\right\|_2^2 &= \left\|\ba'-\ba\right\|_2^2 -\left\|\mathcal{P}_{\bar{\bX}}\left(\ba'-\ba\right)\right\|_2^2 =\left\|\ba'-\ba\right\|_2^2 -\sum_{i=1}^n((\mathcal{P}_{\bar{\bX}})_{i, S}(\ba'-\ba))^2 \nonumber \\
    &\geq \left\|\ba'-\ba\right\|_2^2 - \sum_{i=1}^n \left\|(\mathcal{P}_{\bar{\bX}})_{i, S}\right\|_2^2\left\|\ba'-\ba\right\|_2^2  = \left(1-\left\|(\mathcal{P}_{\bar{\bX}})_{\cdot, S}\right\|_F^2\right)\left\|\ba'-\ba\right\|_2^2 \nonumber \\
    &=\left(1-\left\|(\mathcal{P}_{\bar{\bX}})_{S,\cdot}\right\|_F^2\right)\left\|\ba'-\ba\right\|_2^2\geq \left(1-|S|\left\|(\mathcal{P}_{\bar{\bX}})_{S,\cdot}\right\|_{2,\infty}^2\right)\left\|\ba'-\ba\right\|_2^2  \nonumber \\
    &\geq \left(1-\frac{2 c_0^2 (d+1)k}{n}\right)\left\|\ba'-\ba\right\|_2^2\geq \frac{1}{2}\left\|\ba'-\ba\right\|_2^2.\label{RSCeq2}
\end{align}
Combine Eq.\eqref{RSCeq1} and Eq.\eqref{RSCeq2} we get
\begin{align*}
    (\tb'-\tb)^\top\nabla^2 \mathcal{L}_\tau(\tb'')(\tb'-\tb)&\geq \left(\tau+\frac{c_2pn}{8\kappa_1 e^C}\right)\left(\frac{1}{2}\left\|\ba'-\ba\right\|_2^2 + \left\|\bb'-\bb\right\|_2^2\right)  \\
    &\geq \frac{1}{2}\left(\tau+\frac{c_2pn}{8\kappa_1 e^C}\right)\left\|\tb'-\tb\right\|_2^2.
\end{align*}

\end{proof}

\subsection{Proof of Lemma \ref{lem2}}\label{lem2proof}

\begin{proof}
Since $\tb_R = \argmin \cL(\tb)+\lambda\|\ba\|_1$, we know that 
\begin{align*}
    -\left[\nabla\cL_\tau(\tb_R)\right]_{1:n}\in \partial \lambda\left\|\widehat\ba_R\right\|_1, \quad \left[\nabla\cL_\tau(\tb_R)\right]_{n+1:n+d} = 0.
\end{align*}
As a result, we know that $\textsf{SOFT}_{\eta\lambda}(\tb_R-\eta\nabla\cL_\tau(\tb_R)) = \tb_R$. As a result, we know that
\begin{align}
    \left\|\tb^{t+1}-\tb_R\right\|_2 &= \left\|\textsf{SOFT}_{\eta\lambda}\left(\tb^{t}-\eta\nabla\cL_\tau(\tb^{t})\right)-\textsf{SOFT}_{\eta\lambda}\left(\tb_R-\eta\nabla\cL_\tau(\tb_R)\right)\right\|_2  \nonumber\\
    &\leq \left\|\tb^{t}-\eta\nabla\cL_\tau(\tb^{t})-\left(\tb_R-\eta\nabla\cL_\tau(\tb_R)\right)\right\|_2. \label{pgd1eq1}
\end{align}
Consider $\tb(\gamma) = \tb_R+\gamma\left(\tb^t-\tb_R\right)$ for $\gamma\in [0,1]$. By the fundamental theorem of calculus we have
\begin{align}
     \tb^{t}-\eta\nabla\cL_\tau(\tb^{t})-\left(\tb_R-\eta\nabla\cL_\tau(\tb_R)\right)
     = \left\{\boldsymbol{I}_{n+d}-\eta\int^1_0\nabla^2\cL_\tau(\tb(\gamma))d\tau\right\}\left(\tb^t-\tb_R\right).\label{pgd1eq2}
\end{align}
Let $\displaystyle\boldsymbol{A} = \int^1_0\nabla^2\cL_\tau(\tb(\gamma))d\gamma$. By Lemma \ref{ub} and the definition of $\cL_\tau(\cdot)$ we know that $(\tau+0.5c_1np)\boldsymbol{I}_{n+d}\succeq\boldsymbol{A}\succeq\tau\boldsymbol{I}_{n+d}$. Therefore, it holds that
\begin{align}
    \left\Vert (\boldsymbol{I}_{n+d}-\eta\boldsymbol{A})(\tb^t-\tb_R)\right\Vert_2^2&=\left\|\tb^t-\tb_R\right\|_2^2-2\eta\left(\tb^t-\tb_R\right)^\top\boldsymbol{A}\left(\tb^t-\tb_R\right)+\eta^2\left(\tb^t-\tb_R\right)^\top\boldsymbol{A}^2\left(\tb^t-\tb_R\right) \nonumber \\
    &\leq \left(1-\eta\tau\right)^2\left\Vert \tb^t-\tb^*\right\Vert_2^2  \label{pgd1eq3}
\end{align}
Combine Eq.~\eqref{pgd1eq1}, Eq.~\eqref{pgd1eq2} with Eq.~\eqref{pgd1eq3}, we know that
\begin{align*}
    \left\|\tb^{t+1}-\tb_R\right\|_2\leq \left(1-\eta\tau\right)\left\Vert\tb^t-\tb^*\right\Vert_2 = \rho\left\Vert\tb^t-\tb^*\right\Vert_2.
\end{align*}
Therefore, under event $\mathcal{A}_2$, we have
\begin{align*}
    \left\|\tb^t-\tb_R\right\|_2\leq \rho^t\left\Vert\tb^0-\tb^*\right\Vert_2.
\end{align*}

\end{proof}

\subsection{Proof of Lemma \ref{lem3}}\label{lem3proof}
\begin{proof}
Since $\tb_R$ is the minimizer, we have that $\mathcal{L}_\tau(\tb^*)+\lambda\|\ba^*\|_1\geq \mathcal{L}_\tau(\tb_R)+\lambda\|\widehat\ba_R\|_1\geq \mathcal{L}_\tau(\tb_R)$. By the mean value theorem, for some $\tb'$ between $\tb^*$ and $\tb_R$, we have
\begin{align*}
    \mathcal{L}_\tau(\tb_R) 
     = \mathcal{L}_\tau(\tb^*)+ \nabla \mathcal{L}_\tau(\tb^*)^\top (\tb_R-\tb^*)+\frac{1}{2}(\tb_R-\tb^*)^\top \nabla^2 \mathcal{L}_\tau(\tb')(\tb_R-\tb^*).
\end{align*}
As a result, we have
\begin{align*}
    \mathcal{L}_\tau(\tb^*)+\lambda\left\|\ba^*\right\|_1&\geq \mathcal{L}_\tau(\tb^*)+ \nabla \mathcal{L}_\tau(\tb^*)^\top (\tb_R-\tb^*)+\frac{1}{2}(\tb_R-\tb^*)^\top \nabla^2 \mathcal{L}_\tau(\tb^*)(\tb_R-\tb^*) \\
    & \geq \mathcal{L}_\tau(\tb^*)+ \nabla \mathcal{L}_\tau(\tb^*)^\top (\tb_R-\tb^*)+\frac{\tau}{2}\left\Vert\tb_R-\tb^*\right\Vert_2^2.
\end{align*}
Therefore, we get
\begin{align*}
    \frac{\tau}{2}\left\Vert\tb_R-\tb^*\right\Vert_2^2 & \leq \lambda\left\|\ba^*\right\|_1-\nabla \mathcal{L}_\tau(\tb^*)^\top (\tb_R-\tb^*)\\
    &\leq \lambda\left\|\ba^*\right\|_1+\left\Vert \nabla \mathcal{L}_\tau(\tb^*)\right\Vert_2\left\Vert \tb_R-\tb^*\right\Vert_2.
\end{align*}
As a result, on event $\mathcal{A}_1$ we have
\begin{align*}
    \left\Vert \tb_R-\tb^*\right\Vert_2&\leq \frac{2\left\|\nabla \mathcal{L}_\tau(\tb^*)\right\|_2+\sqrt{2\tau\lambda\left\|\ba^*\right\|_1}}{\tau}\leq \frac{2C_0\sqrt{n^2p
    \log n/L}+\sqrt{2\tau\lambda\left\|\ba^*\right\|_1}}{\tau}  \\
    &\leq \frac{2C_0\sqrt{n}}{c_\tau}\max\left\{\frac{\kappa_2}{\kappa_1},\kappa_3\right\}+ \sqrt{\frac{2c_\lambda\sqrt{d+1}}{c_\tau}\max\left\{\kappa_2^2,\kappa_1\kappa_2\kappa_3\right\}}.
\end{align*}
We conclude the proof of Lemma \ref{lem3}.
\end{proof}

\subsection{Proof of Lemma \ref{lem5}}\label{lem5proof}
\begin{proof}
Combine Lemma \ref{lem2} and Lemma \ref{lem3} we have
\begin{align*}
\max\left\{\left\Vert \tb^{T-1}-\tb_R\right\Vert_2, \left\Vert \tb^T-\tb_R\right\Vert_2\right\}&\leq \rho^{T-1}\left\Vert \tb^0-\tb_R\right\Vert_2 \\
&\lesssim \left(1-\frac{2\tau}{2\tau+c_1 np}\right)^{n^5-1}n\\
&\leq n\exp\left(-\frac{2\tau (n^5-1)}{2\tau+c_1 np}\right) \\
&\leq C_7 \kappa_1\sqrt{\frac{(d+1)\log n}{npL}}
\end{align*}
for $L \leq c_4\cdot n^{c_5}$ and $n$ which is large enough.
\end{proof}

\subsection{Proof of Lemma \ref{induction1}}\label{induction1proof}
\begin{proof}
By definition we know that 
\begin{align*}
\tb^{t+1}-\tb^* &= \textsf{SOFT}_{\eta\lambda}\left(\tb^t-\eta\nabla\mathcal{L}_\tau(\tb^t)\right)-\tb^* \\
&=\textsf{SOFT}_{\eta\lambda}\left(\tb^t-\eta\nabla\mathcal{L}_\tau(\tb^t)\right) - \textsf{SOFT}_{\eta\lambda}\left(\tb^*\right)+ \textsf{SOFT}_{\eta\lambda}\left(\tb^*\right)-\tb^*.
\end{align*}
By triangle inequality as well as the definition of $\textsf{SOFT}$ we know that
\begin{align}
    \left\|\tb^{t+1}-\tb^*\right\|_2&\leq \left\|\textsf{SOFT}_{\eta\lambda}\left(\tb^t-\eta\nabla\mathcal{L}_\tau(\tb^t)\right) - \textsf{SOFT}_{\eta\lambda}\left(\tb^*\right)\right\|_2 +\left\|\textsf{SOFT}_{\eta\lambda}\left(\tb^*\right)-\tb^*\right\|_2 \nonumber\\
    &\leq \left\|\tb^t-\eta\nabla\mathcal{L}_\tau(\tb^t)-\tb^*\right\|_2+\eta\lambda\sqrt{k}.\label{inductionAeq1}
\end{align}
Consider $\tb(\gamma) = \tb^*+\gamma\left(\tb^t-\tb^*\right)$ for $\gamma\in [0,1]$. By the fundamental theorem of calculus, we have
\begin{align*}
     \tb^t-\eta\nabla\mathcal{L}_\tau(\tb^t)-\tb^* 
    & = \tb^t-\eta\nabla\mathcal{L}_\tau(\tb^t)-\left[\tb^*-\eta\nabla\mathcal{L}_\tau(\tb^*)\right]-\eta\nabla\mathcal{L}_\tau(\tb^*) \\
    & = \left\{\boldsymbol{I}_{n+d}-\eta\int^1_0\nabla^2\cL_\tau(\tb(\gamma))d\tau\right\}\left(\tb^t-\tb^*\right)-\eta \nabla\cL_\tau(\tb^*).
\end{align*}
Let $npL$ be large enough such that 
\begin{align*}
    2C_6\kappa_1^2\sqrt{\frac{(d+1)\log n}{npL}}\leq 0.1,\quad 2C_3\kappa_1\sqrt{c_3}\sqrt{\frac{(d+1)\log n}{npL}}\leq 0.1.
\end{align*}
By the assumption of induction, we have 
\begin{align*}
    \Vert \ba(\gamma)-\ba^*\Vert_\infty\leq 0.05,\quad \Vert \bb(\gamma)-\bb^*\Vert_2\leq 0.05\sqrt{\frac{n}{c_3(d+1)}}. 
\end{align*}
Then by Lemma \ref{lb} as well as the induction assumption, we have
\begin{align*}
    \left(\tb^t-\tb^*\right)^\top\nabla^2\cL_\tau(\tb(\gamma))\left(\tb^t-\tb^*\right)\geq \left(\tau+\frac{c_2pn}{8\kappa_1e^{0.2}}\right)\left\|\tb^t-\tb^*\right\|_2^2\geq \left(\tau+\frac{c_2pn}{10\kappa_1}\right)\left\|\tb^t-\tb^*\right\|_2^2
\end{align*}
for all $0\leq \gamma\leq 1$. On the other hand, by Lemma \ref{ub}, we have
\begin{align*}
    \lambda_{\text{max}}\left(\nabla^2\cL_\tau(\tb(\gamma))\right)\leq \tau+\frac{1}{2}c_1pn.
\end{align*}
Let $\displaystyle\boldsymbol{A} = \int^1_0\nabla^2\cL_\tau(\tb(\gamma))d\gamma$, then it holds that
\begin{align}
    \left\Vert (\boldsymbol{I}_{n+d}-\eta\boldsymbol{A})(\tb^t-\tb^*)\right\Vert_2^2&=\left\|\tb^t-\tb^*\right\|_2^2-2\eta\left(\tb^t-\tb^*\right)^\top\boldsymbol{A}\left(\tb^t-\tb^*\right)+\eta^2\left(\tb^t-\tb^*\right)^\top\boldsymbol{A}^2\left(\tb^t-\tb^*\right) \nonumber \\
    &\leq \left(1-2\eta\left(\tau+\frac{c_2pn}{10\kappa_1}\right)+\eta^2\left(\tau+\frac{1}{2}c_1pn\right)^2\right)\left\Vert \tb^t-\tb^*\right\Vert_2^2 \nonumber \\
    &\leq \left(1-\frac{c_2}{20\kappa_1}\eta pn\right)^2\left\Vert\tb^t-\tb^*\right\Vert_2^2.\label{inductionAeq2}
\end{align}
Therefore, plugging Eq.\eqref{inductionAeq1} in Eq.~\eqref{inductionAeq2}, and conditioned on event $\mathcal{A}_1$, we have
\begin{align*}
    \left\Vert\tb^{t+1}-\tb^*\right\Vert_2&\leq \left\Vert \left(\boldsymbol{I}_{n+d}-\eta\boldsymbol{A}\right)(\tb^t-\tb^*)-\eta \nabla\cL_\tau(\tb^*)\right\Vert_2 +\eta\lambda\sqrt{k}  \\
    &\leq \left\Vert \left(\boldsymbol{I}_{n+d}-\eta\boldsymbol{A}\right)(\tb^t-\tb^*)\right\Vert_2+\eta \left\Vert\nabla\cL_\tau(\tb^*)\right\Vert_2 +\eta\lambda\sqrt{k} \\
    &\leq \left(1-\frac{c_2}{20\kappa_1}\eta pn\right)\left\Vert\tb^t-\tb^*\right\Vert_2 +C_0\eta \sqrt{\frac{n^2p\log n}{L}} +\eta\lambda\sqrt{k}\\
    &\leq \left(1-\frac{c_2}{20\kappa_1}\eta pn\right)C_3 \kappa_1\sqrt{\frac{\log n}{pL}} +C_0\eta \sqrt{\frac{n^2p\log n}{L}}+\eta c_\lambda\kappa_1\sqrt{\frac{k(d+1)np\log n}{L}} \\
    &\leq C_3 \kappa_1\sqrt{\frac{\log n}{pL}},
\end{align*}
as long as $\displaystyle C_3\geq \frac{40C_0}{c_2}$ and $\displaystyle k(d+1)\leq \frac{c_2^2 C_3^2 n}{1600 c_\lambda^2\kappa_1^2}$.
\end{proof}

\subsection{Proof of Lemma \ref{induction3}}\label{induction3proof}
\begin{proof}
For any $m\in[n]$, by definition we have
\begin{align*}
    \tb^{t+1}-\tb^{t+1,(m)} &=  \textsf{SOFT}_{\eta\lambda}\left(\tb^t-\eta\nabla\mathcal{L}_\tau(\tb^t)\right)-\textsf{SOFT}_{\eta\lambda}\left(\tb^{t,(m)}-\eta\nabla\mathcal{L}_\tau^{(m)}(\tb^{t,(m)})\right).
\end{align*}
This implies
\begin{align}
    \left\|\tb^{t+1}-\tb^{t+1,(m)}\right\|_2\leq \left\|\tb^t-\eta\nabla\mathcal{L}_\tau(\tb^t)-\left[\tb^{t,(m)}-\eta\nabla\mathcal{L}_\tau^{(m)}(\tb^{t,(m)})\right]\right\|_2. \label{induction3eq1}
\end{align}
We consider $\tb(\tau) = \tb^{t,(m)}+\gamma\left(\tb^t-\tb^{t,(m)}\right)$ for $\gamma\in [0,1]$. By the fundamental theorem of calculus we have
\begin{align}
    & \tb^t-\eta\nabla\mathcal{L}_\tau(\tb^t)-\left[\tb^{t,(m)}-\eta\nabla\mathcal{L}_\tau^{(m)}(\tb^{t,(m)})\right]  \nonumber \\
    =& \tb^t-\eta\nabla\mathcal{L}_\tau(\tb^t)-\left[\tb^{t,(m)}-\eta\nabla\mathcal{L}_\tau(\tb^{t,(m)})\right]-\eta \left(\nabla\mathcal{L}_\tau(\tb^{t,(m)})-\nabla\mathcal{L}_\tau^{(m)}(\tb^{t,(m)})\right) \nonumber \\
    =& \left(\boldsymbol{I}_{n+d}-\eta\int^1_0\nabla^2\cL_\tau(\tb(\gamma))d\gamma\right)\left(\tb^{t}-\tb^{t,(m)}\right)-\eta \left(\nabla\mathcal{L}_\tau(\tb^{t,(m)})-\nabla\mathcal{L}_\tau^{(m)}(\tb^{t,(m)})\right).\label{induction3decompose}
\end{align}
From \eqref{inductionA}$\sim$ \eqref{inductionD} we know that
\begin{align*}
    \Vert \ba^{t,(m)}-\ba^*\Vert_{\infty}&\leq \Vert \ba^{t}-\ba^*\Vert_{\infty}+\max_{1\leq m\leq n}\left\Vert \tb^{t,(m)}-\tb^t\right\Vert_2 \leq (C_4+C_6)\kappa_1^2\sqrt{\frac{(d+1)\log n}{npL}}; \\
    \Vert \bb^{t,(m)}-\bb^*\Vert_{2}&\leq \left\Vert \tb^{t}-\tb^*\right\Vert_{2}+\max_{1\leq m\leq n}\left\Vert \tb^{t,(m)}-\tb^t\right\Vert_2 \leq  (C_3+C_4)\kappa_1\sqrt{\frac{\log n}{pL}};\\
    \Vert \ba^{t}-\ba^*\Vert_{\infty}&\leq C_6\kappa_1^2\sqrt{\frac{(d+1)\log n}{npL}}; \\
    \Vert \bb^{t}-\bb^*\Vert_{2}&\leq \left\Vert \tb^{t}-\tb^*\right\Vert_{2}\leq C_3\kappa_1\sqrt{\frac{\log n}{pL}}.
\end{align*}
Consider $npL$ which is large enough such that
\begin{align*}
     2(C_4+C_6)\kappa_1^2\sqrt{\frac{(d+1)\log n}{npL}},\; 2(C_3+C_4)\kappa_1\sqrt{c_3}\sqrt{\frac{(d+1)\log n}{npL}}\leq 0.1.
\end{align*}
Then we also have
\begin{align*}
    2C_6\kappa_1^2\sqrt{\frac{(d+1)\log n}{npL}},\;2C_3\kappa_1\sqrt{c_3}\sqrt{\frac{(d+1)\log n}{npL}}\leq 0.1.
\end{align*}
Use the same approach when deriving Eq.~\eqref{inductionAeq2}, we have
\begin{align}
    \left\|\left(\boldsymbol{I}_{n+d}-\eta\int^1_0\nabla^2\cL_\tau(\tb(\gamma))d\gamma\right)\left(\tb^{t}-\tb^{t,(m)}\right)\right\|_2\leq \left(1-\frac{c_2}{20\kappa_1}\eta pn\right)\left\Vert\tb^t-\tb^{t,(m)}\right\Vert_2,\label{induction3sc}
\end{align}
as long as $\displaystyle 0<\eta\leq \frac{2}{2\lambda+c_1np}$.

It remains to bound $\left\|  \nabla\mathcal{L}_\tau(\tb^{t,(m)})-\nabla\mathcal{L}_\tau^{(m)}(\tb^{t,(m)})\right\|_2.$ By definition, we have
\begin{align*}
    &\nabla\mathcal{L}_\tau(\tb^{t,(m)})-\nabla\mathcal{L}_\tau^{(m)}(\tb^{t,(m)}) \\
    =& \sum_{i\neq m}\left\{\left(-y_{m,i}+\frac{e^{  \tx_i^\top\tb^{t,(m)}}}{e^{\tx_i^\top\tb^{t,(m)}}+e^{\tx_m^\top\tb^{t,(m)}}} \right)\textbf{1}((i,m)\in\mathcal{E})-p\left(-y_{m,i}^*+\frac{e^{  \tx_i^\top\tb^{t,(m)}}}{e^{\tx_i^\top\tb^{t,(m)}}+e^{\tx_m^\top\tb^{t,(m)}}}\right)\right\}(\tx_i-\tx_m)\\
    =&\underbrace{\sum_{i\neq m}\left\{\left(-\frac{e^{  \tx_i^\top\tb^{*}}}{e^{\tx_i^\top\tb^*}+e^{\tx_m^\top\tb^*}}+\frac{e^{  \tx_i^\top\tb^{t,(m)}}}{e^{\tx_i^\top\tb^{t,(m)}}+e^{\tx_m^\top\tb^{t,(m)}}} \right)\left(\textbf{1}((i,m)\in\mathcal{E})-p\right)\right\}(\tx_i-\tx_m)}_{:=\boldsymbol{u}^m} \\
    &+\underbrace{\frac{1}{L}\sum_{(i,m)\in\mathcal{E}}\sum^L_{l=1}\left(-y^{(l)}_{m,i}+\frac{e^{  \tx_i^\top\tb^*}}{e^{\tx_i^\top\tb^*}+e^{\tx_m^\top\tb^*}}\right)(\tx_i-\tx_m)}_{:=\boldsymbol{v}^m}.
\end{align*}

By definition, we also have
\begin{align*}
    v_{j}^{m}= \begin{cases}\frac{1}{L} \sum_{l=1}^{L}\left(-y_{m, j}^{(l)}+\frac{e^{  \tx_j^\top\tb^*}}{e^{\tx_j^\top\tb^*}+e^{\tx_m^\top\tb^*}}\right), & \text { if }(j, m) \in \mathcal{E} \\ \frac{1}{L} \sum_{i:(i, m) \in \mathcal{E}} \sum_{l=1}^{L}\left(y_{m, i}^{(l)}-\frac{e^{  \tx_i^\top\tb^*}}{e^{\tx_i^\top\tb^*}+e^{\tx_m^\top\tb^*}}\right), & \text { if } j=m ; \\ 
    \frac{1}{L} \sum_{i:(i, m) \in \mathcal{E}} \sum_{l=1}^{L}\left(-y_{m, i}^{(l)}+\frac{e^{  \tx_i^\top\tb^*}}{e^{\tx_i^\top\tb^*}+e^{\tx_m^\top\tb^*}}\right)((\tx_i)_j-(\tx_m)_j), & \text { if } j>n ; \\ 
    0, & \text { else. }\end{cases}
\end{align*}
Consider random variable $M = |\left\{i:(i,m)\in \mathcal{E}\right\}|$. By Chernoff bound \citep{tropp2012user}, we know that
\begin{align*}
\mathbb{P}(M\geq 2pn)\leq (e/4)^{pn}\leq O(n^{-11}),
\end{align*}
as long as $np>c_p\log n$ for some $c_p>0$. As long as $\Vert \bx_i-\bx_m\Vert_2\leq 2\sqrt{c_3(d+1)/n}\leq 1$, we have $|(\tx_i)_j-(\tx_m)_j)|\leq 1$ for $j>n$. Since $\left|-y^{(l)}_{m,i}+\frac{e^{  \tx_i^\top\tb^*}}{e^{\tx_i^\top\tb^*}+e^{\tx_m^\top\tb^*}}\right|\leq 1$, by Hoeffding's inequality and union bound, we get
\begin{align*}
    |v_j^m|&\lesssim \sqrt{\frac{M\log n}{L}}, \text{ if } j=m \text{ or } j>n; \\
    |v_j^m|&\lesssim \sqrt{\frac{\log n}{L}}, \text{ if } (j,m)\in\mathcal{E}.
\end{align*}
with probability exceeding $1-O(n^{-11})$ conditioning on $\mathcal{E}$ as long as $d< n$. On the other hand, since $M\leq 2pn$ with probability exceeding $1-O(n^{-11})$, we have
\begin{align*}
    \Vert\boldsymbol{v}^m\Vert_2^2\lesssim (d+1)\frac{2pn\log n}{L} +2pn\frac{\log n}{L}\lesssim \frac{pn(d+1)\log n}{L}
\end{align*}
with probability exceeding $1-O(n^{-11})$. 

On the other hand, for $\boldsymbol{u}^m$ we have
\begin{align*}
    u_{j}^{m}= \begin{cases}\xi_j(1-p), & \text { if }(j, m) \in \mathcal{E} \\ 
    -\sum_{i:(i, m) \in \mathcal{E}} \xi_i\left(\textbf{1}((i,m)\in\mathcal{E})-p\right), & \text { if } j=m ; \\ 
     \sum_{i:(i, m) \in \mathcal{E}} \xi_i\left(\textbf{1}((i,m)\in\mathcal{E})-p\right)((\tx_i)_j-(\tx_m)_j), & \text { if } j>n ; \\ 
    -\xi_j p, & \text { else, }\end{cases}
\end{align*}
where $$\displaystyle \xi_j = -\frac{e^{  \tx_j^\top\tb^{*}}}{e^{\tx_j^\top\tb^*}+e^{\tx_m^\top\tb^*}}+\frac{e^{  \tx_j^\top\tb^{t,(m)}}}{e^{\tx_j^\top\tb^{t,(m)}}+e^{\tx_m^\top\tb^{t,(m)}}} = -\frac{1}{1+e^{\tx_m^\top\tb^*-\tx_j^\top\tb^*}}+\frac{1}{1+e^{\tx_m^\top\tb^{t,(m)}-\tx_j^\top\tb^{t,(m)}}}.$$
Consider $\displaystyle g(x) = \frac{1}{1+e^x}$. Since $\displaystyle\left|g'(x)\right|\leq 1$, we have that
\begin{align*}
    |\xi_j| &= \left|g(\tx_m^\top\tb^{t,(m)}-\tx_j^\top\tb^{t,(m)})-g(\tx_m^\top\tb^*-\tx_j^\top\tb^*)\right| \\
    &\leq \left|(\tx_m^\top\tb^{t,(m)}-\tx_j^\top\tb^{t,(m)})-(\tx_m^\top\tb^*-\tx_j^\top\tb^*)\right|\\
    &\leq \left|\tx_m^\top\tb^{t,(m)}-\tx_m^\top\tb^*\right|+\left|\tx_j^\top\tb^{t,(m)}-\tx_j^\top\tb^*\right| \\
    &\leq \left|\alpha_m^{t,(m)}-\alpha_m^*\right|+\left|\bx_m^\top\bb^{t,(m)}-\bx_m^\top\bb^*\right|+\left|\alpha_j^{t,(m)}-\alpha_j^*\right|+\left|\bx_j^\top\bb^{t,(m)}-\bx_j^\top\bb^*\right| \\
    &\leq 2 \left\|\ba^{t,(m)}-\ba^*\right\|_\infty + 2\sqrt{c_3(d+1)/n}\left\|\bb^{t,(m)}-\bb^*\right\|_2 \\
    &\leq \left[2(C_4+C_6)\kappa_1^2+2(C_3+C_4)\kappa_1\sqrt{c_3}\right]\sqrt{\frac{(d+1)\log n}{npL}} := \tC_1\sqrt{\frac{(d+1)\log n}{npL}}.
\end{align*}
By Bernstein inequality we know that
\begin{align*}
    |u_j^m|&\lesssim \sqrt{\left(p\sum^n_{i=1}\xi_i^2\right)\log n}+\max_{1\leq i\leq n}|\xi_i|\log n\\
    &\leq \left(\sqrt{np\log n}+\log n\right)\tC_1\sqrt{\frac{(d+1)\log n}{npL}},\text{ if } j=m \text{ or } j>n.
\end{align*}
As a result, for $\boldsymbol{u}^m$ we have
\begin{align*}
    \Vert \boldsymbol{u}^m\Vert_2^2 &= (u^m_m)^2 +\sum_{j>n}(u^m_j)^2+\sum_{j:(j,m)\in \mathcal{E}}(u^m_j)^2 +\sum_{j:(j,m)\notin \mathcal{E},j\neq m, j\leq n}(u^m_j)^2 \\
    &\lesssim (d+1)\left(\sqrt{np\log n}+\log n\right)^2\tC_1^2\frac{(d+1)\log n}{npL}+np\tC_1^2\frac{(d+1)\log n}{npL}+p^2n\tC_1^2\frac{(d+1)\log n}{npL} \\
    &\lesssim pn(d+1)\log n \tC_1^2\frac{(d+1)\log n}{npL}.
\end{align*}
In summary,  there exists constants $D_1, D_2$ which are independent of $C_i, i\geq 0$ such that
\begin{align}
     \Vert\boldsymbol{v}^m\Vert_2\leq D_1\sqrt{\frac{pn(d+1)\log n}{L}},\quad \Vert\boldsymbol{u}^m\Vert_2\leq D_2\tC_1 (d+1)\log n\sqrt{\frac{1}{L}}\label{induction3uv}
\end{align}
with probability exceeding $1-O(n^{-11})$. Plugging   Eq.~\eqref{induction3decompose}, Eq.~\eqref{induction3sc} and Eq.~\eqref{induction3uv} in Eq.~\eqref{induction3eq1} we have
\begin{align}
    \left\Vert \tb^{t+1}-\tb^{t+1,(m)}\right\Vert_2\leq& \left\|\tb^t-\eta\nabla\mathcal{L}_\tau(\tb^t)-\left[\tb^{t,(m)}-\eta\nabla\mathcal{L}_\tau^{(m)}(\tb^{t,(m)})\right]\right\|_2 \label{induction3eq2}\\
    \leq &\left(1-\frac{c_2}{20\kappa_1}\eta pn\right)\left\Vert\tb^t-\tb^{t,(m)}\right\Vert_2   \nonumber \\
    &+ \eta \left(D_1\sqrt{\frac{pn(d+1)\log n}{L}}+D_2\tC_1 (d+1)\log n\sqrt{\frac{1}{L}}\right) \nonumber \\
    \leq& \left(1-\frac{c_2}{20\kappa_1}\eta pn\right)C_4\kappa_1\sqrt{\frac{(d+1)\log n}{npL}} \nonumber \\
    &+ \eta \left(D_1\sqrt{\frac{pn(d+1)\log n}{L}}+D_2\tC_1 (d+1)\log n\sqrt{\frac{1}{L}}\right) \nonumber \\
    \leq& C_4\kappa_1\sqrt{\frac{(d+1)\log n}{npL}},\label{induction3eq3}
\end{align}
as long as $\displaystyle C_4\geq \frac{40D_1}{c_2}$ and $n$ is large enough such that $\displaystyle C_4\geq \frac{40D_2}{c_2}\tC_1\sqrt{\frac{(d+1)\log n}{np}}$.

\end{proof}

\subsection{Proof of Lemma \ref{induction2}}\label{induction2proof}
\begin{proof}
For $m\in[n]$, we have
\begin{align*}
     \alpha_m^{t+1,(m)}-\alpha_m^{*} &= s\left(\alpha_m^{t,(m)}-\eta\left[\nabla\cL_\tau^{(m)}\left(\tb^{t,(m)}\right)\right]_m,\eta\lambda\right)-\alpha_m^* \\
     &=s\left(\alpha_m^{t,(m)}-\eta\left[\nabla\cL_\tau^{(m)}\left(\tb^{t,(m)}\right)\right]_m,\eta\lambda\right) -s(\alpha_m^*)+s(\alpha_m^*)-\alpha_m^*.
\end{align*}
According to the induction assumption, we know that $\cS(\ba^{t,(m)})\subset \cS(\ba^*)$. As a result, we have 
\begin{align}
    \left|\alpha_m^{t+1,(m)}-\alpha_m^{*}\right| \leq
    \begin{cases}
      \left|s\left(\eta\left[\nabla\cL_\tau^{(m)}\left(\tb^{t,(m)}\right)\right]_m,\eta\lambda\right)\right|, \quad & m\notin \cS(\ba^*) \\
      \left|\alpha_m^{t,(m)}-\eta\left[\nabla\cL_\tau^{(m)}\left(\tb^{t,(m)}\right)\right]_m-\alpha_m^*\right|+\eta\lambda , \quad & m\in \cS(\ba^*).\label{induction2eq1}
    \end{cases}  
\end{align}

First, when $m\in \cS(\ba^*)$, we have
\begin{align}
    & \alpha_m^{t,(m)}-\eta\left[\nabla\cL_\tau^{(m)}\left(\tb^{t,(m)}\right)\right]_m-\alpha_m^* \nonumber \\
    =& \alpha_m^{t,(m)}-\alpha_m^*-\eta p\sum_{i\neq m}\left\{\frac{e^{\tx_i^\top\tb^*}}{e^{\tx_i^\top\tb^*}+e^{\tx_m^\top\tb^*}}-\frac{e^{\tx_i^\top\tb^{t, (m)}}}{e^{\tx_i^\top\tb^{t, (m)}}+e^{\tx_m^\top\tb^{t, (m)}}}\right\}-\eta\tau\alpha_m^{t,(m)}. \label{eqa}
\end{align}
Normalizing the numerators below to 1 and by the mean value theorem,  there exists some $c_i$ between $\tx_m^\top\tb^* - \tx_i^\top\tb^*$ and $\tx_m^\top\tb^{t, (m)} - \tx_i^\top\tb^{t, (m)}$ such that
\begin{align}
    \frac{e^{\tx_i^\top\tb^*}}{e^{\tx_i^\top\tb^*}+e^{\tx_m^\top\tb^*}}-\frac{e^{\tx_i^\top\tb^{t, (m)}}}{e^{\tx_i^\top\tb^{t, (m)}}+e^{\tx_m^\top\tb^{t, (m)}}} =& -\frac{e^{c_i}}{(1+e^{c_i})^2}\left[\tx_m^\top\tb^* - \tx_i^\top\tb^*-\tx_m^\top\tb^{t, (m)} + \tx_i^\top\tb^{t, (m)}\right] \nonumber \\
    =&-\frac{e^{c_i}}{(1+e^{c_i})^2}\left[\alpha_m^* - \alpha_i^*-\alpha_m^{t,(m)}+\alpha_i^{t,(m)}\right]  \nonumber \\ &\;-\frac{e^{c_i}}{(1+e^{c_i})^2}\left[\bx_m^\top\bb^* - \bx_i^\top\bb^*-\bx_m^\top\bb^{t, (m)} + \bx_i^\top\bb^{t, (m)}\right].
    \label{eqb}
\end{align}
Combining Eq.~\eqref{eqa} and Eq.~\eqref{eqb}, we have
\begin{align*}
    &\alpha_m^{t,(m)}-\eta\left[\nabla\cL_\tau^{(m)}\left(\tb^{t,(m)}\right)\right]_m-\alpha_m^* \\
    =& \left(1-\eta p\sum_{i\neq m}\frac{e^{c_i}}{(1+e^{c_i})^2}\right)\left(\alpha_m^{t,(m)}-\alpha_m^{*}\right)\\
    &+\eta p\sum_{i\neq m}\frac{e^{c_i}}{(1+e^{c_i})^2}\left[\alpha_i^{t,(m)}-\alpha_i^{*}+(\bx_m-\bx_i)^\top\left(\bb^*-\bb^{t,(m)}\right)\right]-\eta\tau\alpha_m^{t,(m)} \\
    =& \left(1-\eta\lambda-\eta p\sum_{i\neq m}\frac{e^{c_i}}{(1+e^{c_i})^2}\right)\left(\alpha_m^{t,(m)}-\alpha_m^{*}\right)\\
    &+\eta p\sum_{i\neq m}\frac{e^{c_i}}{(1+e^{c_i})^2}\left[\alpha_i^{t,(m)}-\alpha_i^{*}+(\bx_m-\bx_i)^\top\left(\bb^*-\bb^{t,(m)}\right)\right]-\eta\tau\alpha_m^{*}.
\end{align*}
By taking absolute value on both side, we get
\begin{align*}
    &\left|\alpha_m^{t,(m)}-\eta\left[\nabla\cL_\tau^{(m)}\left(\tb^{t,(m)}\right)\right]_m-\alpha_m^*\right| \\
    \leq& \left| 1-\eta\lambda-\eta p\sum_{i\neq m}\frac{e^{c_i}}{(1+e^{c_i})^2}\right||\alpha_m^{t,(m)}-\alpha_m^{*}| \\
    &+ \frac{\eta p}{4} \sum_{i\neq m}\left[|\alpha_i^{t,(m)}-\alpha_i^{*}|+\Vert\bx_m-\bx_i\Vert_2\Vert\bb^*-\bb^{t,(m)}\Vert_2\right]+\eta\tau |\alpha_m^*| \\
    \leq& \left| 1-\eta\lambda-\eta p\sum_{i\neq m}\frac{e^{c_i}}{(1+e^{c_i})^2}\right||\alpha_m^{t,(m)}-\alpha_m^{*}| \\
    &+ \frac{\eta p}{4}\left[\sqrt{n}\|\ba^{t,(m)}-\ba^{*}\|_2+n\cdot 2\sqrt{\frac{c_3(d+1)}{n}}\Vert\bb^*-\bb^{t,(m)}\Vert_2\right]+\eta\tau \|\ba^*\|_{\infty}\\
    \leq& \left| 1-\eta\lambda-\eta p\sum_{i\neq m}\frac{e^{c_i}}{(1+e^{c_i})^2}\right||\alpha_m^{t,(m)}-\alpha_m^{*}| \\
    &+ \frac{\eta p}{4}\sqrt{n}\left(1+2\sqrt{c_3(d+1)}\right)\left\Vert\tb^*-\tb^{t,(m)}\right\Vert_2+\eta\tau \|\ba^*\|_{\infty}.
\end{align*}
Since $\displaystyle 1-\eta\lambda-\eta p\sum_{i\neq m}\frac{e^{c_1}}{(1+e^{c_i})^2}\geq 1-\eta\lambda-\eta p\frac{n}{4}\geq 0$, we have
\begin{align*}
    \left| 1-\eta\lambda-\eta p\sum_{i\neq m}\frac{e^{c_i}}{(1+e^{c_i})^2}\right| &= 1-\eta\lambda-\eta p\sum_{i\neq m}\frac{e^{c_i}}{(1+e^{c_i})^2} \\
    &\leq 1-\eta p(n-1)\min_{i\neq m}\frac{e^{c_i}}{(1+e^{c_i})^2}.
\end{align*}
By the defintion of $c_i,$ we have
\begin{align*}
    \max_{i\neq m}|c_i|\leq& \max_{i\neq m} |\tx_m^\top\tb^* - \tx_i^\top\tb^*|+\max_{i\neq m}\left|\tx_m^\top\tb^* - \tx_i^\top\tb^*-\tx_m^\top\tb^{t, (m)} + \tx_i^\top\tb^{t, (m)}\right| \\
    \leq& \log \kappa_1+\max_{i\neq m}|\alpha^*_m-\alpha_i^*-\alpha_m^{t,(m)}+\alpha_i^{t,(m)}|+\max_{i\neq m}\left|(\bx_m-\bx_i)^\top\left(\bb^*-\bb^{t, (m)}\right)\right| \\
    \leq& \log \kappa_1+2\Vert \ba^{t,(m)}-\ba^*\Vert_{\infty}+2\sqrt{\frac{c_3(d+1)}{n}}\Vert \bb^*-\bb^{t, (m)}\Vert_2.
\end{align*}
Consider $npL$ which is large enough such that
\begin{align*}
    2(C_4+C_6)\kappa_1^2\sqrt{\frac{(d+1)\log n}{npL}}\leq 0.1,\quad 2(C_3+C_4)\kappa_1\sqrt{c_3}\sqrt{\frac{(d+1)\log n}{npL}}\leq 0.1.
\end{align*}
Then we have
\begin{align*}
    \max_{i\neq m}|c_i|\leq& \log \kappa_1+2\Vert \ba^{t,(m)}-\ba^*\Vert_{\infty}+2\sqrt{\frac{c_3(d+1)}{n}}\Vert \bb^*-\bb^{t, (m)}\Vert_2 \\
    \leq& \log \kappa_1 + 2\left(\Vert \ba^{t}-\ba^*\Vert_{\infty}+\max_{1\leq m\leq n}\left\Vert \tb^t-\tb^{t,(m)}\right\Vert_2\right) \\
    &+2\sqrt{\frac{c_3(d+1)}{n}}\left(\left\Vert \tb^t-\tb^{*}\right\Vert_2+\max_{1\leq m\leq n}\left\Vert \tb^t-\tb^{t,(m)}\right\Vert_2\right) \\
    \leq &\log \kappa_1+0.2.
\end{align*}
Then it holds that,
\begin{align*}
    \min_{i\neq m}\frac{e^{c_i}}{(1+e^{c_i})^2} = \min_{i\neq m}\frac{e^{-|c_i|}}{(1+e^{-|c_i|})^2}\geq \min_{i\neq m}\frac{e^{-|c_i|}}{4} = \frac{e^{-\max_{i\neq m}|c_i|}}{4}\geq\frac{1}{4\kappa_1e^{0.2}} \geq \frac{1}{5\kappa_1}.
\end{align*}
Using $\displaystyle n-1\geq \frac{n}{2}$ for $n\geq 2$, we have
\begin{align*}
    &\left|\alpha_m^{t,(m)}-\eta\left[\nabla\cL_\tau^{(m)}\left(\tb^{t,(m)}\right)\right]_m-\alpha_m^*\right| \\
    \leq & \left(1-\frac{1}{10\kappa_1}\eta pn\right)|\alpha_m^{t,(m)}-\alpha_m^{*}| \\ &+\frac{1+2\sqrt{c_3(d+1)}}{4}\eta p\sqrt{n}\left(\left\Vert \tb^t-\tb^{*}\right\Vert_2+\max_{1\leq m\leq n}\left\Vert \tb^t-\tb^{t,(m)}\right\Vert_2\right) + \eta\tau \kappa_2 \\
    \leq& \left(1-\frac{1}{10\kappa_1}\eta pn\right)C_5\kappa_1^2\sqrt{\frac{(d+1)\log n}{npL}} \\ &+\frac{1+2\sqrt{c_3(d+1)}}{4}\eta p\sqrt{n}(C_3+C_4)\kappa_1\sqrt{\frac{\log n}{pL}} + \eta\tau \kappa_2.
\end{align*} 
Combine this result with Eq.~\eqref{induction2eq1}, we get
\begin{align}
    \left| \alpha_m^{t+1,(m)} - \alpha_m^*\right| \leq& c_\lambda\eta\kappa_1\sqrt{\frac{(d+1)np\log n}{L}} + \left(1-\frac{1}{10\kappa_1}\eta pn\right)C_5\kappa_1^2\sqrt{\frac{(d+1)\log n}{npL}} \nonumber \\
    &+\frac{1+2\sqrt{c_3(d+1)}}{4}\eta p\sqrt{n}(C_3+C_4)\kappa_1\sqrt{\frac{\log n}{pL}} + \eta\tau \kappa_2 \nonumber \\
    \leq & C_5\kappa_1^2\sqrt{\frac{(d+1)\log n}{npL}}. \label{induction2eq3}
\end{align}
as long as $C_5\geq 30 c_\lambda$, $C_5\geq 7.5(1+2\sqrt{c_3})(C_3+C_4)$ and $C_5\geq 30c_\tau/\sqrt{d+1}$.

Second, let us focus on the case where $m\notin \cS(\ba^*)$. It suffices to control $[\nabla\cL_\tau^{(m)}(\tb^{t,(m)})]_m$, which has been studied before. By Eq.\eqref{eqa}, Eq.\eqref{eqb} as well as the fact that $m\in \cS(\ba^*)^c\subset \cS(\ba^{t,(m)})^c$ we know that
\begin{align*}
    \left[\nabla\cL_\tau^{(m)}\left(\tb^{t,(m)}\right)\right]_m =& p\sum_{i\neq m}\left\{\frac{e^{\tx_i^\top\tb^*}}{e^{\tx_i^\top\tb^*}+e^{\tx_m^\top\tb^*}}-\frac{e^{\tx_i^\top\tb^{t, (m)}}}{e^{\tx_i^\top\tb^{t, (m)}}+e^{\tx_m^\top\tb^{t, (m)}}}\right\}+\tau\alpha_m^{t,(m)}  \\
    =& -p\sum_{i\neq m } \frac{e^{c_i}}{(1+e^{c_i})^2}\left[\alpha_m^* - \alpha_i^*-\alpha_m^{t,(m)}+\alpha_i^{t,(m)}\right]  \nonumber \\ 
    &\;-p\sum_{i\neq m }\frac{e^{c_i}}{(1+e^{c_i})^2}\left[\bx_m^\top\bb^* - \bx_i^\top\bb^*-\bx_m^\top\bb^{t, (m)} + \bx_i^\top\bb^{t, (m)}\right] \\
    =&p\sum_{i\neq m }\frac{e^{c_i}}{(1+e^{c_i})^2}\left[\alpha_i^*-\alpha_i^{t,(m)}+(\bx_m-\bx_i)^\top(\bb^{t,(m)}-\bb^*)\right].
\end{align*}
As a result, the left hand side can be controlled as
\begin{align}
    \left|\left[\nabla\cL_\tau^{(m)}\left(\tb^{t,(m)}\right)\right]_m\right|&\leq \frac{p}{4}\sum_{i\neq m}\left[\left|\alpha_i^*-\alpha_i^{t,(m)}\right| + \left\|\bx_m-\bx_i\right\|_2\left\|\bb^{t,(m)}-\bb^*\right\|_2\right] \nonumber \\
    &\leq \frac{p}{4}\left[\sqrt{n}\|\ba^*-\ba^{t,(m)}\|_2+n\cdot 2\sqrt{\frac{c_3(d+1)}{n}}\|\bb^*-\bb^{t,(m)}\|_2\right]  \nonumber \\
    &\leq \frac{p}{4}\sqrt{n}(1+2\sqrt{c_3(d+1)})\left\|\tb^*-\tb^{t,(m)}\right\|_2.\label{induction2eq2}
\end{align}
Again, since $\|\tb^*-\tb^{t,(m)}\|_2\leq \|\tb^*-\tb^t\|_2+\|\tb^t-\tb^{t,(m)}\|_2$, we have $\|\tb^*-\tb^{t,(m)}\|_2\leq (C_3+C_4)\kappa_1\sqrt{\log n/pL}$. Plugging this in Eq.\eqref{induction2eq2} and using the fact that $c_3\geq 1$, we get
\begin{align}
    \left|\left[\nabla\cL_\tau^{(m)}\left(\tb^{t,(m)}\right)\right]_m\right|&\leq \frac{3\sqrt{c_3}}{4} (C_3+C_4)\kappa_1\sqrt{\frac{(d+1)np\log n}{L}}. \label{induction2eq4}
\end{align}
As a result, as long as $c_\lambda$ satisfies $c_\lambda\geq 0.75\sqrt{c_3}(C_3+C_4)$, we have
\begin{align*}
    \left|\eta\left[\nabla\cL_\tau^{(m)}\left(\tb^{t,(m)}\right)\right]_m\right|\leq \eta\lambda.
\end{align*}
In this case, by Eq.~\eqref{induction2eq1} we know that $|\alpha_m^{t+1,(m)}-\alpha_m^*|=0$. This as well as Eq.~\eqref{induction2eq3} tell us
\begin{align*}
    \max_{1\leq m\leq n}|\alpha_m^{t+1,(m)}-\alpha_m^*|&\leq C_5\kappa_1^2\sqrt{\frac{(d+1)\log n}{npL}}.
\end{align*}

    Next we show $\cS(\ba^{t+1, (m)})\subset\cS(\ba^*)$. Let $k\neq m$ be any index which belongs to $\cS(\ba^*)^c$, it remains to show $k\in\cS(\ba^{t+1, (m)})^c$. By definition we know that
\begin{align*}
    \alpha_k^{t+1, (m)} = s\left(\alpha_k^{t,(m)}-\eta\left[\nabla\cL_\tau^{(m)}\left(\tb^{t,(m)}\right)\right]_k,\eta\lambda\right).
\end{align*}
We write
\begin{align*}
    \alpha_k^{t,(m)}-\eta\left[\nabla\cL_\tau^{(m)}\left(\tb^{t,(m)}\right)\right]_k =& \alpha_k^{t,(k)}-\eta\left[\nabla\cL_\tau^{(k)}\left(\tb^{t,(k)}\right)\right]_k  \\
    &+ \alpha_k^{t,(m)}-\eta\left[\nabla\cL_\tau^{(m)}\left(\tb^{t,(m)}\right)\right]_k - \left(\alpha_k^{t,(k)}-\eta\left[\nabla\cL_\tau^{(k)}\left(\tb^{t,(k)}\right)\right]_k\right).
\end{align*}
According to the induction assumption, we have $\cS(\ba^{t, (k)})\subset\cS(\ba^*)$. This implies that $\alpha_k^{t,(k)}=0$. By triangle inequality we have
\begin{align}
    \left|\alpha_k^{t,(m)}-\eta\left[\nabla\cL_\tau^{(m)}\left(\tb^{t,(m)}\right)\right]_k\right|\leq & \left|\alpha_k^{t,(k)}-\eta\left[\nabla\cL_\tau^{(k)}\left(\tb^{t,(k)}\right)\right]_k\right| \nonumber \\
    &+ \left|\alpha_k^{t,(m)}-\eta\left[\nabla\cL_\tau^{(m)}\left(\tb^{t,(m)}\right)\right]_k - \left(\alpha_k^{t,(k)}-\eta\left[\nabla\cL_\tau^{(k)}\left(\tb^{t,(k)}\right)\right]_k\right)\right| \nonumber \\
    \leq & \eta\left|\left[\nabla\cL_\tau^{(k)}\left(\tb^{t,(k)}\right)\right]_k\right| \nonumber \\
    &+\left\|\tb^{t,(m)}-\eta\nabla\cL_\tau^{(m)}\left(\tb^{t,(m)}\right) - \left(\tb^{t,(k)}-\eta\nabla\cL_\tau^{(k)}\left(\tb^{t,(k)}\right) \right)\right\|_2  \nonumber \\
    \leq& \eta\left|\left[\nabla\cL_\tau^{(k)}\left(\tb^{t,(k)}\right)\right]_k\right| \nonumber \\
    &+2\max_{1\leq i\leq n}\left\|\tb^{t,(i)}-\eta\nabla\cL_\tau^{(i)}\left(\tb^{t,(i)}\right) - \left(\tb^{t}-\eta\nabla\cL_\tau(\tb^{t})\right)\right\|_2.\label{induction2eq5}
\end{align}
From Eq.~\eqref{induction3eq2} to Eq.~\eqref{induction3eq3} we know that
\begin{align*}
    \max_{1\leq i\leq n}\left\|\tb^{t,(i)}-\eta\nabla\cL_\tau^{(i)}\left(\tb^{t,(i)}\right) - \left(\tb^{t}-\eta\nabla\cL_\tau(\tb^{t})\right)\right\|_2\leq C_4\kappa_1\sqrt{\frac{(d+1)\log n}{npL}}.
\end{align*}
Combine this with Eq.~\eqref{induction2eq4} and Eq.~\eqref{induction2eq5} we have
\begin{align*}
    \left|\alpha_k^{t,(m)}-\eta\left[\nabla\cL_\tau^{(m)}\left(\tb^{t,(m)}\right)\right]_k\right|\leq \frac{3\sqrt{c_3}\eta }{4} (C_3+C_4)\kappa_1\sqrt{\frac{(d+1)np\log n}{L}} +2C_4\kappa_1\sqrt{\frac{(d+1)\log n}{npL}}.
\end{align*}
Therefore, as long as $c_\lambda$ is chosen such that
\begin{align*}
    c_\lambda\geq \frac{3\sqrt{c_3}}{4}(C_3+C_4)+\frac{2C_4}{\eta np},
\end{align*}
we have $|\alpha_k^{t,(m)}-\eta[\nabla\cL_\tau^{(m)}(\tb^{t,(m)})]_k|\leq \eta\lambda$. In this case, we know that $\alpha_k^{t+1,(m)}=0$. In other words, $k\in \cS(\ba^{t+1, (m)})^c$. To sum up, it holds that
\begin{align*}
    \cS(\ba^{t+1,(m)})\subset\cS(\ba^*).
\end{align*}
\end{proof}

\subsection{Proof of Lemma \ref{induction4}}\label{induction4proof}
\begin{proof}
For any $m\in[n]$, we have
\begin{align*}
    |\alpha_m^{t+1}-\alpha_m^*|&\leq |\alpha_m^{t+1}-\alpha_m^{t+1,(m)}|+|\alpha_m^{t+1,(m)}-\alpha_m^{*}| \\
    &\leq \left\|\tb_m^{t+1}-\tb_m^{t+1,(m)}\right\|_2+|\alpha_m^{t+1,(m)}-\alpha_m^{*}| \\
    &\leq C_4\kappa_1\sqrt{\frac{(d+1)\log n}{npL}}+C_5\kappa_1^2\sqrt{\frac{(d+1)\log n}{npL}} \\
    &\leq (C_4+C_5)\kappa_1^2\sqrt{\frac{(d+1)\log n}{npL}}.
\end{align*}
As a result, we have
\begin{align*}
    \left\|\ba^{t+1}-\ba^*\right\|_\infty\leq C_6\kappa_1^2\sqrt{\frac{(d+1)\log n}{npL}},
\end{align*}
as long as $C_6\geq C_4+C_5$.

It remains to show that
\begin{align*}
    \cS(\ba^{t+1})\subset\cS(\ba^*).
\end{align*}
For any $k\in \cS(\ba^*)^c$, it suffices to show that $k\in \cS(\ba^{t+1})^c$. By definition we know that
\begin{align*}
    \alpha_k^{t+1} = s\left(\alpha_k^{t}-\eta\left[\nabla\cL_\tau\left(\tb^{t}\right)\right]_k,\eta\lambda\right).
\end{align*}
We write
\begin{align*}
    \alpha_k^{t}-\eta\left[\nabla\cL_\tau\left(\tb^{t}\right)\right]_k =& \alpha_k^{t,(k)}-\eta\left[\nabla\cL_\tau^{(k)}\left(\tb^{t,(k)}\right)\right]_k  \\
    &+ \alpha_k^{t}-\eta\left[\nabla\cL_\tau\left(\tb^{t}\right)\right]_k - \left(\alpha_k^{t,(k)}-\eta\left[\nabla\cL_\tau^{(k)}\left(\tb^{t,(k)}\right)\right]_k\right).
\end{align*}
According to the induction assumption, we have $\cS(\ba^{t, (k)})\subset\cS(\ba^*)$. This implies that $\alpha_k^{t,(k)}=0$. By triangle inequality we have
\begin{align}
    \left|\alpha_k^{t}-\eta\left[\nabla\cL_\tau\left(\tb^{t}\right)\right]_k\right|\leq & \left|\alpha_k^{t,(k)}-\eta\left[\nabla\cL_\tau^{(k)}\left(\tb^{t,(k)}\right)\right]_k\right| \nonumber \\
    &+ \left|\alpha_k^{t}-\eta\left[\nabla\cL_\tau\left(\tb^{t}\right)\right]_k - \left(\alpha_k^{t,(k)}-\eta\left[\nabla\cL_\tau^{(k)}\left(\tb^{t,(k)}\right)\right]_k\right)\right| \nonumber \\
    \leq & \eta\left|\left[\nabla\cL_\tau^{(k)}\left(\tb^{t,(k)}\right)\right]_k\right| \nonumber \\
    &+\left\|\tb^{t}-\eta\nabla\cL_\tau\left(\tb^{t}\right) - \left(\tb^{t,(k)}-\eta\nabla\cL_\tau^{(k)}\left(\tb^{t,(k)}\right) \right)\right\|_2  \nonumber \\
    \leq& \eta\left|\left[\nabla\cL_\tau^{(k)}\left(\tb^{t,(k)}\right)\right]_k\right| \nonumber \\
    &+\max_{1\leq i\leq n}\left\|\tb^{t,(i)}-\eta\nabla\cL_\tau^{(i)}\left(\tb^{t,(i)}\right) - \left(\tb^{t}-\eta\nabla\cL_\tau(\tb^{t})\right)\right\|_2.\label{induction4eq1}
\end{align}
Again, from Eq.~\eqref{induction3eq2} to Eq.~\eqref{induction3eq3} we have
\begin{align*}
    \max_{1\leq i\leq n}\left\|\tb^{t,(i)}-\eta\nabla\cL_\tau^{(i)}\left(\tb^{t,(i)}\right) - \left(\tb^{t}-\eta\nabla\cL_\tau(\tb^{t})\right)\right\|_2\leq C_4\kappa_1\sqrt{\frac{(d+1)\log n}{npL}}.
\end{align*}
Combine this with Eq.~\eqref{induction2eq4} and Eq.~\eqref{induction4eq1} we have
\begin{align*}
    \left|\alpha_k^{t,(m)}-\eta\left[\nabla\cL_\tau^{(m)}\left(\tb^{t,(m)}\right)\right]_k\right|\leq \frac{3\sqrt{c_3}\eta }{4} (C_3+C_4)\kappa_1\sqrt{\frac{(d+1)np\log n}{L}} +C_4\kappa_1\sqrt{\frac{(d+1)\log n}{npL}}.
\end{align*}
Since $c_\lambda$ is already chosen to satisfy that
\begin{align*}
    c_\lambda\geq \frac{3\sqrt{c_3}}{4}(C_3+C_4)+\frac{2C_4}{\eta np}\geq \frac{3\sqrt{c_3}}{4}(C_3+C_4)+\frac{C_4}{\eta np},
\end{align*}
we have $|\alpha_k^{t}-\eta[\nabla\cL_\tau(\tb^{t})]_k|\leq \eta\lambda$. In this case, we know that $\alpha_k^{t+1}=0$. In other words, $k\in \cS(\ba^{t+1})^c$. To sum up, it holds that
\begin{align*}
    \cS(\ba^{t+1})\subset\cS(\ba^*).
\end{align*}

\end{proof}

\subsection{Proof of Lemma \ref{supportlemma}}\label{supportlemmaproof}
\begin{proof}
Since $\tb_R$ is the minimizer of $\cL_R(\cdot)$, we know that 
\begin{align*}
    -\left[\nabla\cL_\tau(\tb_R)\right]_{1:n}\in \partial \left\|\widehat\ba_R\right\|_1, \quad \left[\nabla\cL_\tau(\tb_R)\right]_{n+1:n+d} = 0.
\end{align*}
This implies $\textsf{SOFT}_{\eta\lambda}(\tb_R-\eta\nabla\cL_\tau(\tb_R)) = \tb_R$. For any $k\in \cS(\ba^*)^c$, we want to prove $k\in \cS(\widehat\ba_R)^c$. It suffices to show that 
\begin{align*}
    \left|\left[\tb_R-\eta\nabla\cL_\tau(\tb_R)\right]_k\right| \leq \eta\lambda.
\end{align*}
For the same reason as Eq.\eqref{pgd1eq2} and Eq.~\eqref{pgd1eq3}, we know that 
\begin{align*}
    \left\|\tb^{T-1}-\eta\nabla\cL_\tau(\tb^{T-1}) - \left(\tb_R-\eta\nabla\cL_\tau(\tb_R)\right)\right\|_2\leq \left\|\tb^{T-1}-\tb_R\right\|_2.
\end{align*}
As a result, we can write
\begin{align*}
    \left|\left[\tb_R-\eta\nabla\cL_\tau(\tb_R)\right]_k\right|\leq \left|\left[\tb^{T-1}-\eta\nabla\cL_\tau(\tb^{T-1})\right]_k\right| +\left\|\tb^{T-1}-\tb_R\right\|_2 
\end{align*}
On the other hand, similar to Eq.~\eqref{induction4eq1} we know that
\begin{align*}
    \left|\left[\tb^{T-1}-\eta\nabla\cL_\tau(\tb^{T-1})\right]_k\right|\leq& \eta\left|\left[\nabla\cL_\tau^{(k)}\left(\tb^{T-1,(k)}\right)\right]_k\right| \nonumber \\
    &+\max_{1\leq i\leq n}\left\|\tb^{T-1,(i)}-\eta\nabla\cL_\tau^{(i)}\left(\tb^{T-1,(i)}\right) - \left(\tb^{T-1}-\eta\nabla\cL_\tau(\tb^{T-1})\right)\right\|_2 \\
    \leq & \frac{3\sqrt{c_3}\eta }{4} (C_3+C_4)\kappa_1\sqrt{\frac{(d+1)np\log n}{L}} +C_4\kappa_1\sqrt{\frac{(d+1)\log n}{npL}}.
\end{align*}
Combine this with Lemma \ref{lem5} we know that
\begin{align*}
    \left|\left[\tb_R-\eta\nabla\cL_\tau(\tb_R)\right]_k\right|&\leq \left|\left[\tb^{T-1}-\eta\nabla\cL_\tau(\tb^{T-1})\right]_k\right| +\left\|\tb^{T-1}-\tb_R\right\|_2  \\
    &\leq \frac{3\sqrt{c_3}\eta }{4} (C_3+C_4)\kappa_1\sqrt{\frac{(d+1)np\log n}{L}} +(C_4+C_7)\kappa_1\sqrt{\frac{(d+1)\log n}{npL}}.
\end{align*}
As a result, as long as $c_\lambda$ is chosen to satisfies
\begin{align*}
    c_\lambda\geq \frac{3\sqrt{c_3}}{4}(C_3+C_4)+\frac{C_4+C_7}{\eta np},
\end{align*}
we have $|[\tb_R-\eta\nabla\cL_\tau(\tb_R)]_k| \leq \eta\lambda$. In this case, we know that $[\tb_R]_k = 0$ and thus $k\in \cS(\widehat\ba_R)^c$. To sum up, we have
\begin{align*}
    \cS(\widehat\ba_R)\subset\cS(\ba^*).
\end{align*}
\end{proof}

\begin{lemma}\label{cLinfty}
    With $\tau$ given by \ref{estimationthm}, with probability at least $1-O(n^{-10})$ we have
    \begin{align*}
        \left\|\left[\nabla \cL_\tau(\tb^*)\right]_{1:n}\right\|_\infty\lesssim \sqrt{\frac{np\log n}{L}},\quad \left\|\left[\nabla \mathcal{L}_\tau(\tb^*)\right]_{n+1:n+d}\right\|_2\lesssim \sqrt{\frac{(d+1)np\log n}{L}}.
    \end{align*}
\end{lemma}
\begin{proof}
For any $i\in [n]$, by definition we know that
\begin{align*}
    \left[\nabla \cL_\tau(\tb^*)\right]_{i} &= \sum_{j:j\neq i, (i,j)\in \mathcal{E}}\left\{-y_{j,i}+\frac{e^{\tx_i^\top \tb}}{e^{\tx_i^\top \tb}+e^{\tx_j^\top \tb}}\right\} +\tau \alpha_i \\
    &= \frac{1}{L}\sum_{j:j\neq i, (i,j)\in \mathcal{E}}\sum_{l=1}^L\left\{-y_{j,i}^{(l)}+\frac{e^{\tx_i^\top \tb}}{e^{\tx_i^\top \tb}+e^{\tx_j^\top \tb}}\right\} +\tau\alpha_i.
\end{align*}
Using Bernstein inequality conditioned on the comparison graph $\mathcal{G}$, with probability at least $1-O(n^{-11})$ we have
\begin{align*}
    \left|\sum_{j:j\neq i, (i,j)\in \mathcal{E}}\sum_{l=1}^L\left\{-y_{j,i}^{(l)}+\frac{e^{\tx_i^\top \tb}}{e^{\tx_i^\top \tb}+e^{\tx_j^\top \tb}}\right\}\right|&\lesssim \sqrt{\sum_{j:j\neq i, (i,j)\in \mathcal{E}}\frac{L e^{\tx_i^\top \tb}e^{\tx_j^\top \tb}}{\left(e^{\tx_i^\top \tb}+e^{\tx_j^\top \tb}\right)^2}\log n } + \log n \\
    &\lesssim \sqrt{npL\log n} +\log n \lesssim \sqrt{npL\log n}.
\end{align*}
The last $\lesssim$ holds since $npL\gtrsim \log n $. As a result, we know that
\begin{align*}
    \left|\left[\nabla \cL_\tau(\tb^*)\right]_{i}\right|\lesssim \frac{\sqrt{npL\log n}}{L} +\left|\alpha_i^*\right|\tau\leq \sqrt{\frac{np\log n}{L}}+c_\tau\frac{\left\|\ba^*\right\|_\infty}{\kappa_2}\sqrt{\frac{p\log n}{L}}\lesssim \sqrt{\frac{np\log n}{L}}
\end{align*}
with probability exceeding $1-O(n^{-11})$.

On the other hand, we write 
\begin{align*}
    \left[\nabla \mathcal{L}_\tau(\tb^*)\right]_{n+1:n+d} = \tau\bb^*+\frac{1}{L}\sum_{(i,j)\in\mathcal{E}, i>j}\sum_{l=1}^L\underbrace{\left\{-y_{j,i}^{(l)}+\frac{e^{\tx_i^\top \tb^*}}{e^{\tx_i^\top \tb^*}+e^{\tx_j^\top \tb^*}}\right\}(\bx_i-\bx_j)}_{:=z_{i,j}^{(l)}}.
\end{align*}
Since $\bE [z_{i,j}^{(l)}] = 0, \Vert z_{i,j}^{(l)}\Vert_2\leq\Vert \tx_i-\tx_j\Vert_2\leq 2\sqrt{c_3(d+1)/n}$, we have 
\begin{align*}
    \bE[z_{i,j}^{(l)}z_{i,j}^{(l)\top} ] =  \text{Var}[y_{j,i}^{(l)}]&(\tx_i-\tx_j)(\tx_i-\tx_j)^\top \prec(\tx_i-\tx_j)(\tx_i-\tx_j)^\top\\
    \text{and} &\qquad \bE[z_{i,j}^{(l)\top} z_{i,j}^{(l)} ]\leq \frac{4c_3(d+1)}{n}.
\end{align*}
Thus, with high probability (with respect to the randomness of $\mathcal{G}$), we have
\begin{align*}
    \left\|\sum_{(i, j) \in \mathcal{E}, i>j} \sum_{l=1}^{L} \mathbb{E}\left[z_{i, j}^{(l)} z_{i, j}^{(l) \top}\right]\right\| \leq L\left\|\sum_{(i, j) \in \mathcal{E}, i>j}\left(\bx_{i}-\bx_{j}\right)\left(\bx_{i}-\bx_{j}\right)^{\top}\right\|=L\left\|\boldsymbol{L}_{\mathcal{G}}\right\| \lesssim L n p 
\end{align*}
and 
\begin{align*}
    \left|\sum_{(i, j) \in \mathcal{E}, i>j} \sum_{l=1}^{L} \mathbb{E}\left[\boldsymbol{z}_{i, j}^{(l) \top} \boldsymbol{z}_{i, j}^{(l)}\right]\right| \leq \frac{4c_3(d+1)}{n} L\left|\sum_{(i, j) \in \mathcal{E}, i>j} 1\right| \lesssim (d+1)Ln p.
\end{align*}
Let $V:=\frac{1}{L^{2}} \max \left\{\left\|\sum_{(i, j) \in \mathcal{E}} \sum_{l=1}^{L} \mathbb{E}\left[\boldsymbol{z}_{i, j}^{(l)} \boldsymbol{z}_{i, j}^{(l) \top}\right]\right\|,\left|\sum_{(i, j) \in \mathcal{E}} \sum_{l=1}^{L} \mathbb{E}\left[\boldsymbol{z}_{i, j}^{(l) \top} \boldsymbol{z}_{i, j}^{(l)}\right]\right|\right\}$ and $B:=\max_{i,j,l}\Vert z_{i,j}^{(l)}\Vert/L$. By matrix Bernstein inequality \citep{tropp2015introduction}, with probability at least $1-O(n^{-10})$ we have
\begin{align*}
    \left\|\frac{1}{L}\sum_{(i,j)\in\mathcal{E}, i>j}\sum_{l=1}^L z_{i,j}^{(l)}\right\|_{2} &\lesssim \sqrt{V \log (n+d+1)}+B \log (n+d+1) \\
    &\lesssim \sqrt{\frac{(d+1)n p \log n}{L}}+\sqrt{\frac{d+1}{n}}\frac{\log n}{L}\lesssim  \sqrt{\frac{(d+1)n p \log n}{L}},
\end{align*}
The last $\lesssim$ holds since $np\gtrsim \log n $. Therefore, we get
\begin{align*}
    \left\|\left[\nabla \mathcal{L}_\tau(\tb^*)\right]_{n+1:n+d}\right\|_2&\lesssim \sqrt{\frac{(d+1)n p \log n}{L}} +\left\|\bb^*\right\|\tau\\
    &\leq \sqrt{\frac{(d+1)n p \log n}{L}}+c_\tau\frac{\left\|\tb^*\right\|_\infty}{\kappa_3}\sqrt{\frac{p\log n}{L}}\lesssim \sqrt{\frac{(d+1)n p \log n}{L}}
\end{align*}
with probability exceeding $1-O(n^{-10})$.

\end{proof}

\subsection{Auxiliary Lemma}
\begin{lemma}\label{Lilemma}
	For $i\in [n]$, with probability at least $1-O(n^{-10})$ we have
	\begin{itemize}
		\item $\displaystyle\left|\left(\nabla\cL(\tb^*)\right)_i\right|\lesssim\sqrt{\frac{np\log n}{L}}$;
		\item $\displaystyle\sum_{j\neq i}\left(\nabla^2\cL(\tb^*)\right)_{i,j}^2\lesssim np(1+dp)$, \quad $\displaystyle\sum_{k>n}\left(\nabla^2\cL(\tb^*)\right)_{i,k}^2\lesssim  ndp^2$,  \quad $\displaystyle\sum_{j\in [n],j\neq i}\left|\left(\nabla^2\cL(\tb^*)\right)_{i,j}\right|\lesssim np$.
		\item  $ \left|y_{j,i}-\mathbb{E}y_{j,i}\right|\lesssim\sqrt{\frac{\log n}{L}}$, for any $i,j\in[n], i\neq j$.
	\end{itemize}
\end{lemma}
\begin{proof}
(1) By definition for $i\in [n]$ we have
\begin{align*}
    \left(\nabla\cL(\tb^*)\right)_i &= \sum_{j\neq i, (i,j)\in\mathcal{E}}\left\{-y_{j,i}+\phi(\tx_i^\top\tb^*-\tx_j^\top\tb^*)\right\} \\
    &=\frac{1}{L}\sum_{j\neq i, (i,j)\in\mathcal{E}}\sum_{l=1}^L\left\{-y^{(l)}_{j,i}+\phi(\tx_i^\top\tb^*-\tx_j^\top\tb^*)\right\}.
\end{align*}
Since $\left|-y^{(l)}_{j,i}+\phi(\tx_i^\top\tb^*-\tx_j^\top\tb^*)\right|\leq 1$, by Bernstein inequality we have
\begin{align*}
    \left|\left(\nabla\cL(\tb^*)\right)_i-\mathbb{E}\left[\left(\nabla\cL(\tb^*)\right)_i\bigg|\mathcal{G}\right]\right|&\lesssim\frac{1}{L}\left(\sqrt{\log n \left(\sum_{j\neq i, (i,j)\in\mathcal{E}}1\right)L}+\log n\right) \\
    &\lesssim\sqrt{\frac{np\log n}{L}}
\end{align*}
with probability exceeding $1-O(n^{-10})$, as long as $npL\gtrsim \log n$. On the other hand, since $\mathbb{E}\left[-y^{(l)}_{j,i}+\phi(\tx_i^\top\tb^*-\tx_j^\top\tb^*)\right]=0$, we know that $\mathbb{E}\left[\left(\nabla\cL(\tb^*)\right)_i\bigg|\mathcal{G}\right]=0$. As a result, we have
\begin{align*}
    \left|\left(\nabla\cL(\tb^*)\right)_i\right|\lesssim\sqrt{\frac{np\log n}{L}}.
\end{align*}

(2) By definition we have 
\begin{align*}
    \sum_{j\neq i}\left(\nabla^2\cL(\tb^*)\right)_{i,j}^2 &= \left\Vert\left(\sum_{j\neq i, (i,j)\in \mathcal{E}}\phi'(\tx_i^\top\tb^*-\tx_j^\top\tb^*)\left(\tx_i-\tx_j\right)\right)_{-i} \right\Vert^2_2 \\
    &=\sum_{j\neq i, (i,j)\in \mathcal{E}} 1 +\left\Vert\sum_{j\neq i, (i,j)\in \mathcal{E}}\phi'(\tx_i^\top\tb^*-\tx_j^\top\tb^*)\left(\bx_i-\bx_j\right) \right\Vert^2_2 \\
    &\leq \sum_{j\neq i, (i,j)\in \mathcal{E}} 1 +\left(\sum_{j\neq i, (i,j)\in \mathcal{E}} 1\right)\sum_{j\neq i, (i,j)\in \mathcal{E}} \Vert \bx_i-\bx_j\Vert_2^2 \\
    &\lesssim np+np\cdot dp
\end{align*}
with probability at least $1-O(n^{-10})$. Similarly, we have
\begin{align*}
    \sum_{k>n}\left(\nabla^2\cL(\tb^*)\right)_{i,k}^2 &= \left\Vert\sum_{j\neq i, (i,j)\in \mathcal{E}}\phi'(\tx_i^\top\tb^*-\tx_j^\top\tb^*)\left(\bx_i-\bx_j\right) \right\Vert^2_2 \\
    &\leq \left(\sum_{j\neq i, (i,j)\in \mathcal{E}} 1\right)\sum_{j\neq i, (i,j)\in \mathcal{E}} \Vert \bx_i-\bx_j\Vert_2^2 \\
    &\lesssim np\cdot dp
\end{align*}
and 
\begin{align*}
    \sum_{j\in [n],j\neq i}\left|\left(\nabla^2\cL(\tb^*)\right)_{i,j}\right| = \sum_{j\in [n],j\neq i}\left|\phi'(\tx_i^\top\tb^*-\tx_j^\top\tb^*)\right|\textbf{1}((i,j)\in \mathcal{E})\leq  \sum_{j\in [n],j\neq i}\textbf{1}((i,j)\in \mathcal{E})\lesssim np
\end{align*}
with probability at least $1-O(n^{-10})$.

(3) For $i,j\in [n], i\neq j$, by definition we know that $\displaystyle y_{j,i} = \frac{1}{L}\sum_{l=1}^L y_{j,i}^{(l)}$ is the average of $L$ independent Bernoulli random variables. By Hoeffding's inequality, we know that 
\begin{align*}
    \left|y_{j,i}-\mathbb{E}y_{j,i}\right|\lesssim \sqrt{\frac{\log n}{L}}
\end{align*}
with probability at least $1-O(n^{-12})$. As a result, by union bound we know that 
\begin{align*}
    \left|y_{j,i}-\mathbb{E}y_{j,i}\right|\lesssim \sqrt{\frac{\log n}{L}}
\end{align*}
holds for all $i,j\in [n], i\neq j$ with probability at least $1-O(n^{-10})$.
\end{proof}

\subsection{Proof of Lemma \ref{oaestimation}}\label{oaestimationproof}
\begin{proof}
We apply \cite[Theorem 4.1]{fan2020factor} to prove this statement. Use the notation in \cite{fan2020factor}, we denote by
\begin{align*}
    L_n(\boldsymbol{\theta}) = \ocL(\boldsymbol{\theta}) + \frac{\tau}{2}\left\|\boldsymbol{\theta}\right\|_2^2.
\end{align*}
We can see that \cite[Assumption 4.1]{fan2020factor} holds with $M = 0$ and $A = \infty$. Under our model, we have
\begin{align*}
    &S = \text{supp}(\ba^*)\cup \left\{n+1,n+2,\dots, n+d\right\},\quad S_1 = \text{supp}(\ba^*),\quad S_2 = [n+d]\backslash S = [n]\backslash S_1.
\end{align*}
As a result, the irrepresentable condition can be shown in the following way. By definition we have
\begin{align}
    &\left\|\nabla_{S_2 S}^2 L_n\left(\boldsymbol{\theta}^*\right)\left[\nabla_{S S}^2 L_n\left(\boldsymbol{\theta}^*\right)\right]^{-1}\right\|_{\infty} = \max_{i\in S_2} \left\|\nabla_{i,S}^2 L_n\left(\boldsymbol{\theta}^*\right)\left[\nabla_{S S}^2 L_n\left(\boldsymbol{\theta}^*\right)\right]^{-1}\right\|_1 \nonumber \\
    \leq &\max_{i\in S_2} \sqrt{k+d}\left\|\nabla_{i,S}^2 L_n\left(\boldsymbol{\theta}^*\right)\left[\nabla_{S S}^2 L_n\left(\boldsymbol{\theta}^*\right)\right]^{-1}\right\|_2  \leq \max_{i\in S_2} \sqrt{k+d}\left\|\nabla_{i,S}^2 L_n\left(\boldsymbol{\theta}^*\right)\right\|_2\left\|\left[\nabla_{S S}^2 L_n\left(\boldsymbol{\theta}^*\right)\right]^{-1}\right\|. \label{oasupporteq1}
\end{align}
Take $\tb''= \tb^*, \cS(\tb)\subset \cS(\tb^*), \tb' = \boldsymbol{0}$ in Lemma \ref{lb}, we know that 
\begin{align*}
    \tb^\top \nabla_{S S}^2 L_n\left(\boldsymbol{\theta}^*\right)\tb\gtrsim \frac{np}{\kappa_1}\left\|\tb\right\|_2^2.
\end{align*}
As a result, we have $\|[\nabla_{S S}^2 L_n(\boldsymbol{\theta}^*)]^{-1}\|\lesssim \kappa_1/np$. On the other hand, Since the maximum in \eqref{oasupporteq1} is taken over $S_2$, and $S_2\cap S = \emptyset$, we know that
\begin{align*}
    \left\|\nabla_{i,S}^2 L_n\left(\boldsymbol{\theta}^*\right)\right\|_2 &\leq \sqrt{\sum_{(i,j)\in \mathcal{E}, j\in S}1+ \left\|\nabla_{i,n+1:n+d}^2 L_n\left(\boldsymbol{\theta}^*\right)\right\|_2^2} \lesssim \sqrt{np} + \left\|\nabla_{i,n+1:n+d}^2 L_n\left(\boldsymbol{\theta}^*\right)\right\|_2  \\
    &\lesssim \sqrt{np}+\sum_{(i,j)\in \mathcal{E}}\left\|\bx_i-\bx_j\right\|_2\lesssim \sqrt{np}+np\sqrt{\frac{d+1}{n}} = \sqrt{n}(\sqrt{p}+p\sqrt{d+1})
\end{align*}
with probability at least $1-O(n^{-10})$. Plug these in \eqref{oasupporteq1}, we get
\begin{align*}
    \left\|\nabla_{S_2 S}^2 L_n\left(\boldsymbol{\theta}^*\right)\left[\nabla_{S S}^2 L_n\left(\boldsymbol{\theta}^*\right)\right]^{-1}\right\|_{\infty}\lesssim \kappa_1\sqrt{\frac{k+d}{np}}+\kappa_1\sqrt{\frac{(k+d)(d+1)}{n}}.
\end{align*}
As a result, as long as $np\geq C\kappa_1^2(k+d)$ and $n\geq C\kappa_1^2(k+d)(d+1)$ for some constant $C>0$, the irrepresentable condition \cite[Assumption 4.3]{fan2020factor} holds for $\tau = 0.5$ (for $\tau$ defined in \cite{fan2020factor}). According to Lemma \ref{cLinfty}, since $\lambda$ defined in Theorem \ref{estimationthm} satisfies $\|\nabla L_n(\boldsymbol{\theta^*})\|_\infty \lesssim \lambda$, by \cite[Assumption 4.1]{fan2020factor} we know that $\cS(\oa_R)\subset \cS(\ba^*)$.

On the other hand, one can show that
\begin{align*}
     \left\|\left[\nabla_{S S}^2 L_n\left(\boldsymbol{\theta}^*\right)\right]^{-1}\right\|_{\infty} &= \max_{i\in S}\left\|\left[\nabla_{i, S}^2 L_n\left(\boldsymbol{\theta}^*\right)\right]^{-1}\right\|_{1} \leq \max_{i\in S}\sqrt{k+d}\left\|\left[\nabla_{i, S}^2 L_n\left(\boldsymbol{\theta}^*\right)\right]^{-1}\right\|_{2} \\
     &\leq \max_{i\in S}\sqrt{k+d}\left\|\left[\nabla_{S, S}^2 L_n\left(\boldsymbol{\theta}^*\right)\right]^{-1}\right\| \lesssim \frac{\kappa_1\sqrt{k+d}}{np}, \\
     \left\|\left[\nabla_{S S}^2 L_n\left(\boldsymbol{\theta}^*\right)\right]^{-1}\right\|_{2} &\lesssim \frac{\kappa_1}{np}.
\end{align*}
As a result, by \cite[Theorem 4.1]{fan2020factor}, Lemma \ref{cLinfty} we know that
\begin{align*}
    \left\|\oa_R-\ba^*\right\|_\infty\lesssim \kappa_1^2\sqrt{\frac{(k+d)(d+1)\log n}{npL}}, \quad  \left\|\ob_R-\tb^*\right\|_2\lesssim \kappa_1^2\sqrt{\frac{(k+d)(d+1)\log n}{npL}}.
\end{align*}
\end{proof}

\subsection{Proof of Theorem \ref{residue2norm}}\label{residue2normproof}
\begin{proof}
By the strongly convex property we know that for any $\boldsymbol{v}\in \partial \cL_R(\tb_R)$, we have
\begin{align*}
    \cL_R(\ob_R)\geq \cL_R(\tb_R)+\boldsymbol{v}^\top\left(\ob_R-\tb_R\right)+\frac{1}{2}\left(\ob_R-\tb_R\right)^\top\left(\tau + \nabla^2\cL(\tb_R)\right)\left(\ob_R-\tb_R\right).
\end{align*}
Since $\tb_R$ is the minimizer of $\cL_R(\cdot)$, we know that $\boldsymbol{0}\in \partial \cL_R(\tb_R)$. As a result, we know that
\begin{align}
    \cL_R(\ob_R)\geq \cL_R(\tb_R)+\frac{1}{2}\left(\ob_R-\tb_R\right)^\top\left(\tau + \nabla^2\cL(\tb_R)\right)\left(\ob_R-\tb_R\right). \label{residue2normeq1}
\end{align}
Similarly, for any $\boldsymbol{v}\in \partial \ocL_R(\ob_R)$, we have
\begin{align*}
    \ocL_R(\tb_R)\geq \ocL_R(\ob_R)+\boldsymbol{v}^\top\left(\tb_R- \ob_R\right)+\frac{1}{2}\left(\tb_R - \ob_R\right)^\top\left(\tau + \nabla^2\ocL(\ob_R)\right)\left(\tb_R - \ob_R\right).
\end{align*}
Since $\boldsymbol{0}\in \partial \ocL_R(\ob_R)$ and by the definition of $\ocL(\cdot)$ we have
\begin{align}
    \ocL_R(\tb_R)\geq \ocL_R(\ob_R)+\frac{1}{2}\left(\tb_R - \ob_R\right)^\top\left(\tau + \nabla^2\cL(\tb^*)\right)\left(\tb_R - \ob_R\right).\label{residue2normeq2}
\end{align}
According to Theorem \ref{estimationthm} and Lemma \ref{oaestimation}, we know that $\cS(\ba_R)\subset \cS$ and $\cS(\oa_R)\subset \cS$. Therefore, by Lemma \ref{lb}, \eqref{residue2normeq1} and \eqref{residue2normeq2} we know that 
\begin{align}
    \frac{np}{\kappa_1}\left\|\tb_R-\ob_R\right\|_2^2&\lesssim \left|\cL_R(\tb_R) - \ocL_R(\tb_R)\right|+\left|\cL_R(\ob_R) - \ocL_R(\ob_R)\right| \nonumber \\
    &=\left|\cL(\tb_R) - \ocL(\tb_R)\right|+\left|\cL(\ob_R) - \ocL(\ob_R)\right|.\label{residue2normeq5}
\end{align}
It remains to control the right hand side of the above inequality.

For any $\tb\in \mathbb{R}^{n+d}$, we have
\begin{align}
    \left|\cL(\tb) - \ocL(\tb)\right| &= \left|\int_0^1\left(\nabla\cL(\tb^*+\gamma(\tb-\tb^*)) - \nabla\ocL(\tb^*+\gamma(\tb-\tb^*))\right)^\top(\tb-\tb^*)d\gamma + \cL(\tb^*) - \ocL(\tb^*)\right| \nonumber\\
    &=\left|\int_0^1\left(\nabla\cL(\tb^*+\gamma(\tb-\tb^*)) - \nabla\ocL(\tb^*+\gamma(\tb-\tb^*))\right)^\top(\tb-\tb^*)d\gamma \right| \nonumber\\
    &\leq \max_{\gamma\in [0, 1]}\left\|\nabla\cL(\tb^*+\gamma(\tb-\tb^*)) - \nabla\ocL(\tb^*+\gamma(\tb-\tb^*))\right\|_2\left\|\tb-\tb^*\right\|_2.\label{residue2normeq3}
\end{align}
On the other hand, we also know that
\begin{align}
 &\left\|\nabla\cL(\tb^*+\gamma(\tb-\tb^*)) - \nabla\ocL(\tb^*+\gamma(\tb-\tb^*))\right\|_2  \nonumber \\
=& \left\|\int_0^1\left(\nabla^2\cL(\tb^*+t\gamma(\tb-\tb^*)) - \nabla^2\ocL(\tb^*+t\gamma(\tb-\tb^*))\right)(\tb-\tb^*) dt\right\|_2  \nonumber \\
\leq & \sup_{t\in [0, 1]}\left\|\nabla^2\cL(\tb^*+t\gamma(\tb-\tb^*)) - \nabla^2\cL(\tb^*)\right\|\left\|\tb-\tb^*\right\|_2.\label{residue2normeq4}
\end{align}
\eqref{residue2normeq3} together with \eqref{residue2normeq4} implies
\begin{align}
    \left|\cL(\tb) - \ocL(\tb)\right|\leq \sup_{t\in [0, 1]}\left\|\nabla^2\cL(\tb^*+t(\tb-\tb^*)) - \nabla^2\cL(\tb^*)\right\|\left\|\tb-\tb^*\right\|_2^2. \label{residue2normeq6}
\end{align}
By the definition we can control the difference of hessian as
\begin{align*}
    &\left\Vert\nabla^2\cL(\tb^*+t(\tb-\tb^*))-\nabla^2\cL(\tb^*)\right\Vert  \\
    =&\left\Vert\sum_{(i,j)\in\mathcal{E},i>j}\left(\phi'(t(\tx_i^\top\tb-\tx_j^\top\tb)+(1-t)(\tx_i^\top\tb^*-\tx_j^\top\tb^*))-\phi'(\tx_i^\top\tb^*-\tx_j^\top\tb^*)\right)(\tx_i-\tx_j)(\tx_i-\tx_j)^\top\right\Vert \\
    \leq&\left\Vert\sum_{(i,j)\in\mathcal{E},i>j}\left|\phi'(t(\tx_i^\top\tb-\tx_j^\top\tb)+(1-t)(\tx_i^\top\tb^*-\tx_j^\top\tb^*))-\phi'(\tx_i^\top\tb^*-\tx_j^\top\tb^*)\right|(\tx_i-\tx_j)(\tx_i-\tx_j)^\top\right\Vert \\
    \lesssim&\left\Vert\sum_{(i,j)\in\mathcal{E},i>j}\left|t(\tx_i^\top\tb-\tx_j^\top\tb)-t(\tx_i^\top\tb^*-\tx_j^\top\tb^*)\right|(\tx_i-\tx_j)(\tx_i-\tx_j)^\top\right\Vert \\
    \lesssim &\left\Vert\sum_{(i,j)\in\mathcal{E},i>j}(\tx_i-\tx_j)(\tx_i-\tx_j)^\top\right\Vert \left\Vert\tb-\tb^* \right\Vert_c \lesssim \left\Vert \boldsymbol{L}_\mathcal{G}\right\Vert\left\Vert\tb-\tb^* \right\Vert_c.
\end{align*}
As a result, by Lemma \ref{eigenlem} we know that 
\begin{align}
    \sup_{t\in [0, 1]}\left\|\nabla^2\cL(\tb^*+t(\tb-\tb^*)) - \nabla^2\cL(\tb^*)\right\|\lesssim np\left\Vert\tb-\tb^* \right\Vert_c \label{residue2normeq7}
\end{align}
with probability at least $1-O(n^{-10})$. Plugging \eqref{residue2normeq6} and \eqref{residue2normeq7} in \eqref{residue2normeq5} we get
\begin{align*}
    \frac{np}{\kappa_1}\left\|\tb_R-\ob_R\right\|_2^2&\lesssim np\left\|\tb_R-\tb^*\right\|_2^2\left\Vert\tb_R-\tb^* \right\Vert_c +np \left\|\ob_R-\tb^*\right\|_2^2\left\Vert\ob_R-\tb^* \right\Vert_c.
\end{align*}
According to Theorem \ref{estimationthm} and Lemma \ref{oaestimation}, the estimation error can be controlled as 
\begin{align*}
    \left\|\tb_R-\ob_R\right\|_2&\lesssim \kappa_1^{3.5}\left(\frac{(k+d)(d+1)\log n}{npL}\right)^{3/4}.
\end{align*}
\end{proof}

\subsection{Proof of Lemma \ref{alphadebiasapprox1}}\label{alphadebiasapprox1proof}
\begin{proof}
By \eqref{oadexpansion} we know that
\begin{align*}
    \overline{\alpha}_{R,i}^{\debias} = \alpha_i^* - \frac{\left(\nabla \cL(\tb^*)\right)_i + \sum_{j\neq i}\left(\nabla^2\cL(\tb^*)\right)_{i,j}\left(\ob_{R,j} - \tb^*_j\right)}{\left(\nabla^2\cL(\tb^*)\right)_{i,i}}.
\end{align*}
Similarly, according to \eqref{dotalphadefinition} we know that
\begin{align}
    \dot{\alpha}_{R,i}^{\debias} = \alpha_i^* - \frac{\left(\nabla \cL(\tb^*)\right)_i + \sum_{j\neq i}\left(\nabla^2\cL(\tb^*)\right)_{i,j}\left(\tb_{R,j} - \tb^*_j\right)}{\left(\nabla^2\cL(\tb^*)\right)_{i,i}}. \label{alpharesidue1eq1}
\end{align}
As a result, one can see that
\begin{align*}
    \dot{\alpha}_{R,i}^{\debias} - \overline{\alpha}_{R,i}^{\debias} = \frac{\sum_{j\neq i}\left(\nabla^2\cL(\tb^*)\right)_{i,j}\left(\ob_{R,j} - \tb_{R,j}\right)}{\left(\nabla^2\cL(\tb^*)\right)_{i,i}}.
\end{align*}
By Lemma \ref{Lilemma}, the numerator can be controlled as
\begin{align*}
    \left|\sum_{j\neq i}\left(\nabla^2\cL(\tb^*)\right)_{i,j}\left(\ob_{R,j} - \tb_{R,j}\right)\right| &\leq \sqrt{\sum_{j\neq i}\left(\nabla^2\cL(\tb^*)\right)_{i,j}^2}\left\|\ob_{R} - \tb_{R}\right\|_2 \\
    &\leq \sqrt{np(1+d)}\left\|\ob_{R} - \tb_{R}\right\|_2.
\end{align*}
On the other hand, with probability at least $1-O(n^{-10})$ we have $(\nabla^2\cL(\tb^*))_{i,i}\gtrsim np/\kappa_1$. As a result, by Theorem \ref{residue2norm} we know that
\begin{align*}
    | \dot{\alpha}_{R,i}^{\debias} - \overline{\alpha}_{R,i}^{\debias}|\lesssim \frac{\sqrt{np(1+d)}\left\|\ob_{R} - \tb_{R}\right\|_2}{np/\kappa_1}\lesssim \frac{\kappa_1^{4.5}(d+1)}{np}\sqrt{\frac{k\log n}{L}}\left(\frac{(k+d)(d+1)\log n}{npL}\right)^{1/4}.
\end{align*}
\end{proof}

\subsection{Proof of Lemma \ref{alphadebiasapprox2}}\label{alphadebiasapprox2proof}
\begin{proof}
Since $\widehat{\alpha}_{R, i}$ is the minimizer of $\cL_{R,\tb_{R, -i}}(x)$, we know that
\begin{align*}
    0 &= \cL_{R,\tb_{R, -i}}'(\widehat{\alpha}_{R, i}) =\cL_{\tb_{R, -i}}'(\widehat{\alpha}_{R, i})+\tau\widehat{\alpha}_{R, i}+\lambda\partial |\widehat{\alpha}_{R, i}|\\
    &= \cL_{\tb_{R, -i}}'(\alpha_i^*) + \cL_{\tb_{R, -i}}''(b_1)(\widehat{\alpha}_{R, i}-\alpha_i^*)+\tau\widehat{\alpha}_{R, i}+\lambda\partial |\widehat{\alpha}_{R, i}|
\end{align*}
for some real number $b_1$ between $\alpha_i^*$ and $\widehat{\alpha}_{R, i}$. Reorganizing the terms gives
\begin{align*}
    \widehat{\alpha}_{R, i} = \alpha_i^* - \frac{ \cL_{\tb_{R, -i}}'(\alpha_i^*)+\tau\widehat{\alpha}_{R, i}+\lambda\partial |\widehat{\alpha}_{R, i}|}{\cL_{\tb_{R, -i}}''(b_1)}.
\end{align*}
Combine this with the definition of $ \widehat{\alpha}_{R,i}^{\debias}$ \eqref{debiasedestimator} gives
\begin{align*}
    \widehat{\alpha}_{R,i}^{\debias} = \alpha_i^* - \frac{ \cL_{\tb_{R, -i}}'(\alpha_i^*)}{\cL_{\tb_{R, -i}}''(b_1)}+\left(\tau\widehat{\alpha}_{R, i}+\lambda\partial |\widehat{\alpha}_{R, i}|\right)\left(\frac{1}{\left(\nabla^2\cL(\tb_R)\right)_{i,i}}- \frac{1}{\cL_{\tb_{R, -i}}''(b_1)}\right).
\end{align*}
Recall \eqref{alpharesidue1eq1} that $\dot{\alpha}_{R,i}^{\debias}$ could be written as
\begin{align*}
    \dot{\alpha}_{R,i}^{\debias} = \alpha_i^* - \frac{\left(\nabla \cL(\tb^*)\right)_i + \sum_{j\neq i}\left(\nabla^2\cL(\tb^*)\right)_{i,j}\left(\tb_{R,j} - \tb^*_j\right)}{\left(\nabla^2\cL(\tb^*)\right)_{i,i}}. 
\end{align*}
As a result, the difference $\widehat{\alpha}_{R,i}^{\debias} - \dot{\alpha}_{R,i}^{\debias}$ can be written as
\begin{align}
    \widehat{\alpha}_{R,i}^{\debias} - \dot{\alpha}_{R,i}^{\debias} =& \frac{\left(\nabla \cL(\tb^*)\right)_i + \sum_{j\neq i}\left(\nabla^2\cL(\tb^*)\right)_{i,j}\left(\tb_{R,j} - \tb^*_j\right)}{\left(\nabla^2\cL(\tb^*)\right)_{i,i}} - \frac{\cL_{\tb_{R, -i}}'(\alpha_i^*)}{\cL_{\tb_{R, -i}}''(b_1)} \nonumber \\
    &+ \left(\tau\widehat{\alpha}_{R, i}+\lambda\partial |\widehat{\alpha}_{R, i}|\right)\left(\frac{1}{\left(\nabla^2\cL(\tb_R)\right)_{i,i}}- \frac{1}{\cL_{\tb_{R, -i}}''(b_1)}\right). \label{alpharesidue2eq1}
\end{align}
We begin with controlling several terms in \eqref{alpharesidue2eq1}.

\noindent\textbf{Control $\left|\cL_{\tb_{R, -i}}''(b_1) - \left(\nabla^2\cL(\tb^*)\right)_{i,i}\right|$:} By the definition we know that  
\begin{align*}
    \left|\cL_{\tb_{R,-i}}''(b_1)-\left(\nabla^2\cL(\tb^*)\right)_{i,i}\right| =& \left|\sum_{j:(i,j)\in \mathcal{E}}\phi'(\tx_i^\top\tb_R-\tx_j^\top\tb_R+b_1-\widehat{\alpha}_{R,i})-\phi'(\tx_i^\top\tb^*-\tx_j^\top\tb^*)\right| \\
    \leq &\sum_{j:(i,j)\in \mathcal{E}}\left|(\tx_i^\top\tb_R-\tx_j^\top\tb_R+b_1-\widehat{\alpha}_{R,i})-(\tx_i^\top\tb^*-\tx_j^\top\tb^*)\right| \\
    \lesssim & \sum_{j:(i,j)\in \mathcal{E}}\left(\left\Vert \tb_R-\tb^*\right\Vert_c+\left| b_1-\widehat{\alpha}_{R,i} \right|\right)\lesssim  np\left\Vert \tb_R-\tb^*\right\Vert_c
\end{align*}
with probability at least $1-O(n^{-10})$.

\noindent\textbf{Control $\left|\cL_{\tb_{R, -i}}''(b_1) - \left(\nabla^2\cL(\tb_R)\right)_{i,i}\right|$:} By the definition we know that  
\begin{align*}
    \left|\cL_{\tb_{R,-i}}''(b_1)-\left(\nabla^2\cL(\tb_R)\right)_{i,i}\right| =& \left|\sum_{j:(i,j)\in \mathcal{E}}\phi'(\tx_i^\top\tb_R-\tx_j^\top\tb_R+b_1-\widehat{\alpha}_{R,i})-\phi'(\tx_i^\top\tb_R-\tx_j^\top\tb_R)\right| \\
    \leq &\sum_{j:(i,j)\in \mathcal{E}}\left|(\tx_i^\top\tb_R-\tx_j^\top\tb_R+b_1-\widehat{\alpha}_{R,i})-(\tx_i^\top\tb_R-\tx_j^\top\tb_R)\right| \\
    \lesssim & \sum_{j:(i,j)\in \mathcal{E}}\left| b_1-\widehat{\alpha}_{R,i} \right|\lesssim  np\left\Vert \tb_R-\tb^*\right\Vert_c
\end{align*}
with probability at least $1-O(n^{-10})$.

\noindent\textbf{Control $\left| \cL_{\tb_{R, -i}}'(\alpha_i^*) - \left(\nabla \cL(\tb^*)\right)_i - \sum_{j\neq i}\left(\nabla^2\cL(\tb^*)\right)_{i,j}\left(\tb_{R,j} - \tb^*_j\right)\right|$:} By definition one can write that
\begin{align}
    &\left| \cL_{\tb_{R, -i}}'(\alpha_i^*) - \left(\nabla \cL(\tb^*)\right)_i - \sum_{j\neq i}\left(\nabla^2\cL(\tb^*)\right)_{i,j}\left(\tb_{R,j} - \tb^*_j\right)\right| \nonumber \\
   =&\sum_{j:(i,j)\in \mathcal{E}}\left\{-y_{j,i}+\phi(\tx_i^\top\tb_R-\tx_j^\top\tb_R+\alpha^*_i-\widehat{\alpha}_{R,i})\right\}- \sum_{j:(i,j)\in \mathcal{E}}\left\{-y_{j,i}+\phi(\tx_i^\top\tb^*-\tx_j^\top\tb^*)\right\} \nonumber \\
&-\sum\limits_{j\neq i,j\in[n+d]}\left(\tb_{R,j}-\tb^*_j\right)\left(\nabla^2\cL(\tb^*)\right)_{i,j}  \nonumber \\
=& \sum_{j:(i,j)\in \mathcal{E}}\left\{\phi(\tx_i^\top\tb_R-\tx_j^\top\tb_R+\alpha^*_i-\widehat{\alpha}_{R,i})-\phi(\tx_i^\top\tb^*-\tx_j^\top\tb^*)-\left(\alpha^*_j-\widehat{\alpha}_{R,j}\right)\phi'(\tx_i^\top\tb^*-\tx_j^\top\tb^*)\right\} \nonumber \\
&-\sum_{k\in[d]}\left(\tb_{R,n+k}-\tb^*_{n+k}\right)\left(\sum_{j:(i,j)\in\mathcal{E}}\phi'(\tx_i^\top\tb^*-\tx_j^\top\tb^*)\left(\bx_i-\bx_j\right)_k\right)  \nonumber \\
=&\sum_{j:(i,j)\in \mathcal{E}} r_j, \label{alpharesidue2eq2}
\end{align}
where 
\begin{align*}
    r_j =& \phi(\tx_i^\top\tb_R-\tx_j^\top\tb_R+\alpha^*_i-\widehat{\alpha}_{R,i})-\phi(\tx_i^\top\tb^*-\tx_j^\top\tb^*)-\left(\alpha^*_j-\widehat{\alpha}_{R,j}\right)\phi'(\tx_i^\top\tb^*-\tx_j^\top\tb^*) \\
    &- (\bx_i-\bx_j)^\top(\widehat{\bb}_R-\bb^*)\phi'(\tx_i^\top\tb^*-\tx_j^\top\tb^*).
\end{align*}
On the other hand, by Taylor expansion we know that
\begin{align*}
    \phi(\tx_i^\top\tb_R-\tx_j^\top\tb_R+\alpha^*_i-\widehat{\alpha}_{R,i}) =& \phi(\tx_i^\top\tb^*-\tx_j^\top\tb^*) \\
    &+ \phi'(\tx_i^\top\tb^*-\tx_j^\top\tb^*) \left((\bx_i-\bx_j)^\top(\widehat{\bb}_R-\bb^*)+\alpha^*_j-\widehat{\alpha}_{R,j}\right) \\
    &+ \phi''(b_2)\left((\bx_i-\bx_j)^\top(\widehat{\bb}_R-\bb^*)+\alpha^*_j-\widehat{\alpha}_{R,j}\right)^2,
\end{align*}
where $b_2$ is some real number between $\tx_i^\top\tb^*-\tx_j^\top\tb^*$ and $\tx_i^\top\tb_R-\tx_j^\top\tb_R+\alpha^*_i-\widehat{\alpha}_{R,i}$. As a result, we have
\begin{align}
    |r_j|\leq |\phi''(b_2)|\left((\bx_i-\bx_j)^\top(\widehat{\bb}_R-\bb^*)+\alpha^*_j-\widehat{\alpha}_{R,j}\right)^2\lesssim \left\Vert \tb_R-\tb^*\right\Vert_c^2.\label{alpharesidue2eq3}
\end{align}
Plugging \eqref{alpharesidue2eq2} in \eqref{alpharesidue2eq3} gives us
\begin{align*}
    \left| \cL_{\tb_{R, -i}}'(\alpha_i^*) - \left(\nabla \cL(\tb^*)\right)_i - \sum_{j\neq i}\left(\nabla^2\cL(\tb^*)\right)_{i,j}\left(\tb_{R,j} - \tb^*_j\right)\right|\lesssim np \left\Vert \tb_R-\tb^*\right\Vert_c^2.
\end{align*}

\noindent \textbf{Control $\left|\left(\nabla \cL(\tb^*)\right)_i + \sum_{j\neq i}\left(\nabla^2\cL(\tb^*)\right)_{i,j}\left(\tb_{R,j} - \tb^*_j\right)\right|$:} One can see that 
\begin{align*}
    &\left|\left(\nabla\cL(\tb^*)\right)_i+\sum\limits_{j\neq i}\left(\tb_{R,j}-\tb^*_j\right)\left(\nabla^2\cL(\tb^*)\right)_{i,j}\right| \\
    \leq& \left|\left(\nabla\cL(\tb^*)\right)_i\right|+\left|\sum\limits_{j\neq i}\left(\tb_{R,j}-\tb^*_j\right)\left(\nabla^2\cL(\tb^*)\right)_{i,j}\right| \\
    \leq& \left|\left(\nabla\cL(\tb^*)\right)_i\right|+\Vert\widehat{\ba}_R-\ba^*\Vert_\infty\sum_{j\in [n],j\neq i}\left|\left(\nabla^2\cL(\tb^*)\right)_{i,j}\right|+\Vert\widehat{\bb}_R-\bb^*\Vert_2 \sqrt{\sum_{k>n}\left(\nabla^2\cL(\tb^*)\right)_{i,k}^2}\\
    \lesssim& \sqrt{\frac{np\log n}{L}}+\kappa_1^2\sqrt{\frac{(d+1)\log n}{npL}}np+\kappa_1\sqrt{\frac{\log n}{pL}}\sqrt{dnp^2}\lesssim \kappa_1^2\sqrt{\frac{(d+1)np\log n}{L}}
\end{align*}
with probability at least $1-O(n^{-6})$.

We come back to \eqref{alpharesidue2eq1}. The first term on the right hand side can be controlled as
\begin{align}
    &\left|\frac{\left(\nabla \cL(\tb^*)\right)_i + \sum_{j\neq i}\left(\nabla^2\cL(\tb^*)\right)_{i,j}\left(\tb_{R,j} - \tb^*_j\right)}{\left(\nabla^2\cL(\tb^*)\right)_{i,i}} - \frac{ \cL_{\tb_{R, -i}}'(\alpha_i^*)}{\cL_{\tb_{R, -i}}''(b_1)}\right| \nonumber \\
    \leq & \left|\frac{1}{\left(\nabla^2\cL(\tb^*)\right)_{i,i}} - \frac{1}{\cL_{\tb_{R, -i}}''(b_1)}\right|\left|\left(\nabla \cL(\tb^*)\right)_i + \sum_{j\neq i}\left(\nabla^2\cL(\tb^*)\right)_{i,j}\left(\tb_{R,j} - \tb^*_j\right)\right| \nonumber \\
    &+\frac{1}{\cL_{\tb_{R, -i}}''(b_1)}\left| \cL_{\tb_{R, -i}}'(\alpha_i^*) - \left(\nabla \cL(\tb^*)\right)_i - \sum_{j\neq i}\left(\nabla^2\cL(\tb^*)\right)_{i,j}\left(\tb_{R,j} - \tb^*_j\right)\right| \nonumber \\
    \lesssim & \frac{np\left\Vert \tb_R-\tb^*\right\Vert_c}{\frac{np}{\kappa_1}\left(\frac{np}{\kappa_1} - np\left\Vert \tb_R-\tb^*\right\Vert_c\right)}\left|\left(\nabla \cL(\tb^*)\right)_i + \sum_{j\neq i}\left(\nabla^2\cL(\tb^*)\right)_{i,j}\left(\tb_{R,j} - \tb^*_j\right)\right| \\
    &+ \frac{1}{\frac{np}{\kappa_1} - np\left\Vert \tb_R-\tb^*\right\Vert_c} np\left\Vert \tb_R-\tb^*\right\Vert_c^2 \nonumber \\
    \lesssim & \kappa_1^4\sqrt{\frac{(d+1)\log n}{npL}}\left\Vert \tb_R-\tb^*\right\Vert_c + \kappa_1^2\left\Vert \tb_R-\tb^*\right\Vert_c^2\lesssim \kappa_1^6 \frac{(d+1)\log n}{npL}.\label{alpharesidue2eq4}
\end{align}
On the other hand, we have
\begin{align}
    \left|\left(\tau\widehat{\alpha}_{R, i}+\lambda\partial |\widehat{\alpha}_{R, i}|\right)\left(\frac{1}{\left(\nabla^2\cL(\tb_R)\right)_{i,i}}- \frac{1}{\cL_{\tb_{R, -i}}''(b_1)}\right)\right|&\lesssim \frac{\lambda np\left\Vert \tb_R-\tb^*\right\Vert_c}{\frac{ np}{\kappa_1}\left(\frac{np}{\kappa_1} - np\left\Vert \tb_R-\tb^*\right\Vert_c\right)}  
    \nonumber \\
    &\lesssim \kappa_1^3\sqrt{\frac{(d+1)\log n}{npL}}\left\Vert \tb_R-\tb^*\right\Vert_c  \nonumber \\
    &\lesssim \kappa_1^5 \frac{(d+1)\log n}{npL}.\label{alpharesidue2eq5}
\end{align}
Plugging \eqref{alpharesidue2eq4} and \eqref{alpharesidue2eq5} in \eqref{alpharesidue2eq1} gives us
\begin{align*}
     |\widehat{\alpha}_{R,i}^{\debias} - \dot{\alpha}_{R,i}^{\debias}|\lesssim \kappa_1^6 \frac{(d+1)\log n}{npL}
\end{align*}
with probability at least $1-O(n^{-6})$.
\end{proof}

\subsection{Proof of Lemma \ref{alphadebiasapprox3}}\label{alphadebiasapprox3proof}
\begin{proof}
By definition we know that 
\begin{align*}
    \ocL_{\ob_{R, -i}}'(\alpha_{i}^*)-\left(\nabla\cL(\tb^*)\right)_i = \sum_{j\neq i}\left(\nabla^2\cL(\tb^*)\right)_{i,j}\left(\ob_{R,j} - \tb^*_j\right).
\end{align*}
According to Lemma \ref{oaestimation}, we have $\cS(\oa_R)\subset \cS(\ba^*)$. Therefore, the right hand side can be controlled as
\begin{align*}
    \left|\ocL_{\ob_{R, -i}}'(\alpha_{i}^*)-\left(\nabla\cL(\tb^*)\right)_i\right| &= \left|\sum_{j\in \cS(\ba^*)\cup\{n+1,n+2,\dots,n+d\}\backslash\{i\}}\left(\nabla^2\cL(\tb^*)\right)_{i,j}\left(\ob_{R,j} - \tb^*_j\right)\right| \\
    &\leq \sqrt{\sum_{j\in \cS(\ba^*)\cup\{n+1,n+2,\dots,n+d\}\backslash\{i\}}\left(\nabla^2\cL(\tb^*)\right)_{i,j}^2}\left\|\ob_R-\tb^*\right\|_2  \\
    &\leq \sqrt{np+ndp^2}\left\|\ob_R-\tb^*\right\|_2\lesssim \kappa_1^2(d+1)\sqrt{\frac{(k+d)\log n}{L}}.
\end{align*}
The last line follows from Lemma \ref{Lilemma} and Lemma \ref{oaestimation}. As a result, the approximation error can be controlled as 
\begin{align*}
    \left|\overline{\alpha}_{i,R}^\debias - \left(\alpha^*_i - \left(\nabla\cL(\tb^*)\right)_i/\left(\nabla^2\cL(\tb^*)\right)_{i,i}\right)\right|\lesssim \frac{\kappa_1^3(d+1)}{np}\sqrt{\frac{(k+d)\log n}{L}}.
\end{align*}
\end{proof}

\subsection{Proof of Lemma \ref{betadebiasapprox}}\label{betadebiasapproxproof}
\begin{proof}
According to Lemma \ref{oaestimation}, we have $\cS(\oa_R)\subset \cS(\ba^*)$. As a result, we have
\begin{align}
    \left\|\boldsymbol{B}\left(\oa_R-\ba^*\right)\right\|_\infty = \left\|\boldsymbol{B}_{:, \cS(\ba^*)}\left(\oa_R-\ba^*\right)_{\cS(\ba^*)}\right\|_\infty \leq \left\|\boldsymbol{B}_{:, \cS(\ba^*)}\right\|_{2,\infty}\left\|\oa_R-\ba^*\right\|_2. \label{betaapproxdebiaseq1}
\end{align}
By the definition of $\bB$ we know that
\begin{align*}
    \left\|\boldsymbol{B}_{:, \cS(\ba^*)}\right\|_{2,\infty}&\leq \left\|\boldsymbol{B}_{:, \cS(\ba^*)}\right\|_{F}\leq \sqrt{\sum_{i\in \cS(\ba^*)}\left\|\sum_{j:(i,j)\in \mathcal{E}}\phi'(\tx_i^\top\tb^*-\tx_j^\top\tb^*)(\bx_i-\bx_j)\right\|_2^2} \\
    &\lesssim \sqrt{k}np\sqrt{\frac{d+1}{n}}\lesssim \sqrt{k(d+1)np}.
\end{align*}
Plugging this in \eqref{betaapproxdebiaseq1}, by Lemma \ref{oaestimation} we have
\begin{align*}
    \left\|\boldsymbol{B}\left(\oa_R-\ba^*\right)\right\|_\infty\lesssim \kappa_1^2(d+1)\sqrt{\frac{k(k+d)\log n}{L}}.
\end{align*}
As a result, the approximation error can be controlled as
\begin{align*}
    \left\|\bA^{-1}\boldsymbol{B}\left(\oa_R-\ba^*\right)\right\|_2&\leq \left\|\bA^{-1}\right\|\left\|\boldsymbol{B}\left(\oa_R-\ba^*\right)\right\|_2\leq \sqrt{d}\left\|\bA^{-1}\right\|\left\|\boldsymbol{B}\left(\oa_R-\ba^*\right)\right\|_\infty \\
    &\lesssim \frac{\kappa_1^3(d+1)}{np}\sqrt{\frac{kd(k+d)\log n}{L}}.
\end{align*}
\end{proof}

\subsection{Proof of Theorem \ref{approxthm}}\label{approxthmproof}

\begin{proof}
Combine Lemma \ref{alphadebiasapprox1}, Lemma \ref{alphadebiasapprox2} and Lemma \ref{alphadebiasapprox3} we get
\begin{align}
    &\left|\widehat{\alpha}_{R,i}^{\debias} - \left(\alpha^*_i - \left(\nabla\cL(\tb^*)\right)_i/\left(\nabla^2\cL(\tb^*)\right)_{i,i}\right)\right| \nonumber \\
\leq & \left|\widehat{\alpha}_{R,i}^{\debias} - \dot{\alpha}_{R,i}^{\debias}\right| + \left|\dot{\alpha}_{R,i}^{\debias}-\overline{\alpha}_{R,i}^{\debias}\right| +\left|\overline{\alpha}_{R,i}^{\debias} - \left(\alpha^*_i - \left(\nabla\cL(\tb^*)\right)_i/\left(\nabla^2\cL(\tb^*)\right)_{i,i}\right)\right| \nonumber \\
\lesssim & \frac{\kappa_1^{4.5}(d+1)}{np}\sqrt{\frac{k\log n}{L}}\left(\frac{(k+d)(d+1)\log n}{npL}\right)^{1/4} + \kappa_1^6 \frac{(d+1)\log n}{npL} + \frac{\kappa_1^3(d+1)}{np}\sqrt{\frac{(k+d)\log n}{L}}  \nonumber \\
\lesssim & \frac{\kappa_1^3(d+1)}{np}\left(\frac{\kappa_1^3\log n}{L}+\sqrt{\frac{(k+d)\log n}{L}}\right). \label{approxthmproofeq1}
\end{align}

By Theorem \ref{residue2norm} we know that
\begin{align*}
     \left\|\widehat{\bb}_R^\debias - \overline{\bb}_{R, n+1:n+d}^\debias \right\|_2&= \left\|\bA^{-1}(\bA+\tau \bI_d)\left(\widehat{\bb}_R - \overline{\bb}_{R, n+1:n+d}\right)\right\|_2 \\
     &\leq \left\|\bA^{-1}\right\|\left\|\bA+\tau \bI_d\right\|\left\|\widehat{\bb}_R - \overline{\bb}_{R, n+1:n+d}\right\|_2 \lesssim \kappa_1^{4.5}\left(\frac{(k+d)(d+1)\log n}{npL}\right)^{3/4}.
\end{align*}
Combine this with Lemma \ref{betadebiasapprox} and Lemma \ref{betaapprox} we get 
\begin{align*}
    &\left\|\hat{\bb}_{R} -\left( \bb^*- \bA^{-1}\left(\nabla\cL(\tb^*)\right)_{n+1:n+d}\right)\right\|_2 \\
    \leq&\left\|\hat{\bb}_{R} -\hat{\bb}_{R}^\debias\right\|_2 +\left\|\hat{\bb}_{R}^\debias-\overline{\bb}_{R, n+1:n+d}^\debias\right\|_2+ \left\|\overline{\bb}_{R, n+1:n+d}^\debias -\left( \bb^*- \bA^{-1}\left(\nabla\cL(\tb^*)\right)_{n+1:n+d}\right)\right\|_2\\
    \lesssim& \frac{\kappa_1^3(d+1)}{np}\sqrt{\frac{kd(k+d)\log n}{L}} + \kappa_1^{4.5}\left(\frac{(k+d)(d+1)\log n}{npL}\right)^{3/4}.
\end{align*}
    
\end{proof}

\subsection{Proof of Theorem \ref{uqmainthm}}\label{uqmainthmproof}
\begin{proof}
We begin with the asymptotic distribution of $\widehat{\alpha}_{R,i}^\debias$. We let $\Delta\alpha_i = \left(\nabla\cL(\tb^*)\right)_i/\left(\nabla^2\cL(\tb^*)\right)_{i,i}$. For $j$ such that $(i,j)\in \mathcal{E}$ and $l\in [L]$, we define
\begin{align*}
    X_{j}^{(l)} = \frac{1}{\left(\nabla^2\cL(\tb^*)\right)_{i,i}L}\left\{y_{j,i}^{(l)}-\frac{e^{\tx_i^\top \tb}}{e^{\tx_i^\top \tb}+e^{\tx_j^\top \tb}}\right\}.
\end{align*}
In this case, we have $\Delta\alpha_i = \sum_{l=1}^L\sum_{j:(i,j)\in \mathcal{E}} X_{j}^{(l)}.$ One can see that
\begin{align*}
    \frac{\mathbb{E}\left|X_{j}^{(l)}\right|^3}{\mathbb{E}\left(X_{j}^{(l)}\right)^2} &= \frac{1}{\left(\nabla^2\cL(\tb^*)\right)_{i,i}L} \left(\phi(\tx_i^\top\tb^*-\tx_i^\top\tb^*)^2+(1-\phi(\tx_i^\top\tb^*-\tx_i^\top\tb^*))^2\right)\leq \frac{1}{\left(\nabla^2\cL(\tb^*)\right)_{i,i}L}
\end{align*}
conditioned on graph $\mathcal{G}$. On the other hand, since $\text{Var}[\nabla\cL(\tb^*)\mid \mathcal{G}] = \nabla^2\cL(\tb^*) / L$, we know that 
$\text{Var}[\Delta\alpha_i\mid \mathcal{G}] = 1/(\nabla^2\cL(\tb^*))_{i,i}L.$ As a result, by \cite{berry1941accuracy} we have
\begin{align*}
    \sup_{x\in \mathbb{R}}\left|\mathbb{P}\left(\sqrt{\left(\nabla^2\cL(\tb^*)\right)_{i,i}L}\Delta\alpha_i\leq x\mid\mathcal{G}\right) - \mathbb{P}(\mathcal{N}(0,1)\leq x)\right|\lesssim \frac{1}{\sqrt{\left(\nabla^2\cL(\tb^*)\right)_{i,i}L}}.
\end{align*}
Since $(\nabla^2\cL(\tb^*))_{i,i}\gtrsim np/\kappa_1$ with probability at least $1-O(n^{-10})$, we have
\begin{align*}
    &\sup_{x\in \mathbb{R}}\left|\mathbb{P}\left(\sqrt{\left(\nabla^2\cL(\tb^*)\right)_{i,i}L}\Delta\alpha_i\leq x\right) - \mathbb{P}(\mathcal{N}(0,1)\leq x)\right| \\
=&\sup_{x\in \mathbb{R}}\left|\mathbb{E}_{\mathcal{G}}\left[\mathbb{P}\left(\sqrt{\left(\nabla^2\cL(\tb^*)\right)_{i,i}L}\Delta\alpha_i\leq x\mid\mathcal{G}\right)\right] - \mathbb{P}(\mathcal{N}(0,1)\leq x)\right| \\
\leq & \mathbb{E}_{\mathcal{G}}\sup_{x\in \mathbb{R}}\left|\mathbb{P}\left(\sqrt{\left(\nabla^2\cL(\tb^*)\right)_{i,i}L}\Delta\alpha_i\leq x\mid\mathcal{G}\right) - \mathbb{P}(\mathcal{N}(0,1)\leq x)\right| \\
\lesssim & \sqrt{\frac{\kappa_1}{npL}}\cdot \mathbb{P}\left(\left(\nabla^2\cL(\tb^*)\right)_{i,i}\gtrsim np/\kappa_1\right) + 1\cdot \left(1-\mathbb{P}\left(\left(\nabla^2\cL(\tb^*)\right)_{i,i}\gtrsim np/\kappa_1\right)\right)\lesssim \sqrt{\frac{\kappa_1}{npL}}.
\end{align*}
For simplicity we let $$\Gamma = \frac{\kappa_1^3(d+1)}{\sqrt{np}}\left(\frac{\kappa_1^3\log n}{\sqrt{L}}+\sqrt{(k+d)\log n}\right).$$ Consider event $A = \left\{\left|\sqrt{(\nabla^2\cL(\tb^*))_{i,i}L}(\widehat{\alpha}_{R,i}^{\debias}-\alpha_i^*+\Delta\alpha_i)\right|\leq \Lambda\Gamma\right\}$, where $\Lambda>0$ is some constant such that $\mathbb{P}(A^c) = O(n^{-6})$. Then for any fixed $x\in \mathbb{R}$, we consider the following three events
\begin{align*}
    B_1 &= \left\{\sqrt{(\nabla^2\cL(\tb^*))_{i,i}L}(\widehat{\alpha}_{R,i}^{\debias}-\alpha_i^*)\leq x\right\}, \\
    B_2 &= \left\{\sqrt{(\nabla^2\cL(\tb^*))_{i,i}L}\Delta\alpha_i\leq x-\Lambda\Gamma\right\}, \\
    B_3 &= \left\{\sqrt{(\nabla^2\cL(\tb^*))_{i,i}L}\Delta\alpha_i\leq x+\Lambda\Gamma\right\}.
\end{align*}
Then we have
\begin{align}
&\left|\mathbb{P}\left(\sqrt{\left(\nabla^2\cL(\tb^*)\right)_{i,i}L}(\widehat{\alpha}_{R,i}^{\debias}-\alpha_i^*)\leq x\right) - \mathbb{P}(\mathcal{N}(0,1)\leq x)\right| \nonumber \\
=&\left|\mathbb{P}(B_1\cap A)+\mathbb{P}(B_1\cap A^c) - \mathbb{P}(\mathcal{N}(0,1)\leq x)\right|\lesssim \left|\mathbb{P}(B_1\cap A) - \mathbb{P}(\mathcal{N}(0,1)\leq x)\right| + \frac{1}{n^6}. \label{asymptoticeq1}
\end{align}
On the other hand, since $B_2\cap A \subset B_1\cap A \subset B_3\cap A$, we know that 
\begin{align*}
    \left|\mathbb{P}(B_1\cap A) - \mathbb{P}(\mathcal{N}(0,1)\leq x)\right|\leq\max\{\left|\mathbb{P}(B_2\cap A) - \mathbb{P}(\mathcal{N}(0,1)\leq x)\right|, \left|\mathbb{P}(B_3\cap A) - \mathbb{P}(\mathcal{N}(0,1)\leq x)\right|\}.
\end{align*}
One can see that
\begin{align*}
    &\left|\mathbb{P}(B_2\cap A) - \mathbb{P}(\mathcal{N}(0,1)\leq x)\right| = \left|\mathbb{P}(B_2) - \mathbb{P}(B_2\cap A^c) - \mathbb{P}(\mathcal{N}(0,1)\leq x)\right| \\
    \leq & \left|\mathbb{P}(B_2) - \mathbb{P}(\mathcal{N}(0,1)\leq x)\right| +\mathbb{P}(A^c) \\
    \leq & \left|\mathbb{P}(B_2) - \mathbb{P}(\mathcal{N}(0,1)\leq  x-\Lambda\Gamma)\right| + \left|\mathbb{P}(\mathcal{N}(0,1)\leq x) - \mathbb{P}(\mathcal{N}(0,1)\leq x-\Lambda\Gamma)\right| +\mathbb{P}(A^c) \\
    \lesssim & \Gamma +\sqrt{\frac{\kappa_1}{npL}} +\frac{1}{n^6}\lesssim \Gamma.
\end{align*}
Similarly, one can show that $\left|\mathbb{P}(B_3\cap A) - \mathbb{P}(\mathcal{N}(0,1)\leq x)\right|\lesssim \Gamma$. Therefore, we have $$\left|\mathbb{P}(B_1\cap A) - \mathbb{P}(\mathcal{N}(0,1)\leq x)\right|\lesssim\Gamma.$$ Plugging this back to \eqref{asymptoticeq1}, and by the arbitrariness of $x$, we have
\begin{align*}
    \sup_{x\in \mathbb{R}}\left|\mathbb{P}\left(\sqrt{\left(\nabla^2\cL(\tb^*)\right)_{i,i}L}(\widehat{\alpha}_{R,i}^{\debias}-\alpha_i^*)\leq x\right) - \mathbb{P}(\mathcal{N}(0,1)\leq x)\right|\lesssim \frac{\kappa_1^3(d+1)}{\sqrt{np}}\left(\frac{\kappa_1^3\log n}{\sqrt{L}}+\sqrt{(k+d)\log n}\right).
\end{align*}

Next we focus on the asymptotic distribution of $\widehat{\beta}_{R,k}$. Similarly, we define 
\begin{align*}
    \Delta \beta_k = \left[\bA^{-1}\left(\nabla\cL(\tb^*)\right)_{n+1:n+d}\right]_k.
\end{align*}
For $(i,j)$ such that $(i,j)\in\mathcal{E}, i>j$ and $l\in[L]$, we define 
\begin{align*}
	X_{i,j}^{(l)} = \frac{1}{L}\left\{y_{j,i}^{(l)}-\frac{e^{\tx_i^\top \tb}}{e^{\tx_i^\top \tb}+e^{\tx_j^\top \tb}}\right\}\left(\bA^{-1}(\bx_i-\bx_j)\right)_k.
\end{align*}
In this case, we have $\Delta\beta_k = \sum_{l=1}^L\sum_{(i,j)\in \mathcal{E}} X_{i,j}^{(l)}.$ One can see that
\begin{align*}
    \frac{\mathbb{E}\left|X_{i,j}^{(l)}\right|^3}{\mathbb{E}\left(X_{i,j}^{(l)}\right)^2} &= \frac{\left|\left(\bA^{-1}(\bx_i-\bx_j)\right)_k\right|}{L} \left(\phi(\tx_i^\top\tb^*-\tx_i^\top\tb^*)^2+(1-\phi(\tx_i^\top\tb^*-\tx_i^\top\tb^*))^2\right)\leq \frac{\left|\left(\bA^{-1}(\bx_i-\bx_j)\right)_k\right|}{L}
\end{align*}
conditioned on graph $\mathcal{G}$. On the other hand, since $\text{Var}[(\nabla\cL(\tb^*))_{n+1:n+d}\mid \mathcal{G}] = \bA / L$, we know that 
$\text{Var}[\Delta\alpha_i\mid \mathcal{G}] = (\bA^{-1})_{k,k}/L$. As a result, by \cite{berry1941accuracy} we have
\begin{align*}
    \sup_{x\in \mathbb{R}}\left|\mathbb{P}\left(\frac{\sqrt{L}\Delta\beta_k}{\sqrt{(\bA^{-1})_{k,k}}}\leq x\mid\mathcal{G}\right) - \mathbb{P}(\mathcal{N}(0,1)\leq x)\right|\lesssim \frac{\sup_{(i,j)\in \mathcal{E}}\left|\left(\bA^{-1}(\bx_i-\bx_j)\right)_k\right| / L}{\sqrt{(\bA^{-1})_{k,k}/L}}.
\end{align*}
Since $\left|\left(\bA^{-1}(\bx_i-\bx_j)\right)_k\right|\lesssim \|\bA^{-1}\|\lesssim \kappa_1/np$ and $(\bA^{-1})_{k,k}\gtrsim 1/np$, we know that
\begin{align*}
    \sup_{x\in \mathbb{R}}\left|\mathbb{P}\left(\frac{\sqrt{L}\Delta\beta_k}{\sqrt{(\bA^{-1})_{k,k}}}\leq x\mid\mathcal{G}\right) - \mathbb{P}(\mathcal{N}(0,1)\leq x)\right|\lesssim \frac{\kappa_1}{\sqrt{npL}}
\end{align*}
with probability at least $1-O(n^{-10})$ (randomness comes from $\mathcal{G}$). Similar to the discussion of $\Delta\alpha_i$ before, one can further get 
\begin{align*}
    \sup_{x\in \mathbb{R}}\left|\mathbb{P}\left(\frac{\sqrt{L}\Delta\beta_k}{\sqrt{(\bA^{-1})_{k,k}}}\leq x\right) - \mathbb{P}(\mathcal{N}(0,1)\leq x)\right|\lesssim \frac{\kappa_1}{\sqrt{npL}}.
\end{align*}
Then we mimic the proof of the $\alpha$ counterpart again and one can show that
\begin{align*}
    &\sup_{x\in \mathbb{R}}\left|\mathbb{P}\left(\frac{\sqrt{L}\left(\hat{\beta}_{R, k}-\beta_k^*\right)}{\sqrt{(\bA^{-1})_{k,k}}}\leq x\right) - \mathbb{P}(\mathcal{N}(0,1)\leq x)\right|\\
    \lesssim & \frac{\kappa_1^3(d+1)\sqrt{kd(k+d)\log n}}{\sqrt{np}} + \frac{\kappa_1^{4.5}((k+d)(d+1)\log n)^{3/4}}{(npL)^{1/4}}.
\end{align*}
\end{proof}

\subsection{Proof of Theorem \ref{T1thm}}\label{T1thmproof}

\begin{proof}
    By Lemma \ref{Bootstraplemma1} we have
\begin{align*}
    \sup_{z\in \mathbb{R}}\left|P(\cG_1^\sharp\leq z)-P(\cT_1^\sharp\leq z)\right|  &= \sup_{z\in \mathbb{R}}\left|\mathbb{E}P(\cG_1^\sharp\leq z|\cE)-\mathbb{E}P(\cT_1^\sharp\leq z|\cE)\right|  \nonumber \\
    &\leq \mathbb{E}\sup_{z\in \mathbb{R}}\left|P(\cG_1^\sharp\leq z|\cE)-P(\cT_1^\sharp\leq z|\cE)\right|  \nonumber \\
    &\lesssim \left(\frac{\log^5 n}{np}\right)^{1/4}+\frac{1}{n^{10}}\lesssim \left(\frac{\log^5 n}{np}\right)^{1/4}.
\end{align*}
Combine this with Lemma \ref{Bootstraplemma2} and Lemma \ref{Bootstraplemma3} we get
\begin{align*}
    \sup_{z\in \mathbb{R}}\left|P(\cG_1\leq z)-P(\cT_1\leq z)\right|\lesssim \left(\frac{\log^5 n}{np}\right)^{1/4}  + \frac{\kappa_1^3(d+1)\log n}{\sqrt{np}}\left(\kappa_1^3\sqrt{\frac{\log n}{L}}+\sqrt{k+d}\right).
\end{align*}
Since $P(\cG_1\leq c_{1-\alpha}) = 1-\alpha$, we know that 
\begin{align*}
    \left|P(\cT_1>c_{1-\alpha})-\alpha\right|\lesssim \left(\frac{\log^5 n}{np}\right)^{1/4}  + \frac{\kappa_1^3(d+1)\log n}{\sqrt{np}}\left(\kappa_1^3\sqrt{\frac{\log n}{L}}+\sqrt{k+d}\right).
\end{align*}
\end{proof}
We define 
\begin{align*}
    \mathcal{T}_1^\sharp &= \max_{i\in [n]}\left|\sum_{l=1}^L \sum_{(i, j)\in\mathcal{E}}\sqrt{\frac{1}{\left(\nabla^2\cL(\tb^*)\right)_{i,i}L}}\left(\phi(\tx_i^\top\tb^*-\tx_j^\top\tb^*)-y_{j, i}^{(l)}\right)\right|, \\
    \mathcal{G}_1^\sharp &= \max_{i\in [n]}\left|\sum_{l=1}^L\sum_{(i,j)\in\mathcal{E}}\sqrt{\frac{1}{\left(\nabla^2\cL(\tb^*)\right)_{i,i}L}}\left(\phi(\tx_i^\top\tb^*-\tx_j^\top\tb^*)-y_{j,i}^{(l)}\right)\omega_{j,i}^{(l)}\right|,
\end{align*}
and let $c_{1-\alpha}^\sharp$ be the $(1-\alpha)$-th quantile of $\cG_1^\sharp$ conditioned on $\cE$ and $\{y_{j, i}:1\leq j<i\leq n\}$. Let $\bZ = (Z_1,Z_2,\dots,Z_n)$ be a random vector such that $\bZ|\cE$ is a Gaussian random vector and 
\begin{align*}
    \textbf{cov}(Z_i,Z_j|\cE) = \textbf{cov}\left(\frac{\sqrt{L}\left(\nabla\cL(\tb^*)\right)_i}{\sqrt{\left(\nabla^2\cL(\tb^*)\right)_{i,i}}},\frac{\sqrt{L}\left(\nabla\cL(\tb^*)\right)_j}{\sqrt{\left(\nabla^2\cL(\tb^*)\right)_{j, j}}}\mid \cE\right),\quad \forall i,j\in [n].
\end{align*}

\begin{lemma}\label{Bootstraplemma1}
Under the conditions of Theorem \ref{uqmainthm} and under the event $\cA_2$, we have
\begin{align*}
    &\sup_{z\in \mathbb{R}}\left|P(\cT_1^\sharp\leq z|\cE)-P\left(\max_{i\in [n]}|Z_i|\leq z|\cE\right)\right| \lesssim \left(\frac{\log^5 n}{np}\right)^{1/4},\\
    &\sup_{z\in \mathbb{R}}\left|P(\cT_1^\sharp\leq z|\cE)-P(\cG_1^\sharp \leq z|\cE)\right|\lesssim \left(\frac{\log^5 n}{np}\right)^{1/4}.
\end{align*}
Furthermore, we also have
\begin{align*}
    \sup_{z\in \mathbb{R}}\left|P(\cG_1^\sharp\leq z|\cE)-P\left(\max_{i\in [n]}|Z_i|\leq z|\cE\right)\right| \lesssim \left(\frac{\log^5 n}{np}\right)^{1/4}.
\end{align*}
\end{lemma}
\begin{proof}
The \cite[Condition E,M]{chernozhuokov2022improved} holds with $b_1 = b_2= 1$ and $B_n\asymp \sqrt{|\cE| / (\nabla^2\cL(\tb^*))_{i,i}}$. By \cite[Theorem 2.1, Theorem 2.2]{chernozhuokov2022improved} we know that 
    \begin{align*}
       &\sup_{z\in \mathbb{R}}\left|P(\cT_1^\sharp\leq z|\cE)-P\left(\max_{i\in [n]}|Z_i|\leq z|\cE\right)\right|\lesssim \left(\frac{\log^5(nL|\cE|)}{L|\cE|}\cdot \frac{|\cE|}{\left(\nabla^2\cL(\tb^*)\right)_{i,i}}\right)^{1/4}\lesssim \left(\frac{\log^5 n}{np}\right)^{1/4}, \\
       &\sup_{z\in \mathbb{R}}\left|P(\cT_1^\sharp\leq z|\cE)-P(\cG_1^\sharp \leq z|\cE)\right|\lesssim \left(\frac{\log^5(nL|\cE|)}{L|\cE|}\cdot \frac{|\cE|}{\left(\nabla^2\cL(\tb^*)\right)_{i,i}}\right)^{1/4}\lesssim \left(\frac{\log^5 n}{np}\right)^{1/4}.
    \end{align*}
\end{proof}

\begin{lemma}\label{Bootstraplemma2}
Under the conditions of Theorem \ref{uqmainthm}, we have
\begin{align*}
    \sup_{z\in \mathbb{R}}|P(\cT_1\leq z) - P(\cT_1^\sharp\leq z)|\lesssim \left(\frac{\log^5 n}{np}\right)^{1/4} + \frac{\kappa_1^3(d+1)\log n}{\sqrt{np}}\left(\kappa_1^3\sqrt{\frac{\log n}{L}}+\sqrt{k+d}\right).
\end{align*}
\end{lemma}
\begin{proof}
By definition of $\cT$ and $\cT^\sharp$ we know that
\begin{align*}
    |\cT_1-\cT_1^\sharp|\leq\max_{i\in [n]}\left|\sqrt{\left(\nabla^2\cL(\tb_R)\right)_{i,i}L}\widehat{\alpha}_{R,i}^{\debias} + \left(\nabla\cL(\tb^*)\right)_i\sqrt{L / \left(\nabla^2\cL(\tb^*)\right)_{i, i}}\right|.
\end{align*}
Under the null hypothesis we have
\begin{align}
&\left|\sqrt{\left(\nabla^2\cL(\tb_R)\right)_{i,i}L}\widehat{\alpha}_{R,i}^{\debias} + \left(\nabla\cL(\tb^*)\right)_i\sqrt{L / \left(\nabla^2\cL(\tb^*)\right)_{i, i}}\right| \nonumber \\
\leq & \sqrt{\left(\nabla^2\cL(\tb_R)\right)_{i,i}L}\left|\widehat{\alpha}_{R,i}^{\debias} - \left(\alpha^*_i - \left(\nabla\cL(\tb^*)\right)_i/\left(\nabla^2\cL(\tb^*)\right)_{i,i}\right)\right| \nonumber  \\
&+\left|\left(\nabla\cL(\tb^*)\right)_i\sqrt{L / \left(\nabla^2\cL(\tb^*)\right)_{i, i}}\right|\left|1-\sqrt{\frac{\left(\nabla^2\cL(\tb_R)\right)_{i, i}}{\left(\nabla^2\cL(\tb^*)\right)_{i, i}}}\right|.\label{TTdiffeq1}
\end{align}
According to \eqref{approxthmproof}, the first term on the right hand side of \eqref{TTdiffeq1} can be bounded as
\begin{align}
&\sqrt{\left(\nabla^2\cL(\tb_R)\right)_{i,i}L}\left|\widehat{\alpha}_{R,i}^{\debias} - \left(\alpha^*_i - \left(\nabla\cL(\tb^*)\right)_i/\left(\nabla^2\cL(\tb^*)\right)_{i,i}\right)\right|  \nonumber \\
\lesssim &  \sqrt{\left(\nabla^2\cL(\tb_R)\right)_{i,i}L}\frac{\kappa_1^3(d+1)}{np}\left(\frac{\kappa_1^3\log n}{L}+\sqrt{\frac{(k+d)\log n}{L}}\right) \nonumber \\
\lesssim & \frac{\kappa_1^3(d+1)}{\sqrt{np}}\left(\frac{\kappa_1^3\log n}{\sqrt{L}}+\sqrt{(k+d)\log n}\right)\label{TTdiffeq2}
\end{align}
with probability at least $1-O(n^{-6})$. When it comes to the second term, by Lemma \ref{Lilemma} we have
\begin{align}
    \left|\left(\nabla\cL(\tb^*)\right)_i\sqrt{L / \left(\nabla^2\cL(\tb^*)\right)_{i, i}}\right|\lesssim \sqrt{\frac{np\log n}{L}}\sqrt{\kappa_1 L/np}\lesssim \sqrt{\kappa_1\log n}\label{TTdiffeq3}
\end{align}
with probability at least $1-O(n^{-10})$. And, one can see that
\begin{align}
    &\left|1-\sqrt{\frac{\left(\nabla^2\cL(\tb_R)\right)_{i, i}}{\left(\nabla^2\cL(\tb^*)\right)_{i, i}}}\right|\leq \left|1-\frac{\left(\nabla^2\cL(\tb_R)\right)_{i, i}}{\left(\nabla^2\cL(\tb^*)\right)_{i, i}}\right| \nonumber \\
    \lesssim & \frac{\sum_{j:(i,j)\in \cE}|\phi'(\tx_i^\top\tb^* - \tx_j^\top\tb^*)-\phi'(\tx_i^\top\tb_R - \tx_j^\top\tb_R)|}{np/\kappa_1}  \nonumber \\
    \lesssim & \frac{\sum_{j:(i,j)\in \cE}|(\tx_i^\top\tb^* - \tx_j^\top\tb^*)-(\tx_i^\top\tb_R - \tx_j^\top\tb_R)|}{np/\kappa_1}\lesssim \kappa_1^3\sqrt{\frac{(d+1)\log n}{npL}}.\label{TTdiffeq4}
\end{align}
with probability at least $1-O(n^{-10})$. Plugging \eqref{TTdiffeq2}, \eqref{TTdiffeq3} and \eqref{TTdiffeq4} in \eqref{TTdiffeq1} we get 
\begin{align*}
    \left|\sqrt{\left(\nabla^2\cL(\tb_R)\right)_{i,i}L}\widehat{\alpha}_{R,i}^{\debias} + \left(\nabla\cL(\tb^*)\right)_i\sqrt{L / \left(\nabla^2\cL(\tb^*)\right)_{i, i}}\right|\lesssim \frac{\kappa_1^3(d+1)}{\sqrt{np}}\left(\frac{\kappa_1^3\log n}{\sqrt{L}}+\sqrt{(k+d)\log n}\right)
\end{align*}
with probability at least $1-O(n^{-6})$. As a result, we know that
\begin{align}
    |\cT_1-\cT_1^\sharp|\lesssim \frac{\kappa_1^3(d+1)}{\sqrt{np}}\left(\frac{\kappa_1^3\log n}{\sqrt{L}}+\sqrt{(k+d)\log n}\right)\label{TTdiffeq5}
\end{align}
with probability at least $1-O(n^{-6})$. We let 
\begin{align*}
    \delta\asymp \frac{\kappa_1^3(d+1)}{\sqrt{np}}\left(\frac{\kappa_1^3\log n}{\sqrt{L}}+\sqrt{(k+d)\log n}\right).
\end{align*}
By Lemma \ref{Bootstraplemma1} we have
\begin{align}
    \sup_{z\in \mathbb{R}}\left|P(\cT_1^\sharp\leq z)-P\left(\max_{i\in [n]}|Z_i|\leq z\right)\right|  &= \sup_{z\in \mathbb{R}}\left|\mathbb{E}P(\cT_1^\sharp\leq z|\cE)-\mathbb{E}P\left(\max_{i\in [n]}|Z_i|\leq z|\cE\right)\right|  \nonumber \\
    &\leq \mathbb{E}\sup_{z\in \mathbb{R}}\left|P(\cT_1^\sharp\leq z|\cE)-P\left(\max_{i\in [n]}|Z_i|\leq z|\cE\right)\right|  \nonumber \\
    &\lesssim \left(\frac{\log^5 n}{np}\right)^{1/4}+\frac{1}{n^{10}}\lesssim \left(\frac{\log^5 n}{np}\right)^{1/4}.\label{TTdiffeq6}
\end{align}
Therefore, by \eqref{TTdiffeq5} and \eqref{TTdiffeq6} we can write
\begin{align}
    &\sup_{z\in \mathbb{R}}|P(\cT_1\leq z) - P(\cT_1^\sharp\leq z)|\leq P(|\cT_1-\cT_1^\sharp|>\delta) + \sup_{z\in \mathbb{R}}P(z < \cT_1^\sharp \leq z+\delta) \nonumber \\
    \lesssim & \frac{1}{n^6}+ \left(\frac{\log^5 n}{np}\right)^{1/4} + \sup_{z\in \mathbb{R}}P\left(z < \max_{i\in [n]}|Z_i|\leq z \leq z+\delta\right). \label{TTdiffeq7}
\end{align}
By \cite[Theorem 3]{chernozhukov2015comparison} the last term on the right hand side can be controlled as 
\begin{align*}
    \sup_{z\in \mathbb{R}}P\left(z < \max_{i\in [n]}|Z_i|\leq z \leq z+\delta\right) &= \sup_{z\in \mathbb{R}}\mathbb{E}P\left(z < \max_{i\in [n]}|Z_i|\leq z \leq z+\delta|\cE\right) \\
    &\lesssim \mathbb{E}\sup_{z\in \mathbb{R}}P\left(z < \max_{i\in [n]}|Z_i|\leq z \leq z+\delta|\cE\right) \\
    &\lesssim \sqrt{\log n}\delta.
\end{align*}
Plugging this in \eqref{TTdiffeq7} we get 
\begin{align*}
    \sup_{z\in \mathbb{R}}|P(\cT_1\leq z) - P(\cT_1^\sharp\leq z)|\lesssim \left(\frac{\log^5 n}{np}\right)^{1/4} + \frac{\kappa_1^3(d+1)\log n}{\sqrt{np}}\left(\kappa_1^3\sqrt{\frac{\log n}{L}}+\sqrt{k+d}\right).
\end{align*}
\end{proof}

\begin{lemma}\label{Bootstraplemma3}  
Under the conditions of Theorem \ref{uqmainthm}, we have
\begin{align*}
     \sup_{z\in \mathbb{R}}|P(\cG_1\leq z) - P(\cG_1^\sharp\leq z)|\lesssim \left(\frac{\log^5 n}{np}\right)^{1/4}+\kappa_1^{3.5}\frac{\sqrt{d+1}\log^{1.5} n}{\sqrt{npL}}.
\end{align*}
\end{lemma}
\begin{proof}
By definition we have
\begin{align}
    |\cG_1-\cG_1^\sharp|\leq \max_{i\in [n]}\left|\sum_{(i,j)\in \cE}\sum_{l=1}^L\Delta_{j, i}^{(l)}\omega_{j,i}^{(l)}\right|, \label{Bootstraplemma3eq2}
\end{align}
where 
\begin{align*}
    \Delta_{j, i}^{(l)}:= \sqrt{\frac{1}{\left(\nabla^2\cL(\tb_R)\right)_{i,i}L}}\left(\phi(\tx_i^\top\tb_R-\tx_j^\top\tb_R)-y_{j, i}^{(l)}\right) - \sqrt{\frac{1}{\left(\nabla^2\cL(\tb^*)\right)_{i,i}L}}\left(\phi(\tx_i^\top\tb^*-\tx_j^\top\tb^*)-y_{j, i}^{(l)}\right).
\end{align*}
$\Delta_{j,i}^{(l)}$ can be controlled as
\begin{align}
    |\Delta_{j, i}^{(l)}|\leq &\left|\sqrt{\frac{1}{\left(\nabla^2\cL(\tb_R)\right)_{i,i}L}} - \sqrt{\frac{1}{\left(\nabla^2\cL(\tb^*)\right)_{i,i}L}}\right|\left|\phi(\tx_i^\top\tb^*-\tx_j^\top\tb^*)-y_{j, i}^{(l)}\right|  \nonumber \\
    &+\sqrt{\frac{1}{\left(\nabla^2\cL(\tb_R)\right)_{i,i}L}}\left|\phi(\tx_i^\top\tb_R-\tx_j^\top\tb_R) - \phi(\tx_i^\top\tb^*-\tx_j^\top\tb^*)\right|.\label{Bootstraplemma3eq1}
\end{align}
Since $1/\sqrt{a}-1/\sqrt{b} = (b-a)/(b\sqrt{a}+a\sqrt{b} )$, we know that
\begin{align*}
    \left|\sqrt{\frac{1}{\left(\nabla^2\cL(\tb_R)\right)_{i,i}L}} - \sqrt{\frac{1}{\left(\nabla^2\cL(\tb^*)\right)_{i,i}L}}\right|&\lesssim \frac{\left|\left(\nabla^2\cL(\tb^*)\right)_{i,i} - \left(\nabla^2\cL(\tb_R)\right)_{i,i}\right|}{(np/\kappa_1)^{1.5}\sqrt{L}} \\
    &\lesssim \frac{np \kappa_1^2\sqrt{(d+1)\log n/npL}}{(np/\kappa_1)^{1.5}\sqrt{L}}\lesssim \frac{\kappa_1^{3.5}\sqrt{(d+1)\log n}}{npL}.
\end{align*}
Plugging this in \eqref{Bootstraplemma3eq1}, we get
\begin{align*}
    |\Delta_{j, i}^{(l)}|&\lesssim \frac{\kappa_1^{3.5}\sqrt{(d+1)\log n}}{npL}\cdot 1 + \sqrt{\frac{\kappa_1}{npL}}\left|(\tx_i^\top\tb_R-\tx_j^\top\tb_R) - (\tx_i^\top\tb^*-\tx_j^\top\tb^*)\right| \\
    &\lesssim \frac{\kappa_1^{3.5}\sqrt{(d+1)\log n}}{npL} + \sqrt{\frac{\kappa_1}{npL}}\kappa_1^2\sqrt{\frac{(d+1)\log n}{npL}}\lesssim \kappa_1^{3.5}\frac{\sqrt{(d+1)\log n}}{npL}
\end{align*}
for all $(i,j)\in \cE$ with probability at least $1-O(n^{-10})$. Plugging this in \eqref{Bootstraplemma3eq2}, we know that
\begin{align*}
     |\cG_1-\cG_1^\sharp|\lesssim \max_{i\in [n]}\sqrt{\sum_{(i,j)\in \cE}\sum_{l=1}^L\left(\Delta_{j,i}^{(l)}\right)^2\log n}\lesssim \kappa_1^{3.5}\frac{\sqrt{d+1}\log n}{\sqrt{npL}}
\end{align*}
with probability at least $1-O(n^{-10})$.

Next, we let 
\begin{align}
    \delta\asymp \kappa_1^{3.5}\frac{\sqrt{d+1}\log n}{\sqrt{npL}}.\label{BootstrapLemma3eq3}
\end{align}
By Lemma \ref{Bootstraplemma1} we have
\begin{align*}
    \sup_{z\in \mathbb{R}}\left|P(\cG_1^\sharp\leq z)-P\left(\max_{i\in [n]}|Z_i|\leq z\right)\right|  &= \sup_{z\in \mathbb{R}}\left|\mathbb{E}P(\cG_1^\sharp\leq z|\cE)-\mathbb{E}P\left(\max_{i\in [n]}|Z_i|\leq z|\cE\right)\right|  \nonumber \\
    &\leq \mathbb{E}\sup_{z\in \mathbb{R}}\left|P(\cG_1^\sharp\leq z|\cE)-P\left(\max_{i\in [n]}|Z_i|\leq z|\cE\right)\right|  \nonumber \\
    &\lesssim \left(\frac{\log^5 n}{np}\right)^{1/4}+\frac{1}{n^{10}}\lesssim \left(\frac{\log^5 n}{np}\right)^{1/4}.
\end{align*}
Then we write
\begin{align}
    \sup_{z\in \mathbb{R}}|P(\cG_1\leq z) - P(\cG_1^\sharp\leq z)|&\leq P(|\cG_1-\cG_1^\sharp|>\delta) + \sup_{z\in \mathbb{R}}P(z < \cG_1^\sharp \leq z+\delta)  \nonumber\\
    &\lesssim \frac{1}{n^{10}}+\left(\frac{\log^5 n}{np}\right)^{1/4}+\sup_{z\in \mathbb{R}}P\left(z < \max_{i\in [n]}|Z_i| \leq z+\delta\right).\label{BootstrapLemma3eq4}
\end{align}
By \cite[Theorem 3]{chernozhukov2015comparison} the last term on the right hand side can be controlled as 
\begin{align}
    \sup_{z\in \mathbb{R}}P\left(z < \max_{i\in [n]}|Z_i|\leq z +\delta\right) &= \sup_{z\in \mathbb{R}}\mathbb{E}P\left(z < \max_{i\in [n]}|Z_i|\leq z +\delta|\cE\right) \nonumber \\
    &\lesssim \mathbb{E}\sup_{z\in \mathbb{R}}P\left(z < \max_{i\in [n]}|Z_i|\leq z+\delta|\cE\right) \nonumber \\
    &\lesssim \sqrt{\log n}\delta.\label{BootstrapLemma3eq5}
\end{align}
Plugging \eqref{BootstrapLemma3eq3} and \eqref{BootstrapLemma3eq5} in \eqref{BootstrapLemma3eq4} we know that
\begin{align*}
     \sup_{z\in \mathbb{R}}|P(\cG_1\leq z) - P(\cG_1^\sharp\leq z)|\lesssim \left(\frac{\log^5 n}{np}\right)^{1/4}+\kappa_1^{3.5}\frac{\sqrt{d+1}\log^{1.5} n}{\sqrt{npL}}.
\end{align*}
\end{proof}

\subsection{Proof of Theorem \ref{T23thm}}\label{T23thmproof}
\begin{proof}
    Let $(\cT, \cG, \cT^\sharp, \cG^\sharp, \cQ^\sharp)$ be any one of the three pairs: $(\cT_t, \cG_t, \cT^\sharp_t, \cG^\sharp_t, \cQ^\sharp_t),t=2,3$. By Lemma \ref{Bootstraplemma4} we have
\begin{align*}
    \sup_{z\in \mathbb{R}}\left|P(\cG^\sharp\leq z)-P(\cT^\sharp\leq z)\right|  &= \sup_{z\in \mathbb{R}}\left|\mathbb{E}P(\cG^\sharp\leq z|\cE)-\mathbb{E}P(\cT^\sharp\leq z|\cE)\right|  \nonumber \\
    &\leq \mathbb{E}\sup_{z\in \mathbb{R}}\left|P(\cG^\sharp\leq z|\cE)-P(\cT^\sharp\leq z|\cE)\right|  \nonumber \\
    &\lesssim \left(\frac{\kappa_1^3 \log^5 n}{np}\right)^{1/4}+\frac{1}{n^{10}}\lesssim \left(\frac{\kappa_1^3 \log^5 n}{np}\right)^{1/4}.
\end{align*}
Combine this with Lemma \ref{Bootstraplemma5} and Lemma \ref{Bootstraplemma6} we get
\begin{align*}
    &\sup_{z\in \mathbb{R}}\left|P(\cG\leq z)-P(\cT\leq z)\right| \\
    \lesssim &\left(\frac{\kappa_1^3\log^5 n}{np}\right)^{1/4} + \frac{\kappa_1^{3.5}(d+1)\log n}{\sqrt{np}}\left(\kappa_1^3\sqrt{\frac{\log n}{L}}+\kappa_1^4\sqrt{\frac{\log n}{L(d+1)}}+\sqrt{k+d}\right).
\end{align*}
Since $P(\cG\leq c_{1-\alpha}) = 1-\alpha$, we know that 
\begin{align*}
    \left|P(\cT>c_{1-\alpha})-\alpha\right|\lesssim \left(\frac{\kappa_1^3\log^5 n}{np}\right)^{1/4} + \frac{\kappa_1^{3.5}(d+1)\log n}{\sqrt{np}}\left(\kappa_1^3\sqrt{\frac{\log n}{L}}+\kappa_1^4\sqrt{\frac{\log n}{L(d+1)}}+\sqrt{k+d}\right).
\end{align*}
\end{proof}

We define 
\begin{align*}
    \cT_2^\sharp &= \max_{m\in \cM}\max_{k\neq m}\left|\frac{(\tz_m-\tz_k)^\top (\nabla^2\cL(\tb^*))^\diamond \nabla\cL(\tb^*) }{\sigma_{m,k}}\right|, \\
    \cG_2^\sharp &= \max_{m\in \cM}\max_{k\neq m}\Bigg|\sum_{l=1}^L \sum_{(i,j)\in \cE,i>j} \frac{(\tz_m-\tz_k)^\top(\nabla^2\cL(\tb^*))^\diamond (\tx_i-\tx_j)}{\sigma_{m,k}L}(\phi(\tx_i^\top\tb^*-\tx_j^\top\tb^*)-y_{j, i}^{(l)})\omega_{j, i}^{(l)} \Bigg|,  \\
    \cT_3^\sharp &= \max_{m\in \cM}\max_{k\neq m}\frac{(\tz_m-\tz_k)^\top (\nabla^2\cL(\tb^*))^\diamond \nabla\cL(\tb^*) }{\sigma_{m,k}}, \\
    \cG_3^\sharp &= \max_{m\in \cM}\max_{k\neq m}\sum_{l=1}^L \sum_{(i,j)\in \cE,i>j} \frac{(\tz_m-\tz_k)^\top(\nabla^2\cL(\tb^*))^\diamond (\tx_i-\tx_j)}{\sigma_{m,k}L}(\phi(\tx_i^\top\tb^*-\tx_j^\top\tb^*)-y_{j, i}^{(l)})\omega_{j, i}^{(l)},
\end{align*}
where $\sigma_{m,k}^2 = (\tz_m-\tz_k)^\top (\nabla^2\cL(\tb^*))^\diamond \nabla^2\cL(\tb^*) (\nabla^2\cL(\tb^*))^\diamond    (\tz_m-\tz_k) / L$. Given $i\in [n], m\in \cM, k\neq m$, we define $X_i = \sqrt{L}(\nabla\cL(\tb^*))_i / \sqrt{(\nabla^2\cL(\tb^*))_{i,i}}$ and $Y_{m,k} = (\tz_m-\tz_k)^\top (\nabla^2\cL(\tb^*))^\diamond \nabla\cL(\tb^*) /\sigma_{m,k}$. Let $\{W_{m,k}, m\in \cM, k\neq m\}$ be a set of random variables such that $\{W_{m,k}|\cE, m\in \cM, k\neq m\}$ is a set of joint Gaussian random variables and
\begin{align*}
   \textbf{cov}(W_{m_1,k_1}, W_{m_2,k_2}|\cE) = \textbf{cov} (Y_{m_1,k_1}, Y_{m_2,k_2}|\cE),  \quad \forall m_1,m_2\in \cM, k_1\neq m_1, k_2\neq m_2.
\end{align*}
We define 
\begin{align*}
    \cQ_2^\sharp = \max_{m\in \cM}\max_{k\neq m}|W_{m,k}|, \quad \cQ_3^\sharp = \max_{m\in \cM}\max_{k\neq m}W_{m,k}.
\end{align*}
In the following proof, we let $(\cT, \cG, \cT^\sharp, \cG^\sharp, \cQ^\sharp)$ be any one of the three pairs: $(\cT_t, \cG_t, \cT^\sharp_t, \cG^\sharp_t, \cQ^\sharp_t),t=2,3$. Given $\alpha\in (0,1)$, let $c_{1-\alpha}$ be the $(1-\alpha)$-th quantile of $\cG$. We are aiming at showing $|P(\cT>c_{1-\alpha})-\alpha|\to 0$.

We start with the following lemmas.

\begin{lemma}\label{diamondlemma}
    Under conditions of Theorem \ref{estimationthm}, we have
    \begin{align*}
        \left\|(\nabla^2\cL(\tb_R))^\diamond \nabla^2\cL(\tb_R) (\nabla^2\cL(\tb_R))^\diamond - (\nabla^2\cL(\tb^*))^\diamond \nabla^2\cL(\tb^*) (\nabla^2\cL(\tb^*))^\diamond\right\|\lesssim \kappa_1^5\sqrt{\frac{(d+1)\log n}{n^3p^3L}}.
    \end{align*}
\end{lemma}
\begin{proof}
By triangle inequality we have
\begin{align*}
    &\left\|(\nabla^2\cL(\tb_R))^\diamond \nabla^2\cL(\tb_R) (\nabla^2\cL(\tb_R))^\diamond - (\nabla^2\cL(\tb^*))^\diamond \nabla^2\cL(\tb^*) (\nabla^2\cL(\tb^*))^\diamond\right\|  \\
    \lesssim & \left\|(\nabla^2\cL(\tb_R))^\diamond \nabla^2\cL(\tb_R) (\nabla^2\cL(\tb_R))^\diamond - (\nabla^2\cL(\tb^*))^\diamond \nabla^2\cL(\tb_R) (\nabla^2\cL(\tb_R))^\diamond\right\| \\
    &+\left\|(\nabla^2\cL(\tb^*))^\diamond \nabla^2\cL(\tb_R) (\nabla^2\cL(\tb_R))^\diamond - (\nabla^2\cL(\tb^*))^\diamond \nabla^2\cL(\tb^*) (\nabla^2\cL(\tb_R))^\diamond\right\|  \\
    &+\left\|(\nabla^2\cL(\tb^*))^\diamond \nabla^2\cL(\tb^*) (\nabla^2\cL(\tb_R))^\diamond - (\nabla^2\cL(\tb^*))^\diamond \nabla^2\cL(\tb^*) (\nabla^2\cL(\tb^*))^\diamond\right\|  \\
    \lesssim & \kappa_1\left\|(\nabla^2\cL(\tb_R))^\diamond  - (\nabla^2\cL(\tb^*))^\diamond \right\| + \frac{\kappa_1^2}{n^2p^2}\left\|\nabla^2\cL(\tb_R)  - \nabla^2\cL(\tb^*)\right\| .
\end{align*}
On one hand, we know that 
\begin{align*}
    \left\|\nabla^2\cL(\tb_R)  - \nabla^2\cL(\tb^*)\right\|&\lesssim \max_{(i,j)\in \cE}|\phi'(\tx_i^\top\tb_R-\tx_j^\top\tb_R) - \phi'(\tx_i^\top\tb^*-\tx_j^\top\tb^*)|\left\|\bL_\cG\right\|  \\
    &\lesssim \kappa_1^2\sqrt{\frac{(d+1)np\log n}{L}}.
\end{align*}
On the other hand, by definition we have
\begin{align*}
    \left\|(\nabla^2\cL(\tb_R))^\diamond  - (\nabla^2\cL(\tb^*))^\diamond \right\|\lesssim \frac{\kappa_1^2}{n^2p^2}\kappa_1^2\sqrt{\frac{(d+1)np\log n}{L}}.
\end{align*}
Therefore, we know that 
\begin{align*}
    \left\|(\nabla^2\cL(\tb_R))^\diamond \nabla^2\cL(\tb_R) (\nabla^2\cL(\tb_R))^\diamond - (\nabla^2\cL(\tb^*))^\diamond \nabla^2\cL(\tb^*) (\nabla^2\cL(\tb^*))^\diamond\right\|\lesssim \kappa_1^5\sqrt{\frac{(d+1)\log n}{n^3p^3L}}.
\end{align*}
    
\end{proof}

\begin{lemma}\label{Bootstraplemma4}
Assume $k\geq 2$. Under the conditions of Theorem \ref{uqmainthm} and under the event $\cA_2$, we have
\begin{align*}
    &\sup_{z\in \mathbb{R}}\left|P(\cT^\sharp\leq z|\cE)-P(\cQ^\sharp\leq z|\cE)\right| \lesssim \left(\frac{\kappa_1^3\log^5 n}{np}\right)^{1/4},\\
    &\sup_{z\in \mathbb{R}}\left|P(\cT^\sharp\leq z|\cE)-P(\cG^\sharp \leq z|\cE)\right|\lesssim \left(\frac{\kappa_1^3\log^5 n}{np}\right)^{1/4}.
\end{align*}
Furthermore, we also have
\begin{align*}
    \sup_{z\in \mathbb{R}}\left|P(\cG^\sharp\leq z|\cE)-P(\cQ^\sharp\leq z|\cE)\right| \lesssim \left(\frac{\kappa_1^3\log^5 n}{np}\right)^{1/4}.
\end{align*}
\end{lemma}
\begin{proof}
The \cite[Condition E,M]{chernozhuokov2022improved} holds with $b_1 = b_2= 1$ and $B_n\asymp\sqrt{|\cE|L} \max_{m\in \cM}\max_{k\neq m}|(\tz_m-\tz_k)^\top(\nabla^2\cL(\tb^*))^\diamond (\tx_i-\tx_j)/(\sigma_{m,k}L)|$. On one hand, we know that
\begin{align*}
    (\tz_m-\tz_k)^\top\left(\nabla^2\cL(\tb^*)\right)^\diamond (\tx_i-\tx_j)\lesssim \frac{\kappa_1}{np} + \frac{\kappa_1}{np}\frac{d+1}{n}\lesssim \frac{\kappa_1}{np}.
\end{align*}
On the other hand, since $\|((\nabla^2\cL(\tb_R))^\diamond    (\tz_m-\tz_k))_{1:n}\|_0\leq 2$, by Lemma \ref{lb} we know that
\begin{align}
    \sigma_{m,k}^2 &= \frac{1}{L}(\tz_m-\tz_k)^\top \left(\nabla^2\cL(\tb^*)\right)^\diamond \nabla^2\cL(\tb^*) \left(\nabla^2\cL(\tb^*)\right)^\diamond    (\tz_m-\tz_k) \nonumber \\
    &\gtrsim \frac{np}{\kappa_1 L}\left\|\left(\nabla^2\cL(\tb^*)\right)^\diamond    (\tz_m-\tz_k)\right\|^2_2\geq \frac{np}{\kappa_1 L}\left\|\left(\nabla^2\cL(\tb^*)\right)^\diamond \right\|^2\left\| \tz_m-\tz_k\right\|^2_2  \nonumber \\
    &\gtrsim \frac{1}{\kappa_1 npL}\left\| \tz_m-\tz_k\right\|^2_2\gtrsim \frac{1}{\kappa_1 npL}.  \label{sigmamklb}
\end{align}
As a result, we know that 
\begin{align*}
    B_n\lesssim \frac{\kappa_1^{1.5}\sqrt{|\cE|}}{\sqrt{np}}.
\end{align*}
Therefore, by \cite[Theorem 2.1, Theorem 2.2]{chernozhuokov2022improved} we know that 
    \begin{align*}
       &\sup_{z\in \mathbb{R}}\left|P(\cT^\sharp\leq z|\cE)-P(\cQ^\sharp\leq z|\cE)\right|\lesssim \left(\frac{\log^5(nL|\cE|)}{L|\cE|}\cdot \frac{\kappa_1^3|\cE|}{np}\right)^{1/4}\lesssim \left(\frac{\kappa_1^3\log^5 n}{np}\right)^{1/4}, \\
       &\sup_{z\in \mathbb{R}}\left|P(\cT^\sharp\leq z|\cE)-P(\cG^\sharp \leq z|\cE)\right|\lesssim \left(\frac{\log^5(nL|\cE|)}{L|\cE|}\cdot \frac{\kappa_1^3|\cE|}{np}\right)^{1/4}\lesssim \left(\frac{\kappa_1^3\log^5 n}{np}\right)^{1/4}.
    \end{align*}

\end{proof}

\begin{lemma}\label{Bootstraplemma5}
Under the conditions of Theorem \ref{uqmainthm}, as long as $n\gtrsim (d+1)^2 k$, we have
\begin{align*}
    &\sup_{z\in \mathbb{R}}|P(\cT\leq z) - P(\cT^\sharp\leq z)| \\
    \lesssim &\left(\frac{\kappa_1^3\log^5 n}{np}\right)^{1/4} + \frac{\kappa_1^{3.5}(d+1)\log n}{\sqrt{np}}\left(\kappa_1^3\sqrt{\frac{\log n}{L}}+\kappa_1^4\sqrt{\frac{\log n}{L(d+1)}}+\sqrt{k+d}\right).
\end{align*}
\end{lemma}
\begin{proof}
By definition of $\cT$ and $\cT^\sharp$ we know that
\begin{align*}
    |\cT-\cT^\sharp|\leq\max_{m\in \cM}\max_{k\neq m}\left|\frac{\widehat{\theta}_k - \widehat{\theta}_m - (\theta_k^* - \theta_m^*)}{\widehat{\sigma}_{m,k}} - \frac{(\tz_m-\tz_k)^\top (\nabla^2\cL(\tb^*))^\diamond \nabla\cL(\tb^*) }{\sigma_{m,k}}\right|.
\end{align*}
We write
\begin{align}
&\left|\frac{\widehat{\theta}_k - \widehat{\theta}_m - (\theta_k^* - \theta_m^*)}{\widehat{\sigma}_{m,k}} - \frac{(\tz_m-\tz_k)^\top (\nabla^2\cL(\tb^*))^\diamond \nabla\cL(\tb^*) }{\sigma_{m,k}}\right| \nonumber \\
\leq & \left|\frac{\widehat{\theta}_k - \widehat{\theta}_m - (\theta_k^* - \theta_m^*)-(\tz_m-\tz_k)^\top (\nabla^2\cL(\tb^*))^\diamond \nabla\cL(\tb^*)}{\widehat{\sigma}_{m,k}} \right|  \nonumber  \\
&+\left|\frac{(\tz_m-\tz_k)^\top (\nabla^2\cL(\tb^*))^\diamond \nabla\cL(\tb^*) }{\sigma_{m,k}}\right|\left|1-\frac{\sigma_{m,k}}{\widehat{\sigma}_{m,k}}\right|.\label{TT2diffeq1}
\end{align}
According to \eqref{approxthmproof}, the first term on the right hand side of \eqref{TT2diffeq1} can be bounded as
\begin{align}
&\left|\frac{\widehat{\theta}_k - \widehat{\theta}_m - (\theta_k^* - \theta_m^*)-(\tz_m-\tz_k)^\top (\nabla^2\cL(\tb^*))^\diamond \nabla\cL(\tb^*)}{\widehat{\sigma}_{m,k}} \right|  \nonumber \\
\lesssim &  \sqrt{\kappa_1 npL}\frac{\kappa_1^3(d+1)}{np}\left(\frac{\kappa_1^3\log n}{L}+\sqrt{\frac{(k+d)\log n}{L}}\right) \nonumber \\
\lesssim & \frac{\kappa_1^{3.5}(d+1)}{\sqrt{np}}\left(\frac{\kappa_1^3\log n}{\sqrt{L}}+\sqrt{(k+d)\log n}\right)\label{TT2diffeq2}
\end{align}
with probability at least $1-O(n^{-6})$. When it comes to the second term, by Lemma \ref{Lilemma} we have
\begin{align}
    \left|\frac{(\tz_m-\tz_k)^\top (\nabla^2\cL(\tb^*))^\diamond \nabla\cL(\tb^*) }{\sigma_{m,k}}\right|\lesssim \frac{\kappa_1}{np}\sqrt{\frac{np\log n}{L}}\sqrt{\kappa_1 npL}\lesssim \sqrt{\kappa_1^3\log n}\label{TT2diffeq3}
\end{align}
with probability at least $1-O(n^{-10})$. And, by Lemma \ref{diamondlemma} we have
\begin{align}
    &\left|1-\frac{\sigma_{m,k}}{\widehat{\sigma}_{m,k}}\right|\leq \left|1-\frac{\sigma_{m,k}^2}{\widehat{\sigma}_{m,k}^2}\right| \lesssim \kappa_1^5\sqrt{\frac{(d+1)\log n}{n^3p^3L^3}}\kappa_1 npL\lesssim  \kappa_1^6\sqrt{\frac{(d+1)\log n}{npL}}.\label{TT2diffeq4}
\end{align}
with probability at least $1-O(n^{-10})$. Plugging \eqref{TT2diffeq2}, \eqref{TT2diffeq3} and \eqref{TT2diffeq4} in \eqref{TT2diffeq1} we get 
\begin{align*}
   & \left|\frac{\widehat{\theta}_k - \widehat{\theta}_m - (\theta_k^* - \theta_m^*)}{\widehat{\sigma}_{m,k}} - \frac{(\tz_m-\tz_k)^\top (\nabla^2\cL(\tb^*))^\diamond \nabla\cL(\tb^*) }{\sigma_{m,k}}\right| \\
    \lesssim &\frac{\kappa_1^{3.5}(d+1)}{\sqrt{np}}\left(\frac{\kappa_1^3\log n}{\sqrt{L}}+\frac{\kappa_1^4\log n}{\sqrt{L(d+1)}}+\sqrt{(k+d)\log n}\right)
\end{align*}
with probability at least $1-O(n^{-6})$. As a result, we know that
\begin{align}
    |\cT-\cT^\sharp|\lesssim \frac{\kappa_1^{3.5}(d+1)}{\sqrt{np}}\left(\frac{\kappa_1^3\log n}{\sqrt{L}}+\frac{\kappa_1^4\log n}{\sqrt{L(d+1)}}+\sqrt{(k+d)\log n}\right)\label{TT2diffeq5}
\end{align}
with probability at least $1-O(n^{-6})$. We let 
\begin{align*}
    \delta\asymp \frac{\kappa_1^{3.5}(d+1)}{\sqrt{np}}\left(\frac{\kappa_1^3\log n}{\sqrt{L}}+\frac{\kappa_1^4\log n}{\sqrt{L(d+1)}}+\sqrt{(k+d)\log n}\right).
\end{align*}
By Lemma \ref{Bootstraplemma4} we have
\begin{align}
    \sup_{z\in \mathbb{R}}\left|P(\cT^\sharp\leq z)-P(\cQ^\sharp\leq z)\right|  &= \sup_{z\in \mathbb{R}}\left|\mathbb{E}P(\cT^\sharp\leq z|\cE)-\mathbb{E}P(\cQ^\sharp\leq z|\cE)\right|  \nonumber \\
    &\leq \mathbb{E}\sup_{z\in \mathbb{R}}\left|P(\cT^\sharp\leq z|\cE)-P(\cQ^\sharp\leq z|\cE)\right|  \nonumber \\
    &\lesssim \left(\frac{\kappa_1^3\log^5 n}{np}\right)^{1/4}+\frac{1}{n^{10}}\lesssim \left(\frac{\kappa_1^3\log^5 n}{np}\right)^{1/4}.\label{TT2diffeq6}
\end{align}
Therefore, by \eqref{TT2diffeq5} and \eqref{TT2diffeq6} we can write
\begin{align}
    &\sup_{z\in \mathbb{R}}|P(\cT\leq z) - P(\cT^\sharp\leq z)|\leq P(|\cT-\cT^\sharp|>\delta) + \sup_{z\in \mathbb{R}}P(z < \cT^\sharp \leq z+\delta) \nonumber \\
    \lesssim & \frac{1}{n^6}+ \left(\frac{\kappa_1^3\log^5 n}{np}\right)^{1/4} + \sup_{z\in \mathbb{R}}P\left(z < \cQ^\sharp\leq z \leq z+\delta\right). \label{TT2diffeq7}
\end{align}
By \cite[Theorem 3]{chernozhukov2015comparison} the last term on the right hand side can be controlled as 
\begin{align*}
    \sup_{z\in \mathbb{R}}P\left(z < \cQ^\sharp\leq z \leq z+\delta\right) &= \sup_{z\in \mathbb{R}}\mathbb{E}P\left(z < \cQ^\sharp\leq z \leq z+\delta|\cE\right) \\
    &\lesssim \mathbb{E}\sup_{z\in \mathbb{R}}P\left(z < \cQ^\sharp\leq z \leq z+\delta|\cE\right) \\
    &\lesssim \sqrt{\log n}\delta.
\end{align*}
Plugging this in \eqref{TT2diffeq7} we get 
\begin{align*}
    &\sup_{z\in \mathbb{R}}|P(\cT\leq z) - P(\cT^\sharp\leq z)| \\
    \lesssim &\left(\frac{\kappa_1^3\log^5 n}{np}\right)^{1/4} + \frac{\kappa_1^{3.5}(d+1)\log n}{\sqrt{np}}\left(\kappa_1^3\sqrt{\frac{\log n}{L}}+\kappa_1^4\sqrt{\frac{\log n}{L(d+1)}}+\sqrt{k+d}\right).
\end{align*}
\end{proof}

\begin{lemma}\label{Bootstraplemma6}
Under the conditions of Theorem \ref{uqmainthm}, we have
\begin{align*}
     \sup_{z\in \mathbb{R}}|P(\cG\leq z) - P(\cG^\sharp\leq z)|\lesssim \left(\frac{\log^5 n}{np}\right)^{1/4}+\kappa_1^{7.5}\frac{\sqrt{d+1}\log^{1.5} n}{\sqrt{npL}}.
\end{align*}
\end{lemma}
\begin{proof}
By definition we have
\begin{align}
    |\cG-\cG^\sharp|\leq \max_{i\in [n]}\left|\sum_{(i,j)\in \cE}\sum_{l=1}^L\Delta_{j, i}^{(l)}\omega_{j,i}^{(l)}\right|, \label{Bootstraplemma6eq2}
\end{align}
where 
\begin{align*}
    \Delta_{j, i}^{(l)}:=& \frac{(\tz_m-\tz_k)^\top(\nabla^2\cL(\tb_R))^\diamond (\tx_i-\tx_j)}{\widehat{\sigma}_{m,k}L}(\phi(\tx_i^\top\tb_R-\tx_j^\top\tb_R)-y_{j, i}^{(l)})  \\
    & - \frac{(\tz_m-\tz_k)^\top(\nabla^2\cL(\tb^*))^\diamond (\tx_i-\tx_j)}{\sigma_{m,k}L}(\phi(\tx_i^\top\tb^*-\tx_j^\top\tb^*)-y_{j, i}^{(l)}).
\end{align*}
Define $\psi = (\tz_m-\tz_k)^\top(\nabla^2\cL(\tb^*))^\diamond (\tx_i-\tx_j)$ and $\widehat{\psi} = (\tz_m-\tz_k)^\top(\nabla^2\cL(\tb_R))^\diamond (\tx_i-\tx_j)$, $\Delta_{j,i}^{(l)}$ can be controlled as
\begin{align}
    |\Delta_{j, i}^{(l)}|\leq &\left|\frac{\psi}{\sigma_{m,k}L} - \frac{\widehat{\psi}}{\widehat{\sigma}_{m,k}L}\right|\left|\phi(\tx_i^\top\tb^*-\tx_j^\top\tb^*)-y_{j, i}^{(l)}\right|  +\frac{\widehat{\psi}}{\widehat{\sigma}_{m,k}L}\left|\phi(\tx_i^\top\tb_R-\tx_j^\top\tb_R) - \phi(\tx_i^\top\tb^*-\tx_j^\top\tb^*)\right|.\label{Bootstraplemma6eq1}
\end{align}
Since $\psi/\sigma_{m,k}-\widehat{\psi}/\widehat{\sigma}_{m,k} = (\psi - \widehat{\psi})/\widehat{\sigma}_{m,k}+ \psi/\sigma_{m,k}(1 - \sigma_{m,k}/\widehat{\sigma}_{m,k})$, we know that
\begin{align*}
    \left|\frac{\psi}{\sigma_{m,k}L} - \frac{\widehat{\psi}}{\widehat{\sigma}_{m,k}L}\right|&\lesssim \frac{|\psi - \widehat{\psi}|}{\sigma_{m,k}L}+\frac{\psi}{\sigma_{m, k}L} \left|1-\frac{\sigma_{m,k}}{\widehat{\sigma}_{m,k}}\right|.
\end{align*}
Combine this with \eqref{sigmamklb} and \eqref{TT2diffeq4} we have 
\begin{align*}
    \left|\frac{\psi}{\sigma_{m,k}L} - \frac{\widehat{\psi}}{\widehat{\sigma}_{m,k}L}\right|&\lesssim \left( \left\|(\nabla^2\cL(\tb_R))^\diamond  - (\nabla^2\cL(\tb^*))^\diamond \right\|+\left\| (\nabla^2\cL(\tb^*))^\diamond \right\| \left|1-\frac{\sigma_{m,k}}{\widehat{\sigma}_{m,k}}\right|\right)\sqrt{\frac{\kappa_1 np}{L}}   \\
    &\lesssim \left(\frac{\kappa_1^2}{n^2p^2}\kappa_1^2\sqrt{\frac{(d+1)np\log n}{L}}+ \frac{\kappa_1}{np}\kappa_1^6\sqrt{\frac{(d+1)\log n}{npL}}\right)\sqrt{\frac{\kappa_1 np}{L}}  \\
    &\lesssim \kappa_1^{7.5}\frac{\sqrt{(d+1)\log n}}{npL}.
\end{align*}
Plugging this in \eqref{Bootstraplemma6eq1}, we get
\begin{align*}
    |\Delta_{j, i}^{(l)}|&\lesssim \kappa_1^{7.5}\frac{\sqrt{(d+1)\log n}}{npL}\cdot 1 + \frac{\left\| (\nabla^2\cL(\tb_R))^\diamond \right\| }{\sigma_{m,k}L}\left|(\tx_i^\top\tb_R-\tx_j^\top\tb_R) - (\tx_i^\top\tb^*-\tx_j^\top\tb^*)\right| \\
    &\lesssim \kappa_1^{7.5}\frac{\sqrt{(d+1)\log n}}{npL} + \frac{\kappa_1}{np}\sqrt{\frac{\kappa_1 np}{L}}\kappa_1^2\sqrt{\frac{(d+1)\log n}{npL}}\lesssim \kappa_1^{7.5}\frac{\sqrt{(d+1)\log n}}{npL}
\end{align*}
for all $(i,j)\in \cE$ with probability at least $1-O(n^{-10})$. Plugging this in \eqref{Bootstraplemma6eq2}, we know that
\begin{align*}
     |\cG-\cG^\sharp|\lesssim \max_{i\in [n]}\sqrt{\sum_{(i,j)\in \cE}\sum_{l=1}^L\left(\Delta_{j,i}^{(l)}\right)^2\log n}\lesssim \kappa_1^{7.5}\frac{\sqrt{d+1}\log n}{\sqrt{npL}}
\end{align*}
with probability at least $1-O(n^{-10})$.

Next, we let 
\begin{align}
    \delta\asymp \kappa_1^{7.5}\frac{\sqrt{d+1}\log n}{\sqrt{npL}}.\label{BootstrapLemma6eq3}
\end{align}
By Lemma \ref{Bootstraplemma4} we have
\begin{align*}
    \sup_{z\in \mathbb{R}}\left|P(\cG^\sharp\leq z)-P(\cQ^\sharp\leq z)\right|  &= \sup_{z\in \mathbb{R}}\left|\mathbb{E}P(\cG^\sharp\leq z|\cE)-\mathbb{E}P(\cQ^\sharp\leq z|\cE)\right|  \nonumber \\
    &\leq \mathbb{E}\sup_{z\in \mathbb{R}}\left|P(\cG^\sharp\leq z|\cE)-P(\cQ^\sharp\leq z|\cE)\right|  \nonumber \\
    &\lesssim \left(\frac{\kappa_1^3\log^5 n}{np}\right)^{1/4}+\frac{1}{n^{10}}\lesssim \left(\frac{\kappa_1^3\log^5 n}{np}\right)^{1/4}.
\end{align*}
Then we write
\begin{align}
    \sup_{z\in \mathbb{R}}|P(\cG\leq z) - P(\cG^\sharp\leq z)|&\leq P(|\cG-\cG^\sharp|>\delta) + \sup_{z\in \mathbb{R}}P(z < \cG^\sharp \leq z+\delta)  \nonumber\\
    &\lesssim \frac{1}{n^{10}}+\left(\frac{\kappa_1^3\log^5 n}{np}\right)^{1/4}+\sup_{z\in \mathbb{R}}P\left(z < \cQ^\sharp \leq z+\delta\right).\label{BootstrapLemma6eq4}
\end{align}
By \cite[Theorem 3]{chernozhukov2015comparison} the last term on the right hand side can be controlled as 
\begin{align}
    \sup_{z\in \mathbb{R}}P\left(z < \cQ^\sharp\leq z +\delta\right) &= \sup_{z\in \mathbb{R}}\mathbb{E}P\left(z < \cQ^\sharp\leq z +\delta|\cE\right) \lesssim \mathbb{E}\sup_{z\in \mathbb{R}}P\left(z < \cQ^\sharp\leq z+\delta|\cE\right) \lesssim \sqrt{\log n}\delta.\label{BootstrapLemma6eq5}
\end{align}
Plugging \eqref{BootstrapLemma6eq3} and \eqref{BootstrapLemma6eq5} in \eqref{BootstrapLemma6eq4} we know that
\begin{align*}
     \sup_{z\in \mathbb{R}}|P(\cG_1\leq z) - P(\cG_1^\sharp\leq z)|\lesssim \left(\frac{\kappa_1^3\log^5 n}{np}\right)^{1/4}+\kappa_1^{7.5}\frac{\sqrt{d+1}\log^{1.5} n}{\sqrt{npL}}.
\end{align*}
\end{proof}

\end{document}